\documentclass[11pt,letterpaper]{article}
\usepackage[hmargin=0.85in,vmargin=1in]{geometry}
\usepackage[T1]{fontenc}
\usepackage[utf8]{inputenc}
\usepackage{textcomp,lmodern}
\usepackage{amsmath,amsfonts,amssymb,amsthm,mathrsfs,bm,bbm}
\usepackage{graphicx,xcolor,booktabs,multirow,array}
\usepackage{caption,subcaption,float,paralist,enumitem}
\usepackage{rotating,threeparttable,pdflscape}
\usepackage{microtype,parskip}
\usepackage[authoryear,round]{natbib}
\usepackage[colorlinks=true,linkcolor=blue,citecolor=blue,urlcolor=blue,
            pdfborder={0 0 0},hypertexnames=true]{hyperref}
\usepackage{bookmark}
\IfFileExists{xurl.sty}{\usepackage{xurl}}{}
\newcommand{\anon}{1}
\newcommand{\spacingset}[1]{}
\floatstyle{ruled}
\newfloat{algorithm}{htbp}{loa}
\floatname{algorithm}{Algorithm}
\floatstyle{plain}

\newcommand{\E}{\mathbb{E}}
\newcommand{\Prob}{\mathbb{P}}
\newcommand{\R}{\mathbb{R}}

\newcommand{\matS}{\mathcal{S}}
\newcommand{\matSpop}{\Sigma}
\newcommand{\tr}{\operatorname{tr}}
\newcommand{\diag}{\operatorname{diag}}

\newcommand{\OrderP}{\mathrm{O}_{\mathbb{P}}}
\newcommand{\orderP}{\mathrm{o}_{\mathbb{P}}}
\newcommand{\Var}{\operatorname{Var}}

\theoremstyle{plain}

\newtheorem{theorem}{Theorem}[section]
\newtheorem{corollary}[theorem]{Corollary}
\newtheorem{remark}[theorem]{Remark}
\newtheorem{definition}[theorem]{Definition}
\newtheorem{lemma}[theorem]{Lemma}
\newtheorem{assumption}[theorem]{Assumption}
\newtheorem{proposition}[theorem]{Proposition}
\theoremstyle{remark}
\newtheorem*{example}{Example}

\theoremstyle{plain}
\newtheorem{suppTheorem}{Theorem}[section]
\newtheorem{suppAssumption}{Assumption}[section]
\newtheorem{suppLemma}{Lemma}[section]
\newtheorem{suppRemark}{Remark}[section]
\newtheorem{suppCorollary}{Corollary}[section]
\newtheorem{suppProposition}{Proposition}[section]
\newtheorem{suppExample}{Example}[section]
\newtheorem{suppDefinition}{Definition}[section]

\title{\bfseries The Spurious Factor Dilemma: Robust Inference in Heavy-Tailed Elliptical Factor Models}
\author{%
Jiang Hu\thanks{%
Jiang Hu was partially supported by National Natural Science Foundation of China (Grant Nos. 12292980, 12292982, 12171078). 
Yangchun Zhang was partially supported by National Natural Science Foundation of China (Grant No.  12301351, 124712361). 
Wang Zhou is partially supported by a grant A-8001450-00-00 at the National University of Singapore.
}\hspace{.2cm}\\
Northeast Normal University, China\\
and\\
Jiahui Xie\\
National University of Singapore, Singapore\\
and\\
Yangchun Zhang\\
Shanghai University, China\\
and\\
Wang Zhou\\
National University of Singapore, Singapore
}
\date{}
\hypersetup{pdftitle={The Spurious Factor Dilemma: Robust Inference in Heavy-Tailed Elliptical Factor Models}}
\ifnum\anon=1
  \hypersetup{pdfauthor={Jiang Hu; Jiahui Xie; Yangchun Zhang; Wang Zhou}}
\else
  \hypersetup{pdfauthor={}}
\fi

\begin{document}
\spacingset{1}
\ifnum\anon=1
  \maketitle
\else
  \bigskip\bigskip\bigskip
  \begin{center}{\LARGE\bfseries The Spurious Factor Dilemma: Robust Inference in Heavy-Tailed Elliptical Factor Models}\end{center}
  \medskip
\fi
\ifnum\anon=1\fi
\bigskip
\begin{abstract}
Standard methods for determining the number of factors often overestimate the true number when data exhibit heavy-tailed randomness, misinterpreting noise-induced outliers as genuine factors. This paper addresses this challenge within the framework of Elliptical Factor Models (EFM), which accommodate both heavy tails and potential non-linear dependencies common in real-world data. We demonstrate, both theoretically and empirically, that heavy-tailed noise generates spurious eigenvalues that mimic true factor signals. To distinguish them, we propose a novel methodology based on a fluctuation magnification algorithm. We show that an eigenvalue generated by a real factor aggregates the contributions of many observations, so that the multipliers average out and its relative fluctuation vanishes asymptotically, whereas a spurious eigenvalue is driven by a single extreme observation and inherits the fluctuation of its multiplier. We develop a formal testing procedure based on this principle and apply it to the problem of accurately selecting the number of common factors in heavy-tailed EFMs. Simulation studies and real data analysis confirm the effectiveness of our approach, particularly in scenarios with pronounced heavy-tailedness.\end{abstract}

\clearpage
\begingroup
\small
\setcounter{tocdepth}{2}
\tableofcontents
\endgroup
\clearpage

\section{Introduction}\label{sec_intro}

Factor models serve as a cornerstone in the analysis of large-scale datasets across various disciplines, including economics, finance, genetics, and signal processing. Their power lies in the ability to parsimoniously capture complex dependencies and interactions among numerous variables by attributing them to a small number of latent common factors \citep{bai2002determining, stock2002forecasting, fan2018large}. By effectively modeling these latent structures, factor analysis provides crucial tools for dimension reduction, feature extraction, and understanding the underlying drivers of observed phenomena \citep{fan2021recent}.

A fundamental challenge in applying factor models is the determination of the correct number of common factors, denoted by $m$. This problem has received considerable attention, yet remains a subject of ongoing research due to its critical impact on subsequent analysis. Underestimating $m$ leads to the omission of significant systematic components, potentially resulting in biased estimates of factor loadings, inconsistent forecasting, and flawed structural interpretations \citep{bai2003inferential, baltagi2017identification}. Conversely, overestimating $m$ introduces noise by fitting spurious factors, which can inflate estimation variance, reduce model interpretability, and increase computational costs \citep{barigozzi2020consistent}.

Existing methodologies for estimating $m$ largely fall into two categories. The first, prevalent in econometrics, leverages the connection between factor analysis and Principal Component Analysis (PCA). Methods like the information criteria of \citep{bai2002determining} and \citep{alessi2010improved}, the eigenvalue ratio tests of \citep{ahn2013eigenvalue} and \citep{lam2012factor}, and the randomization tests \citep{trapani2018randomized, kong2020random} rely on the assumption that eigenvalues associated with common factors diverge at a faster rate than those corresponding to idiosyncratic noise as dimensions grow. The second category employs Random Matrix Theory (RMT) to provide finer distinctions, particularly in high-dimensional settings where both the number of variables ($p$) and observations ($n$) are large. RMT-based tests, such as those by \citep{onatski2009testing}, utilize the fact that the largest noise eigenvalues converge to the Tracy--Widom distribution, while factor-related eigenvalues appear as distinct outliers, often asymptotically Gaussian after proper scaling \citep{onatski2010determining, cai2020limiting, ke2023estimation}.

However, both classes of methods face limitations when faced with data exhibiting heavy-tailed randomness. Heavy-tailed distributions, characterized by a higher probability of extreme events compared to Gaussian distributions, are ubiquitous in financial returns, climate data, and various other fields \citep{roy2021empirical,ke2023estimation}. PCA-based methods, while relatively robust to certain types of noise dependence, can falter because heavy tails can generate large sample eigenvalues purely from noise, mimicking the signature of true factors. RMT-based methods typically rely on moment conditions or concentration properties that are violated by heavy-tailed noise, leading them to misinterpret large noise eigenvalues as signals. 
Figure \ref{fig_illustration_main} provides a visual example of this issue, where a large spurious eigenvalue appears far from the bulk, potentially misleading standard selection criteria.

Furthermore, empirical data often exhibit non-linear dependencies alongside heavy tails. Standard factor models typically assume idiosyncratic errors or linear dependencies, potentially missing complex interaction patterns. Elliptical Factor Models (EFM), based on elliptical distributions, offer a flexible framework that naturally incorporates both heavy-tailedness (via the radial component) and non-linear dependencies (via the elliptical structure) \citep{chamberlain1982arbitrage,baltagi2017identification}. Recently, \cite{bao2025signal} has demonstrated that even heteroscedastic or cross-correlated noise with independent entries along the time dimension can produce misleading spikes, which traditional singular value methods may incorrectly interpret as signals. As a result, spurious eigenvalues are quite common in practice.

This paper addresses the critical issue of factor number overestimation in the EFM framework, specifically focusing on the confusion caused by heavy-tailed noise. We pose two central questions: \textit{Can we reliably detect spurious factors generated by heavy tails in elliptical models? Can we develop a robust factor selection procedure for such data?}

Our primary contribution is a novel methodology with a rigorous theoretical basis for distinguishing between ``real" factor signals and ``spurious" noise-induced signals among the large sample eigenvalues. We introduce a novel algorithm called the fluctuation magnification algorithm, specifically designed for sample covariance matrices derived from EFM data. The core idea is that reweighting the observations by independent random multipliers affects real and spurious signals differently. An eigenvalue generated by a real factor aggregates the contributions of many observations, so the multipliers average out and its relative fluctuation vanishes asymptotically. A spurious eigenvalue generated by heavy tails is instead driven by a single observation with an extreme radius, and therefore inherits the full fluctuation of its multiplier. After normalization by the known variance of the multiplier, the two cases lead to the limits zero and one, which provides a clear mechanism for detection.

Building on the eigenvalue asymptotics of Section \ref{sec_mainresults}, we show in Section \ref{sec_testing_procedures} that these behaviors persist under magnification, and we develop a test statistic that quantifies the fluctuation of each large sample eigenvalue under the magnification algorithm. This allows us to formally test for the presence of spurious factors (Section \ref{sec_testing_procedures}). As a key application, we integrate this detection mechanism into a two-step procedure to robustly estimate the number of common factors ($m$) in heavy-tailed EFMs. The first step uses the magnification algorithm to identify potential spurious signals among the leading eigenvalues. The second step employs existing criteria (e.g., from \cite{onatski2010determining}) as a safeguard, since genuine factors are never flagged, the procedure does not underestimate whenever the initial estimator does not, and it inherits the consistency of the initial estimator when no spurious eigenvalue is present, as in the light-tailed case.

Our work can be viewed as being generated from high-dimensional resampling techniques \citep{han2018gaussian,lopes2019bootstrapping, ke2023estimation, yao2021rates, yu2024testing,  xie2026multiplier} to the challenging setting of elliptical factor models with heavy tails, providing the first procedure, to our knowledge, specifically designed to detect spurious factors in this context. Numerical simulations demonstrate the superior performance of our method compared to established techniques, especially in heavy-tailed scenarios. Application to real financial data yields results consistent with financial theory and highlights the practical relevance of addressing spurious factors.

The remainder of the paper is organized as follows. Section \ref{sec_model} formally introduces the Elliptical Factor Model and outlines the key assumptions. Section \ref{sec_motivation} further illustrates the problem of spurious factors using examples. Section \ref{sec_mainresults} presents the main asymptotic theory that details the behavior of real and spurious eigenvalues. Section \ref{sec_testing_procedures} describes the fluctuation magnifier algorithm and the proposed testing and factor selection procedures. Section \ref{sec_simulation} provides simulation results, and Section \ref{sec_realdata} discusses the real data application. 
The conclusion is offered in Section \ref{sec_conclusion}. All detailed technical proofs are deferred to the Supplementary Materials \cite{SupplementaryMaterials}.

\subsection*{Conventions}
Let $\mathbb{C}_+$ denote the complex upper half-plane. We use $C > 0$ to represent a generic positive constant whose value may change from line to line. For sequences of positive deterministic values $\{a_n\}$ and $\{b_n\}$, $a_n = \mathrm{O}(b_n)$ means $a_n \leq C b_n$ for some $C > 0$. If $a_n = \mathrm{O}(b_n)$ and $b_n = \mathrm{O}(a_n)$, we write $a_n \asymp b_n$. We write $a_n = \mathrm{o}(b_n)$ if $a_n \leq c_n b_n$ for some positive sequence $c_n \downarrow 0$. For a sequence of random variables $\{x_n\}$ and positive real values $\{a_n\}$, $x_n = \OrderP(a_n)$ indicates that $x_n / a_n$ is stochastically bounded. $x_n = \orderP(a_n)$ means $x_n / a_n$ converges to zero in probability. For a sequence of positive random variables $\{y_n\}$, $y_{(k)}$ denotes the $k$-th order statistic, $y_{(1)} \geq y_{(2)} \geq \cdots \geq y_{(n)} > 0$. Vectors are marked in bold.

\section{Elliptical factor model and assumptions}\label{sec_model}

Elliptical distributions provide a versatile class for modeling multivariate data, extending the normal distribution to accommodate heavy tails and capture specific dependence structures like tail dependence, making them particularly relevant in finance and other fields \citep{chamberlain1982arbitrage, fama1993common, baltagi2017identification}. A $p$-dimensional random vector $\mathbf{y}$ follows a centered elliptical distribution, denoted $\mathbf{y} \sim EC_p(0, \matSpop, \xi)$, if it admits the stochastic representation:
\begin{equation}\label{eq_intro_elldistr}
    \mathbf{y} \overset{d}{=}\xi \matSpop^{1/2} \mathbf{u},
\end{equation}
where 
$\matSpop \in \R^{p \times p}$ is a positive definite matrix representing the population covariance matrix, $\xi \ge 0$ is a non-negative scalar random variable representing the ``radius," and $\mathbf{u} \in \R^p$ is a random vector uniformly distributed on the unit sphere $\mathbb{S}^{p-1}$, independent of $\xi$. The variable $\xi$ governs the tail behavior of the distribution.

We integrate this structure with the standard linear factor model.  Let $\mathbf{y}_1, \dots, \mathbf{y}_n$ be independent and identically distributed (i.i.d.) random vectors in $\mathbb{R}^p$. The factor model posits:
\begin{equation}\label{eq_intro_factormodel}
    \mathbf{y}_t = B \mathbf{f}_t + \mathbf{e}_t, \quad t=1, \dots, n,
\end{equation}
where $B \in \R^{p \times m}$ is the deterministic factor loading matrix, $\mathbf{f}_t \in \R^m$ is the vector of $m$ latent common factors, and $\mathbf{e}_t \in \R^p$ is the vector of idiosyncratic errors. We assume $m$ is fixed and much smaller than $p$ and $n$. 
Standard identification conditions often include $\frac{1}{n} \sum_{t=1}^n \mathbf{f}_t \mathbf{f}_t' \overset{\mathbb{P}}{\rightarrow} I_m$ as $p \to \infty$ \citep{fan2018large}.

To define the Elliptical Factor Model (EFM),  we assume that $\mathbf{f}_t$ and $\mathbf{e}_t$ are uncorrelated and the joint distribution of factors and errors follows an elliptical structure.  Specifically, if $(\mathbf{f}_t', \mathbf{e}_t')'$ is elliptically distributed with mean zero and covariance matrix $\diag(I_m, \matSpop_{err})$, where $\operatorname{diag}(I_m, \matSpop_{err})$ means that 
\begin{align*}
    \operatorname{diag}(I_m, \matSpop_{err})=
    \begin{pmatrix}
        I_m &0\\
        0 &\matSpop_{err}
    \end{pmatrix}.
\end{align*}
Then $(\mathbf{f}_t', \mathbf{e}_t')' \sim EC_{m+p}(\mathbf{0}, \diag(I_m, \matSpop_{err}), \xi_t)$. Consequently, the observation vector $\mathbf{y}_t$ also follows an elliptical distribution:
\[ \mathbf{y}_t=
\begin{pmatrix}
B ~~~I
\end{pmatrix}
\begin{pmatrix}
    \mathbf{f}_t\\
    \mathbf{e}_t
\end{pmatrix}
\sim EC_p(\mathbf{0}, \matSpop, \xi_t), \]
where 
\begin{align*}
    \Sigma=BB'+\matSpop_{err}\end{align*}
is the population covariance matrix of $\mathbf{y}_t$ (assuming $\E[\xi_t^2]$ is finite and normalized appropriately). The stochastic representation for the observations becomes:
\begin{equation}\label{eq_def_ellfactmodel}
    \mathbf{y}_t = \sqrt{p} \xi_t \matSpop^{1/2} \mathbf{u}_t, \quad t=1, \dots, n,
\end{equation}
where $\{\xi_t\}_{t=1}^n$ are i.i.d. copies of the radius variable $\xi$, and $\{\mathbf{u}_t\}_{t=1}^n$ are i.i.d. uniform on $\mathbb{S}^{p-1}$, independent of $\xi$. Here, the scaling factor $\sqrt{p}$ ensures that $\mathbf{y}_t$ has the right covariance order $\Sigma$. The data matrix is $Y = (\mathbf{y}_1, \dots, \mathbf{y}_n)= \sqrt{p}\matSpop^{1/2} U \diag(\xi_1, \dots, \xi_n)$, where $U = (\mathbf{u}_1, \dots, \mathbf{u}_n)$.

The key feature allowing for heavy tails is the random radius $\xi_t$. We impose the following assumption on its distribution.

\begin{assumption}[Distribution of Radius]\label{ass_xi}
     The squared radii $\xi_t^2 \sim \xi^2$ for $1 \leq t \leq n$ are i.i.d. copies of a non-negative, non-degenerate random variable $\xi^2$ with support on $\R^{+}$. We assume $\E[\xi^2] = 1$ without loss of generality and $\xi^2$ satisfies either of the following tail conditions:
    
(i)  \textit{Polynomial decay tail:}
            \begin{equation}\label{ass_xi_poly}
                \Prob(\xi^2 > x) = x^{-\alpha} L(x), \quad \text{as } x \to \infty,
            \end{equation}
            for some $\alpha \in (1, \infty)$, where $L(x)$ is a slowly varying function, i.e., $$\lim_{x \to \infty} L(sx)/L(x) = 1$$ for all $s>0$.
            
(ii) \textit{Exponential decay tail:} For some constant $\beta > 0$ and any fixed constant $s > 0$,
            \begin{equation}\label{ass_xi_exp}
                \E e^{s\xi^{2\beta}} < \infty.
            \end{equation}

\end{assumption}

\begin{remark}[Heavy Tails]
Assumption \ref{ass_xi} explicitly allows for heavy tails.
The polynomial decay case \eqref{ass_xi_poly} includes distributions where $\xi^2$ has finite variance but potentially infinite higher moments (e.g., multivariate t-distributions with degrees of freedom $\nu > 2$; here $\alpha = \nu/2$). The case $\alpha \in (1, 2]$ corresponds to particularly heavy tails where the fourth moment of $\xi^2$ (and potentially $\mathbf{y}_t$) may not exist. We term this ``serious heavy-tailed randomness." The case $\alpha \in (2, \infty)$ or the exponential decay case \eqref{ass_xi_exp} covers distributions with finite fourth moments but still heavier tails than Gaussian, termed ``typical heavy-tailed randomness." This contrasts with settings assuming $\xi$ is constant or bounded \citep{hu2019high}. The threshold
\begin{equation*}\label{eq_def_thresholdT}
    \mathsf{T} :=
    \begin{cases}
        n ^{1/\alpha}\log n, & \text{if \eqref{ass_xi_poly} holds}; \\
        (\log n)^{1/\beta}, & \text{if \eqref{ass_xi_exp} holds}.
    \end{cases}
\end{equation*}
characterizes the typical maximal order of $\xi_t^2$. As shown later (Lemma \ref{lem_goodconfiguration}), $\max_{1\le t \le n} \xi_t^2 \lesssim \mathsf{T}$ with high probability. For simplicity, we assume that the slowly varying function $L$ in \eqref{ass_xi_poly} is bounded away from zero and infinity. 
\end{remark}

We assume the population covariance matrix $\matSpop$ exhibits a spiked structure, which is common in factor analysis \citep{chamberlain1982arbitrage, yu2024testing}.

\begin{assumption}[Spiked Covariance Structure]\label{ass_sigma}
    The population covariance matrix $\matSpop = BB' + \matSpop_{err}$ satisfies:
    

(1) $\max\{\|\Sigma_{err}\|, \|\Sigma_{err}^{-1}\|\}\le c_1$, for some constant $0 < c_1< \infty$.  

(2) $BB'$ is of rank $m$ for some fixed integer $m>0$. The descending eigenvalues of $BB'$, 
        $(\lambda_i(BB^{\prime}))^{-1}\mathsf{T}=\mathrm{o}(1)$ for $1\le i\le m$, and $1+c_2\le\lambda_i(BB^{\prime})/\lambda_{i+1}(BB^{\prime})\le c_3 $ for any $1\le i\le m-1$ and some $0 < c_2<c_3< \infty$.
        
(3)  The dimensions satisfy $\phi_n:=p/n \to \phi \in (c_4, c_4^{-1})$ for some constant $0 < c_4< \infty$, as $\min\{p, n\} \to \infty$. 

\end{assumption}

\begin{remark}\label{rem_sigma_diagonal} 
In Assumption \ref{ass_sigma},  (1) is natural, as it imposes that the eigenvalues of the error covariance matrix are uniformly bounded from above and below.  (2) indicates that  $\Sigma$ has a spiked structure. Since $\mathbf u_t$ is uniformly distributed on the sphere, $O'U\overset{d}{=}U$ for every orthogonal $O$. Writing $\Sigma=O\,\diag(\sigma_1,\dots,\sigma_p)\,O'$, the eigenvalues of $n^{-1}YY'$ therefore have the same joint distribution as when $\Sigma$ is replaced by $\diag(\sigma_1,\dots,\sigma_p)$. Hence we will assume without loss of generality that $\Sigma=\diag(\sigma_1,\dots,\sigma_p)$ is diagonal, where $\sigma_1\ge\dots\ge\sigma_p$ are the eigenvalues of $\Sigma$. By Weyl's inequality and (1), $|\sigma_i-\lambda_i(BB')|\le c_1$, so that $\mathsf{T}/\sigma_i=\mathrm{o}(1)$ for $1\le i\le m$ and $c_1^{-1}\le\sigma_i\le c_1$ for $i\ge m+1$. That is, $\sigma_1>\sigma_2>\dots>\sigma_m>\sigma_{m+1}\ge \dots\ge \sigma_p\ge c_1^{-1}$. Also, notice that we demand the large spike of $\Sigma$ should exhibit the same order if they are divergent. $(3)$ confirms the structure of high dimensions, and the fixed $m$ is common in the literature, especially when the target is to determine the number of common factors \citep{fan2018large,yu2024testing}.
\end{remark}
\begin{remark}We want to emphasize that
    although Assumption \ref{ass_sigma} imposes the restriction $(\Sigma_{ii})^{-1}\mathsf{T} = \mathrm{o}(1)$ for $1 \le i \le m$, it is generally difficult to distinguish real signals from spurious ones based solely on their asymptotic divergence rates from the observed data, since the presence of heavy-tailed randomness in the factor noise cannot be ruled out.
\end{remark}

\begin{remark}
    The main difference between Assumption \ref{ass_sigma} and typical settings in factor analysis is that we specify only the configuration of the common factors in $\Sigma$, without imposing structural constraints (such as diagonality or sparsity) on $\Sigma_{err}$ beyond the spectral bound in (1). Additionally, from Assumption \ref{ass_xi}, it is evident that we allow for a broad range of heavy-tailed randomness in factor noise.
\end{remark}

\section{Spurious factors from heavy tails}\label{sec_motivation}

Given observations $\{\mathbf{y}_t\}_{t=1}^n$ from the EFM \eqref{eq_def_ellfactmodel}, we form the sample covariance matrix $S$ and its companion $\matS$:
\begin{equation}\label{eq_def_samplecovariancematrices}
    S :=\frac{1}{n}  YY' = \matSpop^{1/2} U D^2 U' \matSpop^{1/2}, \quad \matS :=\frac{1}{n}  Y'Y =  D U' \matSpop U D,
\end{equation}
where $D = \phi_n^{1/2}\diag(\xi_1, \dots, \xi_n)$. Note that $S$ and $\matS$ share the same non-zero eigenvalues, conventionally denoted $\lambda_1 \ge \lambda_2 \ge \dots \ge \lambda_{\min(p,n)}$. Standard factor analysis methods aim to estimate $m$ by examining these top eigenvalues. The underlying assumption is that $\lambda_1, \dots, \lambda_m$ are primarily influenced by the population spikes $\sigma_1, \dots, \sigma_m$, while $\lambda_{m+1}, \dots$ reflect the noise structure $\matSpop_{err}$.

However, under the EFM with heavy tails (Assumption \ref{ass_xi}), this distinction becomes blurred. The random radii $\xi_t^2$ can take extremely large values. When a large $\xi_t^2$ aligns with certain directions in $U$ and $\matSpop^{1/2}$, it can inflate some sample eigenvalues dramatically, even those notionally corresponding to the noise part $\matSpop_{err}$. These inflated noise eigenvalues can become outliers, exceeding the bulk and potentially mixing with, or even surpassing, the eigenvalues generated by the true factors $B\mathbf{f}_t$. We term these large eigenvalues not directly associated with the $m$ population spikes as ``spurious factors".

Figure~\ref{fig_illustration_main} provides a toy example. Standard methods relying on eigenvalue ratios or gap thresholding (e.g., \cite{ahn2013eigenvalue, onatski2010determining}) identify three factors ($\widehat m=3$) instead of the true $m=2$, mistaking the eigenvalue at approximately $19.72$ for a genuine factor signal. This overestimation poses a significant problem in practice. Therefore, the primary goal of this paper is to develop a methodology to detect these spurious factors. Let $\mathsf{f}$ denote the number of spurious factors among the top eigenvalues (i.e., the number of large eigenvalues $\lambda_k$ that do not correspond to the true population spikes $\sigma_1, \dots, \sigma_m$). We aim to test the hypothesis:
\begin{equation}\label{eq_hypothesistest}
    \mathbf{H}_0: \mathsf{f} = 0 \quad \text{vs.} \quad \mathbf{H}_a: \mathsf{f} > 0.
\end{equation}
Rejecting $\mathbf{H}_0$ implies the presence of spurious factors among the leading eigenvalues. This detection is crucial for accurately estimating the true number of factors $m$. Since the radius has unbounded support under Assumption \ref{ass_xi}, the largest radius $\xi^2_{(1)}$ diverges and $\lambda_{m+1}$ is generated by it with probability tending to one by Section \ref{sec_main_fakesignals}. Thus $\mathbf H_a$ holds asymptotically under Assumption \ref{ass_xi}, and the substantive problem is to locate the boundary between genuine and spurious eigenvalues. Accordingly, our procedure tests, for each index $i$, the hypothesis $\mathbf H_{0,i}$: $\lambda_i$ is a genuine factor eigenvalue ($i\le m$), against $\mathbf H_{1,i}$: $\lambda_i$ is spurious, and estimates $m+1$ by the first rejected index. Our proposed approach, detailed in Section \ref{sec_testing_procedures}, leverages the differential fluctuation of real and spurious signals under magnification  perturbations.
\begin{figure}[htbp]
    \centering
    \includegraphics[width=0.6\textwidth]{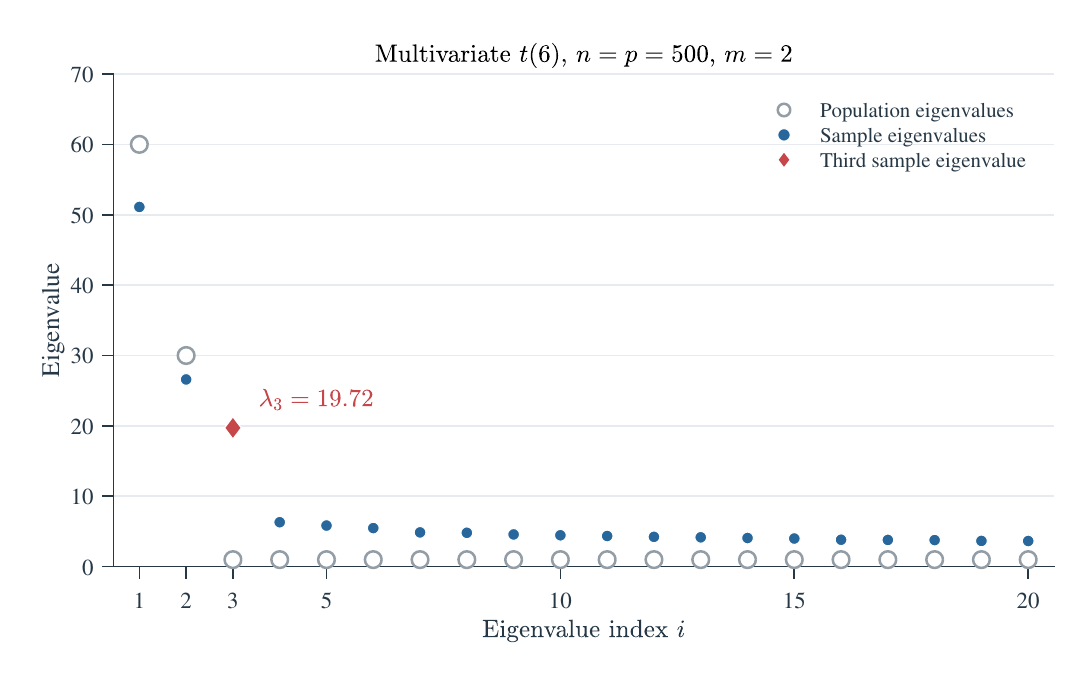} %
    \caption{Illustration of sample eigenvalues from an EFM with $p=n=500$. The data follow a standardized multivariate $t$-distribution with 6 degrees of freedom ($\alpha=3$). The population covariance is $\matSpop=\diag(60,30,1,\dots,1)$ ($m=2$, $\sigma_1=60$, $\sigma_2=30$). Besides the two sample eigenvalues near $51.10$ and $26.60$ corresponding to genuine factors, a spurious eigenvalue appears around $19.72$, well-separated from the main bulk of sample eigenvalues.}
    \label{fig_illustration_main}
\end{figure}
\section{Asymptotic behavior of sample eigenvalues}\label{sec_mainresults}

This section provides the theoretical foundation for our detection method by characterizing the asymptotic behavior of both real factor signals and spurious noise components present in the sample eigenvalues $\{\lambda_k\}$. Due to the spiked structure of $\Sigma$, we decompose $\Sigma=\Sigma_1+\Sigma_2$, where $\Sigma_1$ contains the $m$ largest components of $\Sigma$, and $\Sigma_2$ captures the remaining ones (also recall that $\Sigma$ is diagonal from Remark \ref{rem_sigma_diagonal}). As a consequence, we can rewrite 
\begin{equation}\label{eq_decomposition_S}
    \mathcal{S}=DU^{\prime}\Sigma_1UD+DU^{\prime}\Sigma_2UD.
\end{equation}
In the sequel, we focus on the eigenvalues of $\matS =\frac{1}{n} Y'Y$. Note that  $S$ and $\matS$ share the same non-zero eigenvalues.

\subsection{Asymptotic behavior of real factors}\label{sec_main_realsignals}

The top $m$ eigenvalues $\lambda_1, \dots, \lambda_m$ are primarily influenced by the population spikes $\sigma_1, \dots, \sigma_m$. However, their exact location is perturbed by the noise component $\matSpop_{err}$ and the heavy-tailed radii $D^2$.

Following standard techniques for spiked models, we first establish a first-order approximation.
\begin{theorem}[First Order Approximation]\label{lem_realsignal_preratio}
    Under Assumptions \ref{ass_xi} and \ref{ass_sigma},  we have:
    \begin{itemize}
        \item[(1)] for polynomial decay tail \eqref{ass_xi_poly} with $\alpha\in(2,+\infty)$ or exponential decay tail \eqref{ass_xi_exp} in Assumption \ref{ass_xi}, 
        \begin{equation*}
        \frac{\lambda_i}{\sigma_i} = 1 + \OrderP\left(\frac{\mathsf{T}}{\sigma_i} + \frac{1}{\sqrt{n}}\right), ~~~1 \le i \le m;
    \end{equation*}
        \item[(2)] for polynomial decay tail  \eqref{ass_xi_poly} with $\alpha\in(1,2]$ in Assumption \ref{ass_xi},
         \begin{equation*}
        \frac{\lambda_i}{\sigma_i} = 1 + \OrderP\left(\frac{\mathsf{T}}{\sigma_i} + \sqrt{\frac{\mathsf{T}}{n}}\right),~~~1 \le i \le m.
    \end{equation*}
    \end{itemize}
\end{theorem}

\begin{remark}
    Theorem \ref{lem_realsignal_preratio} indicates that the leading eigenvalues of $S$ are very close to the population ones in the sense that their ratios asymptotically converge to one in probability. However, this doesn't mean that $\lambda_i$ is a consistent estimator of $\sigma_i$ since $\sigma_i$ tends to infinity. On the other hand, the convergence rate in Theorem \ref{lem_realsignal_preratio} can be slower than the order of $n^{-1/2}$ due to the term $\mathsf{T}/\sigma_i$ in case $(1)$. Therefore, the results in Theorem \ref{lem_realsignal_preratio} are far from ideal, especially if our goal is to determine the asymptotic distribution of $\lambda_i$.
\end{remark}
To obtain a more precise description for the locations of $\lambda_i$, we employ the methodology from \cite{yu2024testing} and introduce the random quantities $\theta_i,1\le i\le m$ to trace the asymptotic behavior of each $\lambda_i$. We define a deterministic quantity $\theta_i$ as the unique solution to
\begin{equation}\label{eq_def_theta/sigma}
    \frac{\theta_i}{\sigma_i}=\Big(1-\frac{1}{p\theta_i}\sum_{k=1}^{p-m}\frac{\sigma_{m+k}}{1-\sigma^{-1}_{i}\sigma_{m+k}}\Big)^{-1},\quad \theta_i\in[\sigma_i,2\sigma_i],\quad 1\le i\le m.
\end{equation}
\cite{xie2026limiting} showed that under the elliptical model, $\theta_i$ is a closer approximation to $\lambda_i$ compared with $\sigma_i$. In our study, to address the effect of $D^2$, we need the second order description $\zeta_i$ to be the unique solution to
\begin{equation}\label{eq_def_zetai}
    \zeta_i=\frac{1}{p}\sum_{j=1}^n\phi_n\xi_j^2\Big(1-\frac{\phi_n\xi^2_j}{p\theta_i}\sum_{k=1}^{p-m}\frac{\sigma_{m+k}}{1-\theta_i^{-1}\sigma_{m+k}\zeta_i}\Big)^{-1},~\zeta_i\in[\frac{1}{p}\operatorname{tr}D^2,\frac{2}{p}\operatorname{tr}D^2].
\end{equation}
The existence and uniqueness of $\theta_i, \zeta_i$ are guaranteed under our assumptions for large $p, n$.
\begin{lemma}[Properties of $\theta_i, \zeta_i$]\label{lem_realsignal_secondratio1}
    Under Assumptions \ref{ass_xi} and \ref{ass_sigma}, $\theta_i$ is well defined for all large $n$, and $\zeta_i$ exists uniquely with probability tending to one. It holds that
    \begin{equation*}
        \frac{\theta_i}{\sigma_i}=1+\mathrm{O}\Big(\frac{\operatorname{tr}\Sigma_2}{p\sigma_i}\Big)=1+\mathrm O\Big(\frac{1}{\sigma_i}\Big),
    \end{equation*}
    and, for both tail conditions in Assumption \ref{ass_xi},
    \begin{equation}\label{eq_zeta_expansion}
        \zeta_i-\frac{\theta_i}{\sigma_i}=\frac{1}{n}\sum_{j=1}^n(\xi^2_j-1)+\frac{\operatorname{tr}\Sigma_2}{p\theta_i}\Big(\phi_n\frac1n\sum_{j=1}^n\xi_j^4-1\Big)+\OrderP\Big(\frac{1}{\sigma_i^2}\Big(1+\frac1n\sum_{j=1}^n\xi_j^6\Big)\Big).
    \end{equation}
\end{lemma}
Using these observations, we obtain a more precise characterization of $\lambda_i$. Let $\mathbf{e}_k$ be the $k$-th standard basis vector.
\begin{theorem}[Second Order Approximation]\label{thm_realsignal_secondratio2}
    Under Assumptions \ref{ass_xi} and \ref{ass_sigma}, for $1\le i\le m$,
    \begin{equation}\label{eq_secondorder_green}
        \frac{\lambda_i}{\theta_i}-1=\mathbf{e}^{\prime}_iUD^2U^{\prime}\mathbf{e}_i-\frac{1}{p}\operatorname{tr}D^2-\frac{\theta_i}{\sigma_i}+\zeta_i+\OrderP\Big(\frac{1}{\sigma_i^2}\Big(1+\frac1n\sum_{j=1}^n\xi_j^6\Big)\Big)+\orderP(\varrho_n),
     \end{equation}
    where $\varrho_n:=n^{-1/2}+\sqrt{\mathsf T/n}\,\mathbbm 1\big(\alpha\in(1,2]\big)$, and the indicator is present only under the polynomial tail \eqref{ass_xi_poly}.
 \end{theorem}

Combining Theorem \ref{thm_realsignal_secondratio2} with \eqref{eq_zeta_expansion}, and noting that $\mathbf e_i'UD^2U'\mathbf e_i-p^{-1}\tr D^2+n^{-1}\sum_j(\xi_j^2-1)=n^{-1}\sum_t(\xi_t^2\,p\,u_{it}^2-1)$, where $u_{it}$ denotes the $i$-th entry of $\mathbf u_t$, we arrive at the representation
\begin{equation}\label{eq_realsignal_representation}
\frac{\lambda_i}{\theta_i}-1=\frac1n\sum_{t=1}^n\big(\xi_t^2\,p\,u_{it}^2-1\big)+\frac{\operatorname{tr}\Sigma_2}{p\theta_i}\Big(\phi_n\frac1n\sum_{t=1}^n\xi_t^4-1\Big)+\mathcal E_i,
\end{equation}
\begin{equation*}
\mathcal E_i=\OrderP\Big(\frac{1}{\sigma_i^2}\Big(1+\frac1n\sum_{t=1}^n\xi_t^6\Big)\Big)+\orderP(\varrho_n).
\end{equation*}
The leading term is a normalized sum of the i.i.d.\ variables $\xi_t^2\,p\,u_{it}^2$, which have mean one, and the second term is a shift of order $\sigma_i^{-1}$ whose random part is negligible. Compared with Theorem \ref{lem_realsignal_preratio}, the term $\mathsf T/\sigma_i$ has disappeared, and the size of the fluctuation is governed by the tail index $\alpha$, as made precise in the next two results.

Specifically, in the case where $\xi^2$ either has a polynomial decay tail with $\alpha\in(2,+\infty)$ or has an exponential decay tail in Assumption \ref{ass_xi}, the following central limit theorem holds.


\begin{theorem}[CLT for real factors]\label{thm_realsignal_clt}
     Assume \eqref{ass_xi_poly} with $\alpha\in(2,+\infty)$ or \eqref{ass_xi_exp} in Assumption \ref{ass_xi}, and Assumption \ref{ass_sigma} hold. If $\alpha>4$ or \eqref{ass_xi_exp} holds, assume in addition that $\sigma_m/n^{1/4}\to\infty$. Then for $1\le i\le m$,
     \begin{equation}\label{eq_clt_real}
    \sqrt{n}\Big(\frac{\lambda_i}{\theta_i}-1-b_i\Big)\overset{d}{\rightarrow}\mathrm{N}\big(0,3\mathbb{E}\xi^4-1\big),\qquad b_i:=\frac{\operatorname{tr}\Sigma_2}{p\theta_i}\big(\phi_n\E\xi^4-1\big).
\end{equation}
\end{theorem}
\begin{remark}
    Several remarks are in order. First, the centering $b_i=\mathrm O(\sigma_i^{-1})$ is deterministic and explicit. If $\sigma_i/\sqrt n\to\infty$, then $\sqrt n\,b_i\to0$ and \eqref{eq_clt_real} holds with $b_i$ removed; otherwise $b_i$ cannot be dropped. Second, for multivariate normal data, for which $\xi^2=\chi^2_p/p$, we have $3\E\xi^4-1\to2$ and $b_i\approx(\phi_n-1)\bar\sigma/\sigma_i$, which reproduces the classical location $\lambda_i\approx\sigma_i+\phi_n\bar\sigma$ of a large spiked eigenvalue. Third, the additional condition $\sigma_m\gg n^{1/4}$ for $\alpha>4$ only serves to make the term of order $\sigma_i^{-2}$ in \eqref{eq_realsignal_representation} negligible on the $n^{-1/2}$ scale; for $\alpha\in(2,4]$ it follows from $\sigma_m\gg\mathsf T$. Last, it should be noted that $\theta_i$ is difficult to compute in practice; however, the above theorem establishes the asymptotic normality of the real factors, which serves as a crucial input for our testing procedure.
\end{remark}
When $\xi^2$ exhibits serious heavy-tailed decay with $\alpha\in(1,2]$, the asymptotic normality in Theorem \ref{thm_realsignal_clt} will not hold, but demands a larger scaling. In practice, it is enough to only consider the asymptotic mean and variance for $\lambda_i/\theta_i$. We have the following results for real factors under the serious heavy-tailed scenario.
\begin{theorem}[Real factors under serious heavy tails]\label{thm_realsignal_extremeheavy}
    Suppose \eqref{ass_xi_poly} holds with $\alpha\in(1,2]$ in Assumption \ref{ass_xi}. Then we have that for $1\le i\le m$  
    \begin{equation}\label{eq_realsignal_extremeheavy}
        \frac{\lambda_i}{\theta_i}-1=\OrderP\Big(\sqrt{\frac{\mathsf T}{n}}\Big).
    \end{equation}
    More precisely, the leading term of \eqref{eq_realsignal_representation} truncated at level $\mathsf T$, namely $n^{-1}\sum_t(\xi_t^2\mathbbm 1(\xi_t^2\le\mathsf T)\,p\,u_{it}^2-1)$, coincides with the untruncated term with probability tending to one, has mean $\mathrm O(\mathsf T^{1-\alpha}L(\mathsf T))$, and has variance at most $3\mathsf T/n$.
\end{theorem}

\subsection{Asymptotic behavior of spurious factors}\label{sec_main_fakesignals}

We first prepare several notations. Recalling the decomposition of $\mathcal{S}$ in \eqref{eq_decomposition_S}, we denote the matching matrices removing the spike structure from $\mathcal{S}$ (or $S$) as
\begin{equation}\label{eq_def_S2}
    S_2:=\Sigma^{1/2}_2UD^2U^{\prime}\Sigma_2^{1/2},\quad \mathcal{S}_2:=DU^{\prime}\Sigma_2UD.
\end{equation}
Denote $\lambda_{i,2}$ as the $i$-th largest eigenvalue of $S_2$ or $\mathcal{S}_2$. Denote the quantity $\mathsf{q}$ as 
\begin{equation}\label{eq_def_q}
    \mathsf{q}:=
    \begin{cases}
        n^{1/\alpha-\epsilon}\quad &\text{if \eqref{ass_xi_poly} in Assumption \ref{ass_xi} holds};\\
        1 \quad &\text{if \eqref{ass_xi_exp} in Assumption 
        \ref{ass_xi} holds},
    \end{cases}
\end{equation}
for some sufficiently small constant $\epsilon\in(0,1/(2\alpha))$. 

The following lemma indicates that the leading eigenvalues of $S$, despite the first $m$ spiked eigenvalues, are close to the leading eigenvalues of $\mathcal{S}_2$,
\begin{lemma}\label{lem_fakesignal_compare_nonspikes}
    Under Assumptions \ref{ass_xi} and \ref{ass_sigma},  we have that for $k\ge1$, 
    \begin{equation}
        |\lambda_{m+k}-\lambda_{k,2}|=\mathrm{O}_{\mathbb{P}}(n^{-1/2+2\epsilon}\mathsf{q}).
    \end{equation}
\end{lemma}
Typically, the largest eigenvalues of $\mathcal{S}_2$ are strongly influenced by the ordered statistics in $D^2$, whose phenomena is captured in the following theorem,
\begin{theorem}\label{thm_fakesignal_limit}
    Under Assumptions \ref{ass_xi} and \ref{ass_sigma},  we have  that for $k\ge1$
    \begin{equation}
     \frac{\lambda_{k,2}}{\phi_n\xi_{(k)}^2}=\Bar{\sigma}+\mathrm{o}_{\mathbb{P}}(1),
    \end{equation}
    where $\Bar{\sigma}:=p^{-1}\operatorname{tr}(\Sigma_2)$ and $\xi_{(1)}^2\ge \xi_{(2)}^2\ge\cdots\ge \xi_{(n)}^2$ are the ordered statistics of the squared radii in $D^2$.
\end{theorem}
Combined with Lemma \ref{lem_fakesignal_compare_nonspikes}, Theorem \ref{thm_fakesignal_limit} gives $\lambda_{m+k}/(\phi_n\xi^2_{(k)})\to\bar\sigma$ in probability for each fixed $k\ge1$. Under the polynomial tail, the ratio of the two largest radii, which plays a key role in Section \ref{sec_testing_procedures}.

\begin{lemma}[Separation of the two largest radii]\label{lem_radius_separation}
Suppose \eqref{ass_xi_poly} holds, and let $r_n:=\xi^2_{(1)}/\xi^2_{(2)}$. Then for every $\delta>0$, $\lim_{n\to\infty}\mathbb P(r_n\le1+\delta)=1-(1+\delta)^{-\alpha}$. In particular,
\begin{equation}\label{eq_radius_nondegenerate}
\lim_{\delta\downarrow0}\ \limsup_{n\to\infty}\mathbb P(r_n\le1+\delta)=0.
\end{equation}
If, in addition, Assumption \ref{ass_sigma} holds, then $\lambda_{m+1}/\lambda_{m+2}=r_n(1+\orderP(1))$.
\end{lemma}
\begin{remark}
    From Lemma \ref{lem_fakesignal_compare_nonspikes}, we observe that the spurious noise signals closely approximate the leading eigenvalues of $\mathcal{S}_2$. Furthermore, Theorem \ref{thm_fakesignal_limit} and Lemma \ref{lem_goodconfiguration} reveal that these leading eigenvalues of $\mathcal{S}_2$ diverge in alignment with the ordered statistics in $D^2$. This theoretical insight explains why true common factors and spurious noise signals become indistinguishable among the top eigenvalues of the data matrices.
\end{remark}

\section{Detection via Fluctuation Magnification}\label{sec_testing_procedures}
In this section we turn the contrast between Sections \ref{sec_main_realsignals} and \ref{sec_main_fakesignals} into a testing strategies for the problem \eqref{eq_hypothesistest}. The idea is elementary. A genuine factor eigenvalue aggregates the contributions of all $n$ observations, whereas a spurious eigenvalue is created by a single observation with an extreme radius. If every observation is reweighted by an independent random multiplier, the multipliers average out in the former case but are inherited one-for-one in the latter. Section \ref{sec_testing_bootstrap} makes this precise, and Section \ref{sec_testing_algorithm} turns it into a procedure for estimating $m$.



\subsection{Fluctuation magnification for \texorpdfstring{$S$}{S}}\label{sec_testing_bootstrap}
Fix a value $\bar h\in(0,1)$ and denote $\bar h=0.9$ throughout. For $h\in(0,\bar h]$, let $w$ be uniformly distributed on $[1-h,1+h]$, so that
\begin{equation}\label{eq_def_vh}
\E w=1,\qquad v(h):=\frac{\operatorname{Var}(w)}{(\E w)^2}=\frac{h^2}{3}.
\end{equation}
Suppose we observe the data matrix $Y$. We define a sequential magnifier data matrices $\widetilde{Y}_j,1\le j\le K$ from $Y$ for some large fixed integer $K>0$. Let 
         \begin{equation*}
             \widetilde{Y}_j:=Y\cdot W^{1/2}_j,\quad W_j=\operatorname{diag}(w_t)_{1\le t\le n},
         \end{equation*}
         where $w_t\sim w, 1\le t\le n$ are some properly chosen i.i.d. random variables whose details will be specified below. For $1\le j\le K$, we independently construct from $\widetilde{Y}_j$ with
         \begin{gather*}
             \widetilde{S}^{(j)}:=\frac{1}{n}\widetilde{Y}_j\widetilde{Y}_j^{\prime}=\Sigma^{1/2}UDW_jD^{\prime}U^{\prime}\Sigma^{1/2}\quad \widetilde{\mathcal{S}}^{(j)}:=W_j^{1/2}DU^{\prime}\Sigma UDW_j^{1/2}.
         \end{gather*}
         Let $\lambda_k^{(j)}$ be the $k$-th largest eigenvalue of $\widetilde{S}^{(j)}$.
         In the sequel, we use the notation $\mathbb{P}^*$ to denote the probability measure conditional on the sample $Y$ by $\mathbb{P}^*(\cdot)=\mathbb{P}(\cdot|\mathcal{F}_Y)$ where $\mathcal{F}_Y=\sigma(Y)$ is the sigma algebra generated from $\{\mathbf{y}_1,\dots,\mathbf{y}_n\}$. Then, we define $\overset{d^*}{\rightarrow}$, $\mathrm{o}_{\mathbb{P}^*}$ and $\mathrm{O}_{\mathbb{P}^*}$ accordingly from $\mathbb{P}^*$. And $\mathbb P$ denotes the joint law of $(Y,W_1,\dots,W_K)$. Two elementary facts below drive the analysis.

\smallskip
\noindent\textbf{(F1) Deterministic sandwich.} Since $(1-h)I\preceq W_j\preceq(1+h)I$, we have $(1-h)S\preceq\widetilde S^{(j)}\preceq(1+h)S$. Hence, by the Courant--Fischer theorem,
\begin{equation}\label{eq_sandwich}
X_k^{(j)}:=\lambda_k^{(j)}/\lambda_k\in[1-h,1+h]\quad\text{for every }k .
\end{equation}

\noindent\textbf{(F2) Magnified data are again EFM data.} Unconditionally, $\widetilde S^{(j)}$ is the sample covariance matrix of an EFM with radii $\xi_t^2w_{tj}$. These radii are i.i.d., independent of $U$, and satisfy $\E[\xi_t^2w_{tj}]=1$. They obey \eqref{ass_xi_poly} with the same index $\alpha$ (by Breiman's lemma \citep{breiman1965some}, as $w$ is bounded and bounded away from zero), or \eqref{ass_xi_exp}. Hence all results of Section \ref{sec_mainresults} apply to $\widetilde S^{(j)}$. In particular, the requirement $\sigma_m\gg\mathsf T\max_tw_{tj}$ is automatic, because $\max_tw_{tj}\le1+\bar h<2$.
\smallskip


For each index $k$ we measure the fluctuation of $\lambda_k^{(j)}$ through its conditional population version and its Monte Carlo estimate,
\begin{equation}\label{eq_testing_statistics}
\tau_k(h):=\frac{\operatorname{Var}^*(\lambda_k^{(1)})}{\big(\E^*\lambda_k^{(1)}\big)^2},\quad
\mathbb T_k:=\frac{K^{-1}\sum_{j=1}^K\big(\lambda_k^{(j)}-\bar\lambda_k\big)^2}{\bar\lambda_k^2},\quad\text{where } \bar\lambda_k:=\frac1K\sum_{j=1}^K\lambda_k^{(j)},
\end{equation}
and we propose the following statistic, which is the fluctuation of magnified eigenvalue normalized by the relative variance of the multiplier itself:
\begin{equation}\label{eq_def_Rstat}
\mathsf R_k:=\mathbb T_k/v(h).
\end{equation}
Neither $\tau_k(h)$ nor $\mathbb T_k$ changes if $\lambda_k^{(j)}$ is replaced by $X_k^{(j)}$. The normalization makes the testing procedure scale free. As shown below, $\tau_k(h)/v(h)\to0$ for genuine factors, whereas $\tau_{m+1}(h)/v(h)\to1$ for the leading spurious eigenvalue. We first quantify the Monte Carlo error.

\begin{lemma}[Monte Carlo error]\label{lem_mc_error}
There exists a constant $C_0$, depending only on $\bar h$, such that for every $h\in(0,\bar h]$, every index $k$, every $\eta\in(0,1)$ and every $K\ge\log(4/\eta)$,
\begin{equation*}
\mathbb P^*\Big(\big|\mathsf R_k-\tau_k(h)/v(h)\big|>C_0\sqrt{\log(4/\eta)/K}\Big)\le\eta .
\end{equation*}
\end{lemma}
Lemma \ref{lem_mc_error} is non-asymptotic and holds uniformly in $h$ and in the data which is guaranteed by the boundedness \eqref{eq_sandwich}. In particular, any $K=K_n\to\infty$ suffices for the asymptotic results below, where $K$ needs not grow with $n$ or depends on the tail index.

\smallskip
\noindent\textbf{Genuine factors.} Let $\mathbf v_1,\dots,\mathbf v_m$ be the leading eigenvectors of $S$, $V=(\mathbf v_1,\dots,\mathbf v_m)$ and $\Lambda=\diag(\lambda_1,\dots,\lambda_m)$. Define the leverage of observation $t$ in the leading eigenspace,
\begin{equation}\label{eq_def_leverage}
\ell_t:=\frac1n\,\mathbf y_t'V\Lambda^{-1}V'\mathbf y_t,\quad\text{so that}\quad\sum_{t=1}^n\ell_t=m .
\end{equation}
\begin{proposition}[Genuine factors under magnification]\label{prop_real_magnified}
Suppose Assumptions \ref{ass_xi} and \ref{ass_sigma} hold, and let $1\le k\le m$. Then, with probability tending to one,
\begin{equation}\label{eq_real_bound}
\sup_{0<h\le\bar h}\frac{\tau_k(h)}{v(h)}\le 8m\max_{1\le t\le n}\ell_t+C\Big(\frac{\lambda_{m+1}}{\lambda_m}\Big)^2,
\end{equation}
where $C$ depends only on $\bar h$, $m$ and $c_3$. Moreover, $\max_t\ell_t+(\lambda_{m+1}/\lambda_m)^2=\OrderP(\mathsf T\log n/n+\mathsf T/\sigma_m)$; in particular, $\sup_{0<h\le\bar h}\tau_k(h)/v(h)=\orderP(1)$.
\end{proposition}
The first-order picture behind \eqref{eq_real_bound} is $\lambda_k^{(j)}/\lambda_k\approx\sum_t w_{tj}z_{tk}$, with weights $z_{tk}:=(\mathbf v_k'\mathbf y_t)^2/(n\lambda_k)\ge0$ that sum to one. Hence $\tau_k(h)/v(h)\approx\sum_tz_{tk}^2=:1/n_{\mathrm{eff},k}$, the reciprocal of the effective number of observations carrying the $k$-th factor. Heavy tails could make a few weights larger. But as long as $\sigma_m\gg\mathsf T$, no single observation can carry a non-vanishing share of a genuine factor since $\max_tz_{tk}\le\max_t\ell_t\to0$.

\smallskip
\noindent\textbf{The leading spurious eigenvalue.} For the rest of this section we focus on the polynomial tail \eqref{ass_xi_poly} case. Let $t^*$ be the index of the largest radius, $\xi^2_{t^*}=\xi^2_{(1)}$, and recall $r_n=\xi^2_{(1)}/\xi^2_{(2)}$ from Lemma \ref{lem_radius_separation}.
\begin{proposition}[Spurious factor under magnification]\label{prop_fake_magnified}
Suppose \eqref{ass_xi_poly} and Assumption \ref{ass_sigma} hold. Couple the multipliers for all $h$ simultaneously as $w_{t1}=1+h\gamma_t$, where the $\gamma_t$ are i.i.d.\ uniform on $[-1,1]$, and let $M_n(h):=\max_t\xi_t^2w_{t1}/\xi^2_{(1)}$. Then
\begin{equation}\label{eq_fake_approx}
\Delta_n^2:=\sup_{0<h\le\bar h}\E^*\Big|\frac{\lambda^{(1)}_{m+1}}{\lambda_{m+1}}-M_n(h)\Big|^2\overset{\mathbb P}{\longrightarrow}0 .
\end{equation}
Consequently, there is a constant $C$ depending only on $\bar h$ such that the following holds for any $h=h(Y)\in(0,\bar h]$ chosen as a function of the data. On the event that the multiplier preserves the order of the two largest radii,
\begin{equation}\label{eq_order_preserving}
\frac{1+h}{1-h}\le r_n=\frac{\xi^2_{(1)}}{\xi^2_{(2)}},
\end{equation}
we have $M_n(h)=w_{t^*1}$ and $|\tau_{m+1}(h)-v(h)|\le C(\Delta_n+\Delta_n^2)$.
\end{proposition}
Propositions \ref{prop_real_magnified} and \ref{prop_fake_magnified} formalize the heuristic stated at the beginning of this section and form the theoretical basis of the test. The key quantity that separates the genuine and spurious cases is the normalized relative variance $\tau_k(h)/v(h)$. By Lemma \ref{lem_mc_error}, the statistic $\mathsf R_k$ estimates it with an error of order $K^{-1/2}$. 

For a genuine factor ($k\le m$), $\lambda_k^{(j)}/\lambda_k$ behaves like the weighted average $\sum_tw_{tj}z_{tk}$ of $n$ multipliers, with weights $z_{tk}\ge0$ summing to one. Its relative variance is therefore approximately $v(h)\sum_tz_{tk}^2\le v(h)\max_tz_{tk}$. Since no single observation carries a non-vanishing share of a genuine factor, $\tau_k(h)/v(h)\to0$ (Proposition \ref{prop_real_magnified}).

For the leading spurious eigenvalue ($k=m+1$), $\lambda^{(j)}_{m+1}/\lambda_{m+1}$ behaves like the single multiplier $w_{t^*j}$ attached to the largest radius. Its relative variance is therefore that of the multiplier itself, and $\tau_{m+1}(h)/v(h)\to1$ whenever $h$ preserves the order of the two largest radii, as in \eqref{eq_order_preserving} (Proposition \ref{prop_fake_magnified}).

Consequently, $\mathsf R_k$ is asymptotically close to $0$ for every genuine factor and close to $1$ for the first spurious eigenvalue. Any fixed threshold $\kappa\in(0,1)$ separates the two cases, which is made precise in Theorem \ref{thm_firstround} below.

\begin{remark}\label{rem_bootstrap}
Several remarks are in order. First, the limiting fluctuation of the spurious eigenvalue is explicit: $\tau_{m+1}(h)\approx v(h)=h^2/3$, which can be calculated after choosing $h$. Normalizing $\tau_{m+1}$ by $v(h)$ ensures that the limit stays bounded away from zero. This matters because a data-driven $h$ may be small whenever the two largest radii are close Second, in finite samples $\tau_{m+1}(h)/v(h)$ is typically somewhat smaller than one, because $\lambda_{m+1}$ also contains a bulk-induced component that does not scale with the extreme radius. Third, analogous statements hold for $\lambda_{m+2},\lambda_{m+3},\dots$, but only the leading spurious eigenvalue is needed to locate the boundary between genuine and spurious eigenvalues. Fourth, Other bounded magnifiers can be used. For instance, if $w$ takes the values $1\pm h$ with probability $1/2$ each, all results remain valid with $v(h)=h^2$.
\end{remark}

\begin{remark}\label{rem_exptail_spurious}
Under the exponential tail \eqref{ass_xi_exp}, Proposition \ref{prop_real_magnified} holds as stated. Under the Weibull-type condition, Proposition \ref{prop_fake_magnified} also continues to hold with the same proof. However, $r_n\to1$ in this case, so an order-preserving magnifier must have $h=h_n\to0$ at the rate $(\log n)^{-1}$. Then $\mathsf R_{m+1}$ is still asymptotically close to one, provided the relative error of $\lambda_{m+1}/\lambda_{m+2}$ as an estimator of $r_n$ is of smaller order than $r_n-1$. This would require a quantitative version of Theorem \ref{thm_fakesignal_limit}, which we do not pursue. Theorem \ref{thm_firstround} below is therefore stated under \eqref{ass_xi_poly}.
\end{remark}

\subsection{Detection algorithm}\label{sec_testing_algorithm}
Propositions \ref{prop_real_magnified} and \ref{prop_fake_magnified} suggest scanning the indices $i=1,2,\dots$ and flagging $\lambda_i$ as spurious when $\mathsf R_i$ is large. Two practical issues remain.

First, $h$ must satisfy \eqref{eq_order_preserving} at the unknown index $m+1$. For this we use the eigenvalue ratio $\lambda_i/\lambda_{i+1}$, which for $i=m+1$ estimates $r_n$ consistently by Lemma \ref{lem_radius_separation}. To protect against estimation error, we only use its square root. Second, the threshold must separate the limits $0$ and $1$. In view of Remark \ref{rem_bootstrap}(ii), we recommend $\kappa=0.4$. 

Our testing procedure below has two steps. The first step locates the first spurious index, and the second step combines it with an existing estimator of $m$, as a safeguard for the case in which no spurious eigenvalue is present. Recall our statistics $\mathsf{R}_k$ in \eqref{eq_def_Rstat}.

\textbf{1. First step detection.}
\begin{algorithm}[htbp]
\caption{First step: index-wise detection of spurious eigenvalues}\label{alg_firstround_bootstrap}
\normalsize
\begin{flushleft}
\noindent{\bf Inputs:} data matrix $Y$; maximal index $o$ (e.g.\ $o=10$); number of magnifications $K$ (e.g.\ $K=200$); $\bar h\in(0,1)$ (default $0.9$); threshold $\kappa\in(0,1)$ (default $0.4$).

\noindent{\bf Step One:} Compute the eigenvalues $\lambda_1\ge\dots\ge\lambda_{o+1}$ of $S=n^{-1}YY'$. For $i=1,\dots,o$, set
\begin{equation*}
s_i:=\min\Big\{\frac{1+\bar h}{1-\bar h},\Big(\frac{\lambda_i}{\lambda_{i+1}}\Big)^{1/2}\Big\},\quad h_i:=\frac{s_i-1}{s_i+1}.
\end{equation*}
Equivalently, $w\sim U[1-h_i,1+h_i]$, which has mean one and endpoint ration $(1+h_i)/(1-h_i)=s_i$.

\noindent{\bf Step Two:} For $i=1,\dots,o$ and $j=1,\dots,K$, generate $W^{(j)}_i=\diag(w_{1j},\dots,w_{nj})$ with i.i.d.\ $U[1-h_i,1+h_i]$ entries. Let $\lambda_i^{(j)}$ be the $i$-th largest eigenvalue of $n^{-1}YW_i^{(j)}Y'$.

\noindent{\bf Step Three:} Compute $\mathbb T_i$ as in \eqref{eq_testing_statistics} and $\mathsf R_i:=3\,\mathbb T_i/h_i^2$.

\noindent{\bf Output:} The set of flagged indices $\widehat{\mathcal S}:=\{1\le i\le o:\mathsf R_i>\kappa\}$, its size $\mathsf f^*:=\#\widehat{\mathcal S}$, and $r^*:=\min\widehat{\mathcal S}$, with the convention $r^*:=o+1$ if $\widehat{\mathcal S}=\emptyset$.
\end{flushleft}
\end{algorithm}

\begin{theorem}[Consistency of the first step]\label{thm_firstround}
Suppose \eqref{ass_xi_poly} and Assumption \ref{ass_sigma} hold. Let $o\ge m+1$, $\kappa\in(0,1)$ and $\bar h\in(0,1)$ be fixed, and let $K=K_n\to\infty$. Then
\begin{equation}\label{eq_firstround_consistency}
\lim_{n\to\infty}\max_{1\le i\le m}\mathbb P(\mathsf R_i>\kappa)=0,\qquad \lim_{n\to\infty}\mathbb P(\mathsf R_{m+1}>\kappa)=1 ,
\end{equation}
and consequently $\lim_{n\to\infty}\mathbb P(r^*=m+1)=1$.
\end{theorem}

\begin{remark}\label{rem_firstround}
Several remarks are in order. First, by Lemma \ref{lem_mc_error}, the Monte Carlo standard error of $\mathsf R_i$ is of order $K^{-1/2}$, uniformly in all related quantities. For the leading spurious index, $\mathsf R_{m+1}$ is approximately the sample variance of $K$ uniformly distributed variables divided by its expectation, whose standard deviation is about $0.9/\sqrt K$. Thus, $K=200$ already places $\kappa=0.4$ many standard errors away from both limits. Second, under Assumption \ref{ass_sigma}, the genuine eigenvalues are of order $\sigma_m\gg\mathsf T$, while $\lambda_{m+1}=\OrderP(\mathsf T)$, so spurious eigenvalues cannot overtake genuine ones asymptotically. In finite samples with $\sigma_m$ comparable to the spurious level $\phi_n\bar\sigma\xi^2_{(1)}$, a spurious eigenvalue may interlace with the genuine ones. The statistic $\mathsf R_i$ still flags it, but $r^*-1$ then underestimates $m$ (cf.\ Section \ref{sec_simulation}).
\end{remark}

\textbf{2. Second step detection.} 
The first step identifies the leading spurious eigenvalue. Let $r^*$ be the output of Algorithm \ref{alg_firstround_bootstrap}, the first index at which a spurious eigenvalue is detected, with $r^*=o+1$ if none is detected among the first $o$ indices. The region $\{j:j<r^*\}$ is free of spurious contamination. When the data do not contain heavy-tailed noise, the first step procedure may incorrectly classify certain bulk eigenvalues as spurious. To mitigate this issue, we subsequently apply a refined factor-number detection method exclusively to the corresponding potential real signals, which yields more accurate estimates and avoids over-estimation. In this paper, we use the estimators of \cite{onatski2010determining} and \cite{Dobriban} as illustrative examples, see Algorithm \ref{alg_secondround_estrealsignal} below.

\begin{algorithm}[htbp]
\caption{Second step estimation for the number of common factors}\label{alg_secondround_estrealsignal}
\normalsize
\begin{flushleft}
\noindent\textbf{Inputs:} $r^*$ and $\mathsf f^*$ from Algorithm \ref{alg_firstround_bootstrap}, and the sample covariance matrix $S$.

\noindent\textbf{Step One:} Compute an initial estimate $k$ using a factor-number estimator (e.g., the ED estimator of \cite{onatski2010determining} or the DDPA+ procedure of \cite{Dobriban}) applied to $S$.

\noindent\textbf{Step Two:} Set $\widehat{r}=k$ and $\widehat{\mathsf f}=0$ if $k < r^*$ (no spurious signal occurs among the first $k$ indices). Otherwise, set $\widehat{r}=r^*-1$ and $\widehat{\mathsf f}=\mathsf f^*$ (a spurious signal is detected at $r^* \le k$, so the estimate is truncated to $r^*-1$). In particular, $\widehat r=\min\{k,r^*-1\}$.

\noindent\textbf{Output:} The estimated number of common factors $\widehat{r}$; reject $\mathbf H_0$ in \eqref{eq_hypothesistest} if $\widehat{\mathsf f}>0$.
\end{flushleft}
\end{algorithm}

    \begin{theorem}[Consistency of the two-step procedure]\label{lem_secondround_consistency}
Suppose the assumptions of Theorem \ref{thm_firstround} hold, and the initial estimator in Algorithm \ref{alg_secondround_estrealsignal} does not underestimate asymptotically, i.e., $\lim_{n\to\infty}\mathbb P(k\ge m)=1$. Then,
\[
\lim_{n\to\infty}\mathbb P(\widehat r=m)=1.
\]
\end{theorem}

\begin{remark}\label{rem_secondround}
Several remarks are in order. First, Theorem \ref{lem_secondround_consistency} is a direct consequence of Theorem \ref{thm_firstround}: on the event $\{r^*=m+1,\,k\ge m\}$, Algorithm \ref{alg_secondround_estrealsignal} returns $\widehat r=\min\{k,r^*-1\}=m$. In words, the first step removes exactly the spurious eigenvalues that cause $k$ to over-estimate, and the second step only needs $k$ not to under-estimate. Second, the condition $\mathbb P(k\ge m)\to1$ is mild for the estimators of \cite{onatski2010determining} and \cite{Dobriban}. The genuine eigenvalues are separated from the rest by gaps of order $\sigma_m\to\infty$, and heavy tails make these methods tend to over-estimate rather than under-estimate. The condition can also be checked numerically (Section \ref{sec_simulation}). Third, if the radius of EFM is bounded, therefore no spurious eigenvalue arises. This light-tailed case lies outside Assumption \ref{ass_xi}. However, the proof of Proposition \ref{prop_real_magnified} applies unchanged, since it only uses $\lambda_m\asymp\sigma_m$, $\lambda_{m+1}=\OrderP(\max_t\xi_t^2)$ and $\max_t\ell_t\to0$. Hence $\mathbb P(r^*\ge m+1)\to1$, and $\widehat r$ inherits the consistency of $k$. This is the safeguard role of the second step. Fourth, the estimators of \cite{onatski2010determining} and \cite{Dobriban} can be replaced by any estimator with the property in (ii), such as those in \cite{bai2018consistency}.
\end{remark}


\section{Simulation}\label{sec_simulation}
In this section, we conduct simulation studies to examine two aspects:
(1) factor-number overestimation caused by spurious signals in heavy-tailed EFMs; and
(2) the effectiveness of our fluctuation magnification algorithm (Section~\ref{sec_testing_algorithm}) in detecting these signals and estimating the true factor number.
We consider data $\mathbf{y}$ from model~\eqref{eq_intro_elldistr} with polynomially decaying tails in two regimes:
(I) typical heavy tails, with $\alpha\in(2,\infty)$; and
(II) serious heavy tails, with $\alpha\in(1,2]$.

\subsection{Common factor selection: Typical heavy-tailed case}

We first establish the settings for the typical heavy-tailed case for $\mathbf{y}$ in \eqref{eq_intro_elldistr}, characterized by a polynomial decay tail with $\alpha \in (2, +\infty)$. To represent various scenarios of heavy-tailed data, we employ multivariate t-distributions with degrees of freedom set at 6.0, 7.0, and 8.0 (i.e., $t(6.0)$, $t(7.0)$, and $t(8.0)$, respectively), all of which satisfy the condition $\alpha > 2$. We consider a high-dimensional framework with $p=500$ (number of variables) and $n=500$ (number of observations). Our analysis focuses on two population covariance matrices $\Sigma$:
\begin{itemize}
    \item[(i)] $\Sigma_{\mathrm{I}} = \operatorname{diag}\{60, 30, 1, \ldots, 1\}$
    \item[(ii)] $\Sigma_{\mathrm{II}} = \operatorname{diag}\{60, 40, 20, 1, \ldots, 1\}$.
\end{itemize}
Here, the larger diagonal components represent the distinct common factors. It is clear that $\Sigma_{I}$ encompasses two common factors induced by $\{60, 30\}$, while $\Sigma_{II}$ encompasses three common factors induced by $\{60, 40, 20\}$.

In our comparative study, we consider three baseline methods for
estimating the number of factors: the ED estimator of
\cite{onatski2010determining} (denoted as ``ED''), a fixed-threshold
rule based on the consecutive eigenvalue-gap ratios used by
\cite{onatski2009testing} (denoted as ``GR''), and the DDPA+
procedure of \cite{Dobriban} (denoted as ``DDPA+'').
For GR, we define
\[
G_i=\frac{\lambda_i-\lambda_{i+1}}
          {\lambda_{i+1}-\lambda_{i+2}}
\]
and select the first index $i\in\{1,\ldots,8\}$ satisfying $G_i>6$,
returning zero if no such index exists. Their versions enhanced
by the adaptive Fluctuation Magnification (MF) procedure proposed
in Section~\ref{sec_testing_algorithm} are referred to as
``ED with MF'', ``GR with MF'', and ``DDPA+ with MF'', respectively.

In our simulations, Algorithm \ref{alg_firstround_bootstrap} is run with its default settings $\bar h=0.9$, $\kappa=0.4$, $o=10$ and $K=200$; the multipliers are therefore chosen in a data-driven way through the eigenvalue ratios, as described in Step One. Detailed experimental results are presented in Figure \ref{eig4.3}, and Tables \ref{tab:af_comparison1} and \ref{tab:af_comparison2}.

Specifically, Figure~\ref{eig4.3} illustrates the performance of the methods for a selected example from Case~(i), using the ED estimator under $t(6)$. Panel (a) displays the first 20 eigenvalues of the sample covariance matrix. Panel (b) shows the adjacent eigenvalue gaps and the ED threshold, yielding an initial estimate of $k=3$. Panel (c) demonstrates the identification of spurious signals via our proposed normalized statistics $\mathsf{R}_i=3\mathbb{T}_i/h_i^2$ from Algorithm~\ref{alg_firstround_bootstrap}; the statistic corresponding to the spurious signal is markedly larger than those associated with real signals and lies well above the threshold $\kappa=0.4$. These observations illustrate how Algorithm~\ref{alg_firstround_bootstrap} detects the first spurious signal and corrects the overestimation by the ED estimator.

\begin{figure}[htbp]
    \centering
    \includegraphics[width=0.8\textwidth]{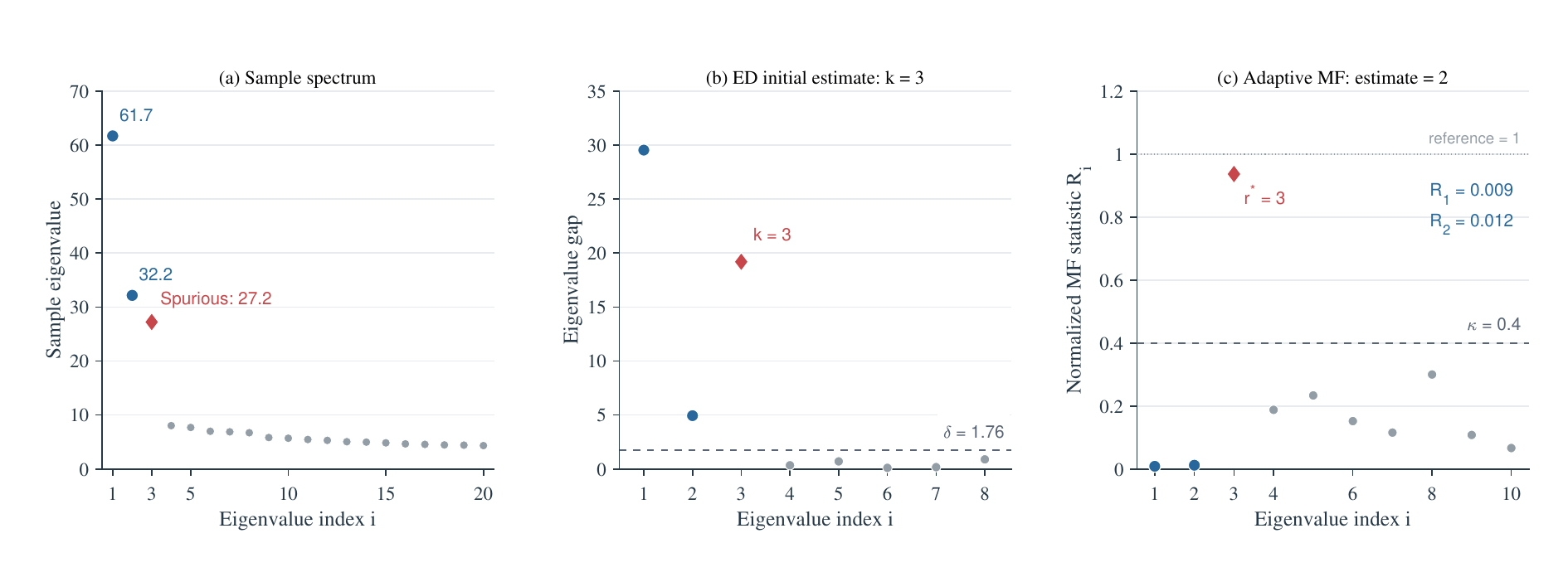} %
    \caption{Spurious signal detection under Case~(i), with standardized multivariate $t(6)$ observations. Adaptive MF with $\kappa=0.4$ identifies the first spurious signal at $r^*=3$ and corrects the ED estimate from $k=3$ to the true factor number $m=2$.}
    \label{eig4.3}
\end{figure}


To compare the performance of ED, GR, and DDPA+ with and without adaptive MF under Cases~(i) and~(ii), we conducted 500 independent replications for each setting. The effective sample size (ESS) reported in Tables~\ref{tab:af_comparison1} and~\ref{tab:af_comparison2} denotes the number of completed replications retained under the screening criterion $\lambda_m/\lambda_{m+1}\ge1.15$, where $m$ is the true factor number specified in the simulation. Specifically, we exclude replications with $\lambda_2/\lambda_3<1.15$ for Case~(i), where $m=2$, and those with $\lambda_3/\lambda_4<1.15$ for Case~(ii), where $m=3$, to avoid nearly coincident sample eigenvalues at these indices. All methods are evaluated on the same retained replications, and the reported rates therefore describe performance conditional on this screening criterion. The tables present the underestimation (Under), overestimation (Over), and correct estimation (Right) rates for three standardized multivariate $t$-distributions with different degrees of freedom.


\begin{table}[!htbp]
\centering
\small
\spacingset{1}
\setlength{\tabcolsep}{4pt}
\renewcommand{\arraystretch}{0.95}
\caption{ED, GR, and DDPA+ with/without adaptive MF under Case~(i).}
\label{tab:af_comparison1}
\resizebox{0.75\textwidth}{!}{%
\begin{tabular}{@{}lll ccc ccc@{}}
\toprule
\multicolumn{3}{c}{} & \multicolumn{3}{c}{\textbf{Without MF}} & \multicolumn{3}{c}{\textbf{With MF}} \\
\cmidrule(lr){4-6} \cmidrule(l){7-9}
\textbf{Dist.} & \textbf{ESS} & \textbf{Method} & \textbf{Under} & \textbf{Over} & \textbf{Right} & \textbf{Under} & \textbf{Over} & \textbf{Right} \\
\midrule
\multirow{3}{*}{$t(6)$} & \multirow{3}{*}{499} & \textbf{ED} & 0.4\% & 46.5\% & 53.1\% & 1.4\% & 8.8\% & 89.8\% \\
& & \textbf{GR} & 8.0\% & 19.6\% & 72.3\% & 8.4\% & 3.0\% & 88.6\% \\
& & \textbf{DDPA+} & 0.0\% & 74.7\% & 25.3\% & 1.0\% & 16.0\% & 83.0\% \\
\addlinespace[1.5mm]
\multirow{3}{*}{$t(7)$} & \multirow{3}{*}{500} & \textbf{ED} & 0.0\% & 43.4\% & 56.6\% & 0.0\% & 6.4\% & 93.6\% \\
& & \textbf{GR} & 3.0\% & 15.4\% & 81.6\% & 3.0\% & 1.8\% & 95.2\% \\
& & \textbf{DDPA+} & 0.0\% & 69.8\% & 30.2\% & 0.0\% & 13.4\% & 86.6\% \\
\addlinespace[1.5mm]
\multirow{3}{*}{$t(8)$} & \multirow{3}{*}{500} & \textbf{ED} & 0.0\% & 43.4\% & 56.6\% & 0.0\% & 6.6\% & 93.4\% \\
& & \textbf{GR} & 1.8\% & 11.6\% & 86.6\% & 1.8\% & 1.2\% & 97.0\% \\
& & \textbf{DDPA+} & 0.0\% & 58.2\% & 41.8\% & 0.0\% & 9.6\% & 90.4\% \\
\bottomrule
\end{tabular}
}
\end{table}

\begin{table}[!htbp]
\centering
\small
\spacingset{1}
\setlength{\tabcolsep}{4pt}
\renewcommand{\arraystretch}{0.95}
\caption{ED, GR, and DDPA+ with/without adaptive MF under Case~(ii).}
\label{tab:af_comparison2}
\resizebox{0.75\textwidth}{!}{%
\begin{tabular}{@{}lll ccc ccc@{}}
\toprule
\multicolumn{3}{c}{} & \multicolumn{3}{c}{\textbf{Without MF}} & \multicolumn{3}{c}{\textbf{With MF}} \\
\cmidrule(lr){4-6} \cmidrule(l){7-9}
\textbf{Dist.} & \textbf{ESS} & \textbf{Method} & \textbf{Under} & \textbf{Over} & \textbf{Right} & \textbf{Under} & \textbf{Over} & \textbf{Right} \\
\midrule
\multirow{3}{*}{$t(6)$} & \multirow{3}{*}{497} & \textbf{ED} & 1.6\% & 41.6\% & 56.7\% & 3.6\% & 11.1\% & 85.3\% \\
& & \textbf{Ratio} & 18.1\% & 26.4\% & 55.5\% & 19.5\% & 5.4\% & 75.1\% \\
& & \textbf{DDPA+} & 0.0\% & 73.8\% & 26.2\% & 2.0\% & 23.7\% & 74.2\% \\
\addlinespace[1.5mm]
\multirow{3}{*}{$t(7)$} & \multirow{3}{*}{498} & \textbf{ED} & 1.0\% & 40.6\% & 58.4\% & 1.2\% & 7.8\% & 91.0\% \\
& & \textbf{Ratio} & 12.4\% & 19.9\% & 67.7\% & 12.9\% & 3.2\% & 83.9\% \\
& & \textbf{DDPA+} & 0.2\% & 66.7\% & 33.1\% & 0.6\% & 16.1\% & 83.3\% \\
\addlinespace[1.5mm]
\multirow{3}{*}{$t(8)$} & \multirow{3}{*}{500} & \textbf{ED} & 0.2\% & 36.2\% & 63.6\% & 0.6\% & 8.2\% & 91.2\% \\
& & \textbf{Ratio} & 7.4\% & 17.4\% & 75.2\% & 7.8\% & 1.4\% & 90.8\% \\
& & \textbf{DDPA+} & 0.0\% & 58.0\% & 42.0\% & 0.4\% & 14.0\% & 85.6\% \\
\bottomrule
\end{tabular}
}
\end{table}


Tables~\ref{tab:af_comparison1} and~\ref{tab:af_comparison2} show that adaptive MF can substantially reduce overestimation and improve the correct estimation rate. The gains are particularly evident in settings where the baseline estimators frequently overestimate the factor number, although their magnitude varies across methods and degrees of freedom. These improvements are accompanied by only modest increases in underestimation.

\subsection{Common factor selection: Serious heavy-tailed case}
Next, we evaluate the detection behavior under serious heavy-tailed settings for $\mathbf{y}$ described in \eqref{eq_intro_elldistr}, characterized by polynomially decayed tails with $\alpha \in (1, 2]$. The remaining settings align with those in the Typical Heavy-tailed Case; only the differences are described here. We use multivariate t-distributions with degrees of freedom 3.0, 3.5, and 4.0 (denoted as $t(3.0)$, $t(3.5)$, and $t(4.0)$, respectively) to model the polynomial decay tail with $\alpha \in (1, 2]$. Additionally, the covariance matrix $\Sigma$ is configured with larger spiked eigenvalues as follows:
\begin{itemize}
    \item[(iii)] $\Sigma_{\mathrm{III}} = \operatorname{diag}\{480, 240, 1, \ldots, 1\}$
    \item[(iv)] $\Sigma_{\mathrm{IV}}= \operatorname{diag}\{480, 320, 160, 1, \ldots, 1\}$
\end{itemize}

Figure~\ref{eig2.5} illustrates the performance of the methods for a representative case (Case (iv)) using the ED estimator under $t(3.5)$ as an example.

\begin{figure}[htbp]
    \centering
    \includegraphics[width=0.8\textwidth]{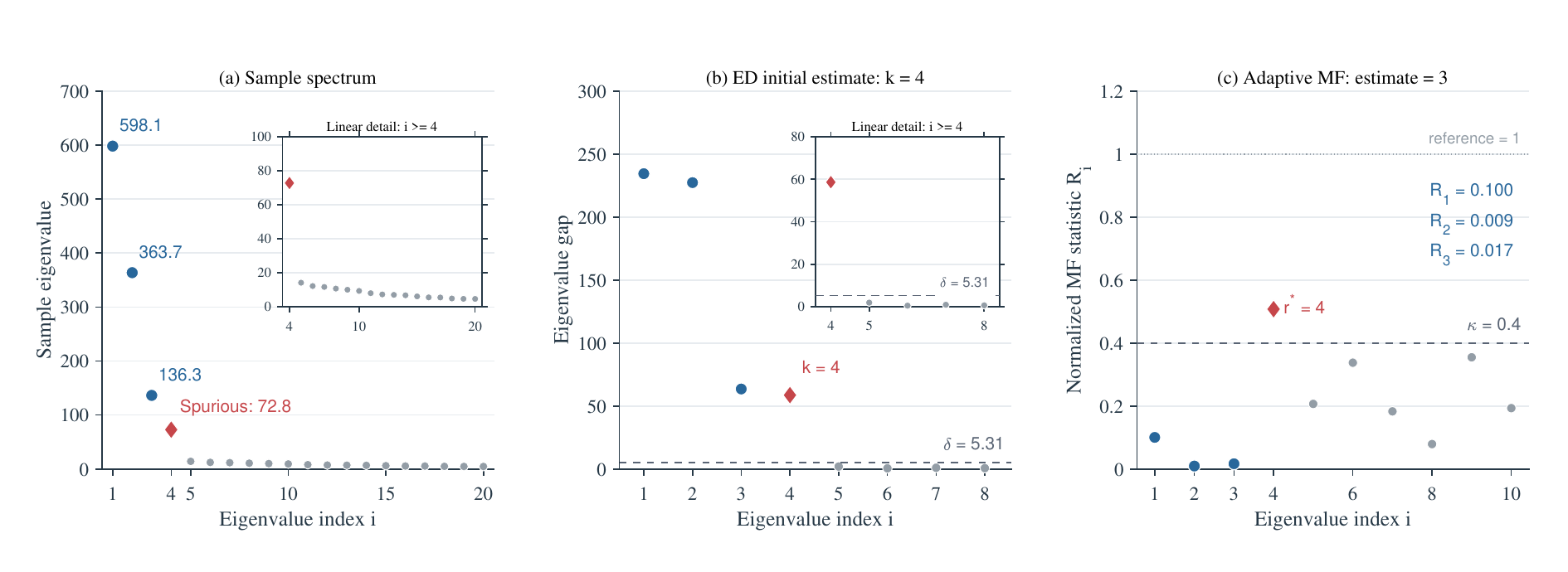} %
    \caption{Spurious signal detection under Case~(iv), with standardized multivariate $t(3.5)$ observations. Adaptive MF with $\kappa=0.4$ identifies the first spurious signal at $r^*=4$ and corrects the ED estimate from $k=4$ to the true factor number $m=3$.}
    \label{eig2.5}
\end{figure}


Using the same evaluation framework as in the typical heavy-tailed case, we assess the underestimation, overestimation, and correct estimation rates for three $t$-distributions with varying degrees of freedom.


\begin{table}[htbp]
\centering
\small
\spacingset{1}
\setlength{\tabcolsep}{4pt}
\renewcommand{\arraystretch}{0.95}
\caption{ED, GR, and DDPA+ with/without adaptive MF under Case~(iii).}
\label{tab:af_comparison3}
\resizebox{0.75\textwidth}{!}{%
\begin{tabular}{@{}lll ccc ccc@{}}
\toprule
\multicolumn{3}{c}{} & \multicolumn{3}{c}{\textbf{Without MF}} & \multicolumn{3}{c}{\textbf{With MF}} \\
\cmidrule(lr){4-6} \cmidrule(l){7-9}
\textbf{Dist.} & \textbf{ESS} & \textbf{Method} & \textbf{Under} & \textbf{Over} & \textbf{Right} & \textbf{Under} & \textbf{Over} & \textbf{Right} \\
\midrule
\multirow{3}{*}{$t(3)$} & \multirow{3}{*}{500} & \textbf{ED} & 0.0\% & 71.2\% & 28.8\% & 4.4\% & 11.8\% & 83.8\% \\
& & \textbf{Ratio} & 4.6\% & 15.6\% & 79.8\% & 7.0\% & 4.0\% & 89.0\% \\
& & \textbf{DDPA+} & 0.2\% & 96.2\% & 3.6\% & 4.6\% & 16.6\% & 78.8\% \\
\addlinespace[1.5mm]
\multirow{3}{*}{$t(3.5)$} & \multirow{3}{*}{500} & \textbf{ED} & 0.0\% & 64.4\% & 35.6\% & 2.4\% & 8.8\% & 88.8\% \\
& & \textbf{Ratio} & 1.6\% & 9.8\% & 88.6\% & 3.0\% & 2.4\% & 94.6\% \\
& & \textbf{DDPA+} & 0.0\% & 92.6\% & 7.4\% & 2.4\% & 14.2\% & 83.4\% \\
\addlinespace[1.5mm]
\multirow{3}{*}{$t(4)$} & \multirow{3}{*}{500} & \textbf{ED} & 0.0\% & 63.0\% & 37.0\% & 0.8\% & 5.2\% & 94.0\% \\
& & \textbf{Ratio} & 1.2\% & 5.2\% & 93.6\% & 1.6\% & 0.6\% & 97.8\% \\
& & \textbf{DDPA+} & 0.0\% & 90.8\% & 9.2\% & 0.8\% & 9.2\% & 90.0\% \\
\bottomrule
\end{tabular}
}
\end{table}


\begin{table}[htbp]
\centering
\small
\setlength{\tabcolsep}{4pt}
\renewcommand{\arraystretch}{0.95}
\caption{ED, GR, and DDPA+ with/without adaptive MF under Case~(iv).}
\label{tab:af_comparison4}
\resizebox{0.75\textwidth}{!}{
\begin{tabular}{@{}lll ccc ccc@{}}
\toprule
\multicolumn{3}{c}{} & \multicolumn{3}{c}{\textbf{Without MF}} & \multicolumn{3}{c}{\textbf{With MF}} \\
\cmidrule(lr){4-6} \cmidrule(l){7-9}
\textbf{Dist.} & \textbf{ESS} & \textbf{Method} & \textbf{Under} & \textbf{Over} & \textbf{Right} & \textbf{Under} & \textbf{Over} & \textbf{Right} \\
\midrule
\multirow{3}{*}{$t(3)$} & \multirow{3}{*}{499} & \textbf{ED} & 0.0\% & 65.9\% & 34.1\% & 4.4\% & 20.4\% & 75.2\% \\
& & \textbf{Ratio} & 8.4\% & 21.0\% & 70.5\% & 10.6\% & 7.8\% & 81.6\% \\
& & \textbf{DDPA+} & 1.0\% & 94.2\% & 4.8\% & 5.4\% & 26.9\% & 67.7\% \\
\addlinespace[1.5mm]
\multirow{3}{*}{$t(3.5)$} & \multirow{3}{*}{500} & \textbf{ED} & 0.0\% & 56.8\% & 43.2\% & 2.0\% & 14.4\% & 83.6\% \\
& & \textbf{Ratio} & 2.6\% & 13.0\% & 84.4\% & 3.8\% & 4.4\% & 91.8\% \\
& & \textbf{DDPA+} & 0.6\% & 91.8\% & 7.6\% & 2.6\% & 23.6\% & 73.8\% \\
\addlinespace[1.5mm]
\multirow{3}{*}{$t(4)$} & \multirow{3}{*}{499} & \textbf{ED} & 0.0\% & 57.3\% & 42.7\% & 1.0\% & 10.6\% & 88.4\% \\
& & \textbf{Ratio} & 1.0\% & 8.6\% & 90.4\% & 1.6\% & 2.6\% & 95.8\% \\
& & \textbf{DDPA+} & 0.2\% & 92.2\% & 7.6\% & 1.2\% & 17.6\% & 81.2\% \\
\bottomrule
\end{tabular}
}
\end{table}

Tables~\ref{tab:af_comparison3} and~\ref{tab:af_comparison4} show that adaptive MF remains effective in the serious heavy-tailed regime, substantially reducing overestimation and improving the correct estimation rate.

\section{Real data analysis}\label{sec_realdata}

We apply adaptive MF to daily returns from the 49 value-weighted industry portfolios in the Kenneth French Data Library \citep{FrenchIndustry49}, using the August 2026 data vintage. These portfolios aggregate stocks by industry and allow us to examine common market movements alongside industry-specific responses. We focus on 1987Q4, which includes the October stock-market crash, and use 1987Q3 and 1988Q1 as adjacent-period comparisons. The three panels contain $n=64$, $64$, and $63$ trading days, respectively, with $p=49$ throughout. This episode is an illustration selected through exploratory analysis.

For each analysis quarter, industry returns are centered using their quarterly sample means and divided by industry-specific sample standard deviations estimated from the 125 trading days in 1987Q1--Q2. These preperiod scale estimates are held fixed across all analysis quarters. There are no missing observations, and all extreme returns are retained. We apply the ED, GR, and DDPA+ estimators from the simulation study and their adaptive MF versions, using $\kappa=0.4$. The ED and GR searches extend to index $8$; MF scans the first $10$ eigenvalues, with the multiplier half-width cap set to $\bar h=0.9$. Centering and scaling remain fixed during perturbation. All quarters use $K=2000$, with three fixed multiplier seeds on each data panel.

In 1987Q4, the first three sample eigenvalues are approximately $404.14$, $19.19$, and $10.52$. Both ED and DDPA+ select $k=3$, while GR selects $k=1$. Figure~\ref{fig:real_detection}(a) shows that $\mathsf R_1$ ranges from $0.203$ to $0.210$, below the threshold, whereas $\mathsf R_2$ ranges from $0.503$ to $0.524$, above it. The first flagged position is therefore $r^*=2$: MF reduces the ED and DDPA+ estimates from $k=3$ to $\widehat m=1$ and leaves GR unchanged. The adjacent-quarter estimates are summarized in Table~\ref{tab:real_factor_counts}.

\begin{table}[htbp]
\centering\small
\setlength{\tabcolsep}{6pt}
\caption{Initial estimates $k$ and estimates $\widehat m$ after MF (identical across three seeds).}
\label{tab:real_factor_counts}
\begin{tabular}{@{}lrr cc cc cc@{}}
\toprule
& & & \multicolumn{2}{c}{ED} & \multicolumn{2}{c}{GR} & \multicolumn{2}{c}{DDPA+}\\
\cmidrule(lr){4-5}\cmidrule(lr){6-7}\cmidrule(l){8-9}
Quarter & $n$ & $K$ & $k$ & $\widehat m$ & $k$ & $\widehat m$ & $k$ & $\widehat m$\\
\midrule
1987Q3 & 64 & 2000 & 2 & 2 & 1 & 1 & 2 & 2\\
1987Q4 & 64 & 2000 & 3 & \textbf{1} & 1 & 1 & 3 & \textbf{1}\\
1988Q1 & 63 & 2000 & 1 & 1 & 1 & 1 & 1 & 1\\
\bottomrule
\end{tabular}
\end{table}

\begin{figure}[htbp]
    \centering
    \includegraphics[width=0.8\textwidth]{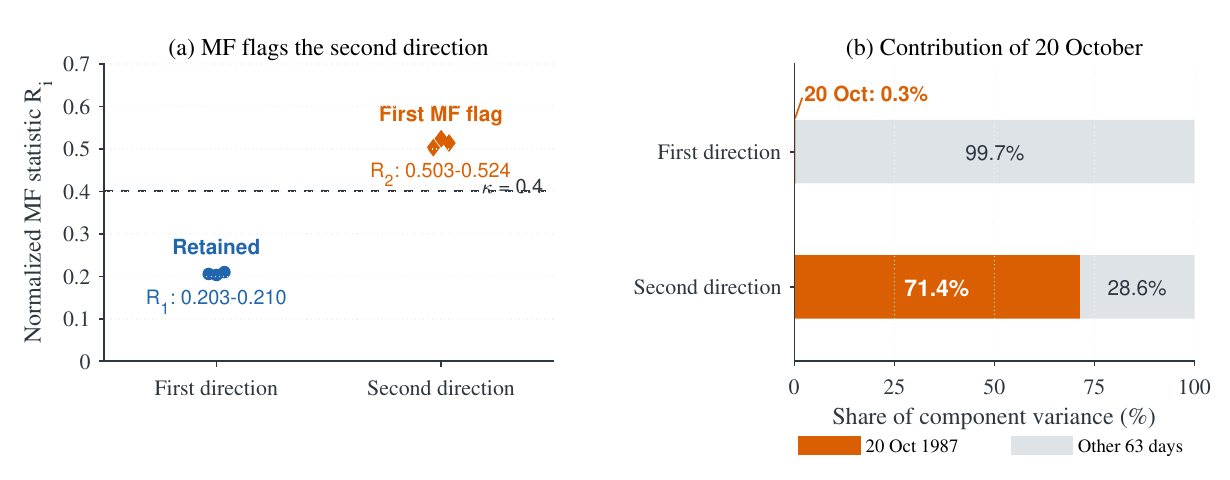}
    \caption{MF analysis of industry returns in 1987Q4. (a) $\mathsf R_i$ for $i=1,2$ across three multiplier seeds, with threshold $\kappa=0.4$. (b) Variance contributions (\%) of 20 October versus the other 63 trading days to the first two principal components.}
    \label{fig:real_detection}
\end{figure}

MF retains the first principal component: the normalized statistic satisfies $\mathsf R_1<\kappa$. The component's same-sign loadings across all 49 industries are also consistent with its interpretation as a common market factor.

By contrast, $\mathsf R_2>\kappa$, so MF identifies the second eigenvalue as the first spurious signal ($r^*=2$), yielding $\widehat m=1$. Figure~\ref{fig:real_detection}(b) helps interpret this classification: 20 October accounts for $71.4\%$ of the second component's variance, compared with $0.3\%$ for the first. This concentration on the day after Black Monday, amid documented trading and liquidity disruptions \citep{Carlson2007Crash}, is consistent with the sensitivity to multiplier perturbations reflected in $\mathsf R_2$.

\section{Conclusion}\label{sec_conclusion}
This paper has confronted the critical ``spurious factor dilemma" in EFMs, where heavy-tailed randomness, prevalent in economic and financial data, can generate noise-induced eigenvalues that masquerade as real factors. We introduce a novel theory-based fluctuation magnification algorithm, which exploits how real and spurious eigenvalues respond to random reweighting of the observations: a real factor eigenvalue averages the multipliers over many observations and remains stable, while a spurious eigenvalue, driven by a single extreme observation, inherits the full fluctuation of its multiplier. This advancement significantly enhances robust factor analysis, offering a more reliable foundation for modeling and forecasting in high-dimensional, heavy-tailed distribution, thereby improving the fidelity of economic and financial decision-making. 

Although classical RMT assumes finite higher-order moments, which heavy-tailed data often violate, this work has extended its results to matrices with heavy-tailed entries. These extensions describe how eigenvalues and eigenvectors behave in this setting and support robust methods for separating signal from noise and for inference with heavy-tailed data.

\clearpage
\phantomsection
\section*{Supplementary Material}
\addcontentsline{toc}{section}{Supplementary Material}
\setcounter{section}{0}
\renewcommand{\thesection}{S\arabic{section}}
\renewcommand{\theHsection}{supp.\arabic{section}}
\setcounter{equation}{0}
\renewcommand{\theequation}{S.\arabic{equation}}
\renewcommand{\theHequation}{supp.\arabic{equation}}
\let\theorem\suppTheorem
\let\endtheorem\endsuppTheorem
\let\assumption\suppAssumption
\let\endassumption\endsuppAssumption
\let\lemma\suppLemma
\let\endlemma\endsuppLemma
\let\remark\suppRemark
\let\endremark\endsuppRemark
\let\corollary\suppCorollary
\let\endcorollary\endsuppCorollary
\let\proposition\suppProposition
\let\endproposition\endsuppProposition
\let\example\suppExample
\let\endexample\endsuppExample
\let\definition\suppDefinition
\let\enddefinition\endsuppDefinition

This supplementary material contains the technical proofs in the paper. 

\paragraph*{Normalization convention.} Recall from Section \ref{sec_model} of the main text that $\mathbf y_t=\sqrt p\,\xi_t\Sigma^{1/2}\mathbf u_t$ with $\E\xi_t^2=1$, and that $D=\phi_n^{1/2}\diag(\xi_1,\dots,\xi_n)$ with $\phi_n=p/n$, so that $S=n^{-1}YY'=\Sigma^{1/2}UD^2U'\Sigma^{1/2}$ and $\mathcal S=DU'\Sigma UD$. To lighten notation, in Sections \ref{sec_prf_prelim}--\ref{sec_prf_fake} we write $\xi_j^2$ for the diagonal entries $d_j^2=\phi_n\xi_j^2$ of $D^2$, so that in these sections $\E\xi_j^2=\phi_n$. With this convention, we have $\frac1p\sum_{j=1}^n\phi=1$. When a result is stated in the original scale (e.g., Lemma \ref{lem_realsignal_secondratio1} and Theorem \ref{thm_realsignal_clt}), we translate at the end of the proof. In Section \ref{sec_proof_testing procedures} we return to the notation of the main text.
\section{Sketch of proof strategy}\label{sketch for proof strategy}

In this section, we outline the main ideas of the proof and full technical details are deferred to the Supplementary Material \cite{SupplementaryMaterials}. Our proof proceeds in three steps. 
\begin{itemize}
	\item 
We first show that, in an elliptical factor model, heavy-tailed randomness in the noise component can generate additional outlying eigenvalues beyond those induced by divergent factor signals. This phenomenon is established in Lemma \ref{lem_fakesignal_compare_nonspikes} and Theorem \ref{thm_fakesignal_limit}.

The starting point is that the sample covariance matrix $S$ admits a separable structure under the elliptical factor model. When the radii in $D$ exhibit heavy-tailed decay, the associated eigenvalues may diverge. Our analysis relies on a perturbative argument; however, it differs substantially from the classical treatment in spiked covariance models. In the present setting, the spectral distribution is unbounded and the eigen-gaps are not necessarily large, particularly when the heavy-tailed noise is moderate. Consequently, a more delicate local-scale analysis is required. To this end, we employ tools from random matrix theory to derive a self-consistent system and establish local laws for the sample covariance matrix $S_2$ (the component without divergent factors), as developed in Theorems \ref{lem_solutionsystem} and \ref{thm_prf_fakesignals_locallaw}. On the other hand, the elliptical structure introduces nonlinear dependence across columns, which necessitates new concentration inequalities beyond the standard i.i.d. framework. Based on these preparations, we modify the perturbation arguments by isolating $\mathbf{y}_{i,2}$ corresponding to the largest noise radius $\xi_{(1)}^2$ from the matrix $Y_2$ as in \eqref{eq_def_S2}. We introduce a real auxiliary quantity $\mu_1>0$ to be the largest solution of $1+(\xi^2_{(1)}+\mathsf{q})m_{1n}(\mu_1)=0$ with $\mathsf{q}/\xi^2_{(1)}=\mathrm{o}_{\mathbb{P}}(1)$. As can be seen in the proof of Theorem \ref{thm_fakesignal_limit}, $\mu_1$ connects with $\xi^2_{(1)}$ naturally by $\mu_1/(\xi^2_{(1)}+\mathsf{q})=\bar{\sigma}+\mathrm{o}_{\mathbb{P}}(1)$. Moreover, it can establish a connection between $\mu_1$ and $\lambda_{1,2}:=\lambda_1(S_2)$ using a modified perturbation argument  as in Theorem \ref{thm_prf_fakesignals_eigenvaluerigidity}. Specifically, as we define a determinant function $M(\cdot)$ as in \eqref{eq_defnmlambda}. $\lambda_{1,2}$ can be uniquely characterized by the equation $M(\lambda_{1,2})=0$. By our established local laws, we can further demonstrate that $1+(\xi^2_{(1)}+\mathsf{q})m_{1n}(\lambda_{1,2})\approx0$. Subsequently, a detailed continuity and stability analysis finally shows that $\lambda_{1,2}/\xi^2_{(1)}=\bar{\sigma}+\mathrm{o}_{\mathbb{P}}(1)$. Thus, the leading eigenvalues of $S_2$ inherit the divergence rate of the largest noise radii, confirming that heavy-tailed noise generates spurious outliers.

	\item We next analyze the leading eigenvalues of $S$ associated with divergent factor signals. Our approach builds on the perturbative framework developed in \cite{cai2020limiting}, suitably generalized to accommodate elliptical noise with heavy tails.

The key observation is that if the signal strengths dominate the heavy-tailed noise, the perturbative argument remains valid. A central role is played by the random quantity $\zeta_1$ defined in \eqref{eq_def_zetai}, which captures the randomness of $\lambda_1/\theta_1$ induced by the noise radii ${\xi_i^2}$. An exact expansion of $\zeta_1$ (Lemma \ref{lem_realsignal_secondratio1}) reduces $\lambda_1/\theta_1-1$ to a normalized sum of i.i.d.\ variables plus an explicit deterministic shift $b_1$, as in \eqref{eq_realsignal_representation}. If the noise radii have polynomial decay with index $\alpha>2$ or exponential decay, asymptotic normality then follows from the Lyapunov central limit theorem. If $\alpha \in (1,2]$, asymptotic normality fails, but the order $\OrderP(\sqrt{\mathsf T/n})$ of the fluctuation can still be determined by truncation. Furthermore, we show that, beyond the first $m$ spikes, the leading eigenvalues of $S$ are asymptotically close to those of $S_2$, which are driven by heavy-tailed noise. This yields a complete description of the limiting behavior of eigenvalues corresponding to both true and spurious signals.

	\item In the final step, we introduce a fluctuation magnification mechanism to distinguish genuine factors from noise-induced outliers. We reweight the observations by i.i.d.\ bounded multipliers $w_t\in[1-h,1+h]$. Two elementary facts drive the analysis. First, the deterministic sandwich $(1-h)S\preceq\widetilde S\preceq(1+h)S$ makes all ratios $\lambda_k^{(j)}/\lambda_k$ bounded, which yields a non-asymptotic Monte Carlo error of order $K^{-1/2}$. Second, unconditionally the magnified data are again EFM data with radii $\xi_t^2w_t$, so that the results of Section \ref{sec_mainresults} apply to them.

Conditionally on $Y$, a genuine eigenvalue behaves like a weighted average $\sum_tw_tz_{tk}$ with weights summing to one. A block perturbation argument (Ostrowski's theorem and a Schur-complement bound) shows that its relative variance is at most of order the maximal leverage $\max_t\ell_t=\OrderP(\mathsf T\log n/n+\mathsf T/\sigma_m)\to0$. In contrast, by Lemma \ref{lem_fakesignal_compare_nonspikes} and Theorem \ref{thm_fakesignal_limit} applied to the magnified data, the leading spurious eigenvalue tracks $\max_t\xi_t^2w_t$. This maximum equals $\xi^2_{(1)}w_{t^*}$ whenever the multiplier preserves the order of the two largest radii, so the relative variance converges to the known variance of the multiplier. A coupling argument makes this uniform in $h$, which allows $h$ to be chosen from the data. Normalizing by the multiplier variance yields a statistic whose limit is $0$ for genuine factors and $1$ for the leading spurious eigenvalue. This is the basis of our testing procedure.
\end{itemize}

\section{Preliminary results}\label{sec_prf_prelim}
In this section, we collect several useful definitions and tools. 
Throughout this article, we use the following notion to provide a simple way of systematizing and making precise statements for two families of random variables $A,B$
	of the form ``$A$ is bounded with high probability by $B$ up to
	small powers of $n$'' .
	\begin{definition}[Stochastic domination]\mbox{}{\newline}
		(a). For two families of nonnegative random variables
		\[
		A=\{A_{n}(t):n\in\mathbb{Z}_{+},t\in T_{n}\},\qquad B=\{B_{n}(t):n\in\mathbb{Z}_{+},t\in T_{n}\},
		\]
		where $T_{n}$ is
		a possibly $n$-dependent parameter set, we say that $A$ is stochastically
		dominated by $B$, uniformly in $t$ if for all (small) $\varepsilon>0$
		and (large) $D>0$ there exists $n_{0}(\varepsilon,D)\in\mathbb{Z}_{+}$
		such that 	as $n\ge n_{0}(\varepsilon,D)$,
		\[
		\sup_{t\in T_{n}}\mathbb{P}\big(A_{n}(t)>n^{\varepsilon}B_{n}(t)\big)\le n^{-D}.
		\]
		If $A$ is stochastically dominated
		by $B$, uniformly in $t$, we use notation $A\prec B$ or $A=O_{\prec}(B)$.
		Moreover, for some complex family $A$ if $|A|\prec B$ we also write $A=O_{\prec}(B)$. \\
		(b). Let $A$ be a family of random matrices and $\zeta$ be a family of nonnegative random variables. Then we denote $A=O_{\prec}(\zeta)$ if $A$ is dominated under weak operator norm sense, i.e. $|\langle\mathbf{v},A\mathbf{w}\rangle|\prec\zeta\|\mathbf{v}\|_2\|\mathbf{w}\|_2$  for any deterministic vectors $\mathbf{v}$ and $\mathbf{w}$.\\
		(c). For two sequences of numbers $\{a_{n}\}_{n=1}^{\infty}$,
		$\{b_{n}\}_{n=1}^{\infty}$, $a_{n}\prec b_{n}$ if for all $\epsilon>0$, $a_n\leq n^{\epsilon}b_n$.
	\end{definition}

         We also use the following definition to control random events,
         \begin{definition}[High probability event]
		We say that an $n$-dependent event $\Omega$ holds with high probability if there exists constant $c>0$ independent of $n$, such that
		\begin{equation}
		\mathbb{P}(\Omega)\geqslant 1-\exp{(-n^c)},
		\end{equation}
		for all sufficiently large $n$.
	\end{definition}

 Now, we give a description of the structure of $D^2$ firstly. Recall Assumption \ref{ass_xi} and $D^2=\operatorname{diag}(\xi_1^2,\dots,\xi^2_n)$. We define the order statistics of $\{\xi^2_i\}$ as $\xi^2_{(1)}\ge\xi^2_{(2)}\ge\dots\ge\xi^2_{(n)}$. We say that $\xi^2_{(i)}$'s are of ``good configuration" if they satisfy the following definition.
        \begin{definition}\label{def_goodconfiguration}
            Let $\Omega \equiv \Omega_n$ be the event on $\{\xi_i^2\}$ so that the following conditions hold:
            \begin{itemize}
                \item [] {\bf (a). Polynomial decay.} When $\{\xi_i^2\}$ has polynomial decay tail as in \eqref{ass_xi_poly}, we assume that for every fixed $b\in(1/2,1]$  and some constants $C,C_1,C_2, c>1,$ the following holds on $\Omega$
        \begin{equation}
        \begin{aligned}
        & \xi^2_{(1)}-\xi^2_{(2)}\ge C_1^{-1}n^{1/\alpha}\log^{-C}n, \  C_2^{-1}n^{1/\alpha}\log^{-1}n\le\xi^2_{(1)}\le C_2n^{1/\alpha}\log n,  \\
        &  \xi^2_{(1)}-\xi^2_{(\lceil n^b \rceil)}\ge c^{-1}n^{1/\alpha}\log^{-1} n, \ 
          \frac{1}{n}\sum_{i=1}^n\xi^2_i< \infty. 
        \end{aligned}
        \end{equation} 
        \item[] {\bf (b). Exponential decay.} When $\{\xi_i^2\}$ has exponential decay tail as in \eqref{ass_xi_exp}, we assume that for some constant $C_1,C_2>1, C_3>C_1$, $0<c<\min\{s(C_1^{-1}-C_3^{-1})^{\beta},1\}$ and small constant $\epsilon>0$,  
        the following holds on $\Omega$
        \begin{equation}
        \begin{aligned}
        & \xi^2_{(1)}-\xi^2_{(2)}\ge C_1^{-1}\log^{1/\beta-1-\epsilon} n, \   C_2^{-1}\log^{1/\beta} n\le\xi^2_{(1)}\le C_2\log^{1/\beta} n,\\
        &\xi^2_{(1)}-\xi^2_{(\lceil n^{1-c}\rceil)}\ge C_3^{-1}\log^{1/\beta}n,\   \frac{1}{n}\sum_{i=1}^n\xi^2_i<\infty. 
        \end{aligned}
        \end{equation}
            \end{itemize}
        \end{definition}

        The following lemma shows that under Assumption \ref{ass_xi}, the probability event $\Omega$ happens with high probability.
        \begin{lemma}\label{lem_goodconfiguration}
            Let $\Omega$ be the event defined in Definition \ref{def_goodconfiguration}. Under Assumption \ref{ass_xi} holds, we have for $n$ sufficiently large that 
            \begin{equation*}
                \mathbb{P}(\Omega)=1-\mathrm{o}(1).
            \end{equation*}
        \end{lemma}
        \begin{remark}
            We remark that on the event $\Omega$, according to $(a)$ and $(b)$ in Definition \ref{def_goodconfiguration}, we see that the first few largest $\xi^2_i$ are divergent and well separated from each other. Second, the proof of Lemma \ref{lem_goodconfiguration} can be found in \cite[Lemma A.15]{xie2026multiplier}. We do not intend to repeat it here for simplification. Third, arguments carried out on $\Omega\cap\Xi$, where $\Xi$ is an event concerning $U$ that holds with high probability (conditionally on the radii), remain valid because $1-\mathbb P(\Omega\cap\Xi)=\mathrm o(1)$. Fourth, working on $\Omega$ amounts to truncating the radii at level $C\mathsf T$. Consequently, moment computations carried out on $\Omega$ (e.g., in the proofs of Theorems \ref{thm_realsignal_clt} and \ref{thm_realsignal_extremeheavy}) concern the truncated radii. 
        \end{remark}
        
        The following large deviation inequality describes the concentration for certain kinds of random variables.
        \begin{lemma}[Large deviation bounds]\label{lem_largedeviationbound}
Let $\mathbf{u}=(u_1, u_2, \cdots, u_p)^{\prime}, \widetilde{\mathbf{u}}=(\widetilde{u}_1, \widetilde{u}_2, \cdots, \widetilde{u}_p)^{\prime}\in \mathbb{R}^p$  be two independent real random vectors. Moreover, let $A$ be a $p \times p$ matrix independent of the above vectors. Then the following holds. 
\begin{enumerate}
\item[(1).] When the entries of the random vectors are centered i.i.d. random variables with variance $p^{-1}$ and  $\mathbb{E}|\sqrt{p} v_{i}|^k \leq C_k,$ where $v_i=u_i, \widetilde{u}_i, 1 \leq i \leq p,$  we have that 
		\begin{align*}
		|\widetilde{\mathbf{u}}^{*}\mathbf{u}|  \prec  \sqrt{\frac{\|\mathbf{u}\|^{2}}{p},} \ \ \left|\mathbf{u}^{*}A\tilde{\mathbf{u}}\right| \prec  \frac{1}{p}\|A\|_{F}, \ \
		|\mathbf{u}^{*}A\mathbf{u}-\frac{1}{p}{\rm Tr}A|  \prec  \frac{1}{p}\|A\|_{F}. 		
		\end{align*}
		\item[(2).] When the random vectors	are sampled from  $\mathrm{U}(\mathbb{S}^{p-1}),$ we have that 
				\begin{align*}
		|\widetilde{\mathbf{u}}^{*}\mathbf{u}|  \prec  \sqrt{\frac{\|\mathbf{u}\|^{2}}{p},} \ \ \left|\mathbf{u}^{*}A\tilde{\mathbf{u}}\right| \prec  \frac{1}{p}\|A\|_{F}, \ \
		|\mathbf{u}^{*}A\mathbf{u}-\frac{1}{p}{\rm Tr}A|  \prec  \frac{1}{p}\|A\|_{F}. 		
		\end{align*}
\end{enumerate}
\end{lemma}

\section{Proof of the results in Section \ref{sec_main_realsignals}}
\subsection{Proof of Theorem \ref{lem_realsignal_preratio}}
We only show the result of $\lambda_1$ while other cases can be handled similarly. Recall the decomposition
\begin{equation*}
    \mathcal{S}=DU^{\prime}\Sigma_1UD+DU^{\prime}\Sigma_2UD:=\mathcal{S}_1+\mathcal{S}_2.
\end{equation*}
Since $\|\Sigma_2\|\le c^{-1}$ from Assumption \ref{ass_sigma}, we have 
\begin{equation}
\|\mathcal{S}_2\|\le\|U^{\prime}\Sigma_2U\|\cdot\|D^2\|\le\mathrm{O}(\xi^2_{(1)}),
\end{equation}
where we borrowed the result from \cite{Wen2021} with $\|U^{\prime}\Sigma_2U\|\le C$ for some constant $C>0$. Beside, notice that $\xi^2_{(1)}\lesssim\mathsf{T}$ by Lemma \ref{lem_goodconfiguration}. Then, from Assumption \ref{ass_sigma} and Weyl's theorem, we can conclude that 
\begin{equation}\label{eq_prf_realsignals_preratio_1}
    \big|\frac{\lambda_1}{\sigma_1}-\frac{\lambda_1(\mathcal{S}_1)}{\sigma_1}\big|\le\frac{\|\mathcal{S}_2\|}{\sigma_1}=\mathrm{O}_{\mathbb{P}}(\frac{\mathsf{T}}{\sigma_1}).
\end{equation}
It remains is to consider $\lambda_1(\mathcal{S}_1)$, or equivalently $\lambda_1(S_1)$ where $S_1:=\Sigma^{1/2}_1UD^2U^{\prime}\Sigma^{1/2}_1$. Denote $\mathscr{H}(x):=x\Lambda_{\Sigma}^{-1}-\Gamma_1^{\prime}UD^2U^{\prime}\Gamma_1$, where we decompose $\Sigma_1=\Gamma_1\Lambda_{\Sigma}\Gamma_1^{\prime}$ and $\Lambda_{\Sigma}=\operatorname{diag}(\sigma_1,\dots, \sigma_m)\in\mathbb{R}^{m\times m}$, $\Gamma_1=(\mathbf{e}_1,\dots, \mathbf{e}_m)\in\mathbb{R}^{p\times m}$ consists of $m$ orthogonal basis vectors, say $\mathbf{e}_i:=(0,\dots, 0,1,0,\dots,0)^{\prime}\in\mathbb{R}^p$ with $i$-th position being one and others being zero. By definition, $\operatorname{det}\big[\mathscr{H}\big(\lambda_1(S_1)\big)\big]=0$. Note that $\Gamma_1^{\prime}UD^2U^{\prime}\Gamma_1$ is an $m\times m$ matrix with diagonal entries satisfying
\begin{align*}
        &\mathbf{e}_i^{\prime}UD^2U^{\prime}\mathbf{e}_i=1+\frac{1}{p}\sum_i(\xi^2_i-\phi_n)+\sum_k(u_{ik}^2-\frac{1}{p})\xi^2_k\\
        &=1+\frac{1}{p}\sum_i(\xi^2_{i}-\phi_n)+\mathrm{O}_{\mathbb{P}}\Big(\frac{1}{\sqrt{p}}+\sqrt{\frac{\mathsf{T}}{n}}\cdot\mathbbm{1}(\alpha\in(1,2])\Big)\\
        &=1+\mathrm{O}_{\mathbb{P}}\Big(\frac{1}{\sqrt{n}}+\sqrt{\frac{\mathsf{T}}{n}}\cdot\mathbbm{1}(\alpha\in(1,2])\Big),\quad 1\le i\le m,
\end{align*}
where we used Lemma \ref{lem_largedeviationbound} and we used the notation $\mathbbm{1}(\alpha\in(1,2])$ to denote the case \eqref{ass_xi_poly} with $\alpha\in(1,2]$ in Assumption \ref{ass_xi}.  Assumption \ref{ass_xi} and basic observations that
\begin{equation*}
    \mathbb{E}\big[\frac{1}{p}\sum_{i}(\xi^2_i-\phi_n)\big]^2=\frac{1}{p^2}\sum_i\mathbb{E}(\xi^2_i-\phi_n)^2\lesssim\frac{2}{p}\sum_{i}\mathbb{E}\xi^2_i\times\mathrm{O}(\frac{\mathsf{T}}{p}).
\end{equation*}

Parallel, we have the off-diagonal entries 
\begin{equation*}
    \mathbf{e}_i^{\prime}UD^2U^{\prime}\mathbf{e}_j=\mathrm{O}_{\mathbb{P}}\Big(\frac{1}{\sqrt{n}}+\sqrt{\frac{\mathsf{T}}{n}}\cdot\mathbbm{1}(\alpha\in(1,2])\Big),\quad 1\le i\neq j\le m.
\end{equation*}
Then, define the error $\mathsf{r}_n:=\frac{1}{\sqrt{n}}+\sqrt{\frac{\mathsf{T}}{n}}\cdot\mathbbm{1}(\alpha\in(1,2])$. The matrix $\mathscr{H}(x)$ can be written as 
\begin{equation*}
    \begin{pmatrix}
        \dfrac{x}{\sigma_1}-1+\mathrm{O}_{\mathbb{P}}(\mathsf{r}_n)\quad &\mathrm{O}_{\mathbb{P}}(\mathsf{r}_n)\quad &\dots \quad &\mathrm{O}_{\mathbb{P}}(\mathsf{r}_n)\\
        \mathrm{O}_{\mathbb{P}}(\mathsf{r}_n)\quad &\dfrac{x}{\sigma_2}-1+\mathrm{O}_{\mathbb{P}}(\mathsf{r}_n)\quad &\dots \quad &\mathrm{O}_{\mathbb{P}}(\mathsf{r}_n)\\
        \vdots\quad &\vdots \quad &\ddots \quad &\vdots\\
        \mathrm{O}_{\mathbb{P}}(\mathsf{r}_n)\quad &\mathrm{O}_{\mathbb{P}}(\mathsf{r}_n)\quad &\cdots \quad &\dfrac{x}{\sigma_m}-1+\mathrm{O}_{\mathbb{P}}(\mathsf{r}_n)
    \end{pmatrix}.
\end{equation*}
Let $x=\sigma_1(1-c)$ for some small $c>0$. Since $\sigma_1/\sigma_i>1$ for $i=2,\dots, m$, we find that 
\begin{equation*}
    \frac{x}{\sigma_i}=\frac{x}{\sigma_1}\cdot\frac{\sigma_1}{\sigma_i}>1,
\end{equation*}
for sufficiently small $c>0$. This implies that $\operatorname{det}[\mathscr{H}(x)]<0$ with probability tending to one. On the other hand, if $x=\sigma_1(1+c)$ for some sufficiently small $c>0$, we can also conclude that $\operatorname{det}[\mathscr{H}(x)]>0$ with probability tending to one. Therefore, there must be an eigenvalue in $[\sigma_1(1-c),\sigma_1(1+c)]$, which is $\lambda_1(S_1)$ (or $\lambda_1(\mathcal{S}_1)$). Further, by Leibniz's formula for determinant, we have
\begin{equation*}
    \operatorname{det}\big[\mathscr{H}\big(\lambda_1(S_1)\big)\big]=\frac{\lambda_1(\mathcal{S}_1)}{\sigma_1}-1+\mathrm{O}_{\mathbb{P}}\Big(\frac{1}{\sqrt{n}}+\sqrt{\frac{\mathsf{T}}{n}}\cdot\mathbbm{1}(\alpha\in(1,2])\Big)=0.
\end{equation*}
That is, $\lambda_1(\mathcal{S}_1)/\sigma_1=1+\mathrm{O}_{\mathbb{P}}\Big(\frac{1}{\sqrt{n}}+\sqrt{\frac{\mathsf{T}}{n}}\cdot\mathbbm{1}(\alpha\in(1,2])\Big)$. Similarly, we can conclude that $\lambda_i(\mathcal{S}_1)/\sigma_i=1+\mathrm{O}_{\mathbb{P}}\Big(\frac{1}{\sqrt{n}}+\sqrt{\frac{\mathsf{T}}{n}}\cdot\mathbbm{1}(\alpha\in(1,2])\Big)$ for any $2\le i\le m$. Combined with \eqref{eq_prf_realsignals_preratio_1}, we conclude this theorem.

\subsection{Proof of Lemma \ref{lem_realsignal_secondratio1}}
Throughout the proof we work on $\Omega$, and $\xi_j^2$ denote the rescaled radii (with mean $\mathbb{E}\xi_j^2=\phi_n$; see the normalization convention). Fix $1\le i\le m$ and set
\begin{equation*}
x_i:=\frac{1}{p\theta_i}\sum_{k=1}^{p-m}\frac{\sigma_{m+k}}{1-\sigma_{m+k}/\sigma_i},\qquad
y_i:=\frac{1}{p\theta_i}\sum_{k=1}^{p-m}\frac{\sigma_{m+k}}{1-\sigma_{m+k}\zeta_i/\theta_i}.
\end{equation*}
By \eqref{eq_def_theta/sigma}, $\theta_i/\sigma_i=(1-x_i)^{-1}$. Since $\sigma_{m+k}\le c_1$ and $\sigma_i\to\infty$, a fixed-point argument on $[\sigma_i,2\sigma_i]$ gives existence and uniqueness of $\theta_i$ for large $n$. Moreover $x_i=\bar\sigma\theta_i^{-1}(1+\mathrm O(\sigma_i^{-1}))$, so
\[\theta_i/\sigma_i-1=x_i/(1-x_i)=\mathrm O(\tr\Sigma_2/(p\sigma_i)),
\] 
which is the first claim.

On $\Omega$, we have $p^{-1}\tr D^2=\mathrm O(1)$ and $\max_j\xi_j^2\le C\mathsf T=\mathrm o(\sigma_i)$. The existence and uniqueness of $\zeta_i\in[p^{-1}\tr D^2,2p^{-1}\tr D^2]$ then follow from the intermediate value theorem and monotonicity, as in \cite{yu2024testing}. Since $\zeta_i=\mathrm O(1)$, we also have $y_i=\bar\sigma\theta_i^{-1}(1+\mathrm O(\sigma_i^{-1}))$ and $\max_j\xi_j^2y_i\le C\mathsf T/\sigma_i=\mathrm o(1)$.

By \eqref{eq_def_zetai} and the identity $p^{-1}\sum_{j=1}^n\phi_n=1$,
\begin{equation*}
\zeta_i-\frac{\theta_i}{\sigma_i}=\frac1p\sum_{j=1}^n\Big(\frac{\xi_j^2}{1-\xi_j^2y_i}-\frac{\phi_n}{1-x_i}\Big).
\end{equation*}
Using 
\[
\frac{a}{1-ay}=a+a^2y+\frac{a^3y^2}{1-ay}
\]
for $a=\xi_j^2$, $y=y_i$, and 
\[
\frac{\phi_n}{1-x}=\phi_n+\phi_n x+\frac{\phi_n x^2}{1-x}
\]
for $x=x_i$, we obtain the exact decomposition
\begin{equation}\label{eq_prf_zeta_exact}
\zeta_i-\frac{\theta_i}{\sigma_i}=\frac1p\sum_{j=1}^n(\xi_j^2-\phi_n)+y_i\,\frac1p\sum_{j=1}^n\xi_j^4-x_i+R_i,\quad
R_i:=\frac1p\sum_{j=1}^n\frac{\xi_j^6y_i^2}{1-\xi_j^2y_i}-\frac{x_i^2}{1-x_i}.
\end{equation}
Since $\max_j\xi_j^2y_i=\mathrm o(1)$ and $x_i,y_i=\mathrm O(\sigma_i^{-1})$, we have $|R_i|\le 2y_i^2p^{-1}\sum_j\xi_j^6+2x_i^2=\mathrm O\big(\sigma_i^{-2}(1+p^{-1}\sum_j\xi_j^6)\big)$. Furthermore,
\begin{equation*}
y_i\,\frac1p\sum_j\xi_j^4-x_i=\frac{\bar\sigma}{\theta_i}\Big(\frac1p\sum_j\xi_j^4-1\Big)+\mathrm O\Big(\frac{1}{\sigma_i^{2}}\Big(1+\frac1p\sum_j\xi_j^4\Big)\Big).
\end{equation*}
Finally, we return to the original radii $\tilde\xi_j^2=\xi_j^2/\phi_n$, which have mean one. In terms of these, 
\[
p^{-1}\sum_j(\xi_j^2-\phi_n)=n^{-1}\sum_j(\tilde\xi_j^2-1),\quad p^{-1}\sum_j\xi_j^4=\phi_n\,n^{-1}\sum_j\tilde\xi_j^4,
\]
and
\[
p^{-1}\sum_j\xi_j^6=\phi_n^2n^{-1}\sum_j\tilde\xi_j^6.
\]
Using also $x^2\le1+x^3$ for $x\ge0$, we obtain \eqref{eq_zeta_expansion}. Since $\mathbb P(\Omega)\to1$, the $\mathrm O$-terms become $\OrderP$-terms unconditionally.

\subsection{Proof of Theorem \ref{thm_realsignal_secondratio2}}
In this subsection, we show the limiting representation of $\lambda_i/\theta_i$. We use the rescaled radii of the normalization convention, so that $\E\xi_j^2=\phi_n$ and $p^{-1}\sum_{j=1}^n\phi_n=1$, and write
\begin{gather*}
\varrho_n:=n^{-1/2}+\sqrt{\mathsf T/n}\,\mathbbm 1(\alpha\in(1,2]),\qquad \omega_{n,i}:=\varrho_n+\frac{1+m_4}{\theta_i},\\
m_4:=\frac1p\sum_{j=1}^n\xi_j^4,\qquad m_6:=\frac1p\sum_{j=1}^n\xi_j^6.
\end{gather*}
Since $m_6=\phi_n^2n^{-1}\sum_j\tilde\xi_j^6$ in the original radii $\tilde\xi_j^2=\xi_j^2/\phi_n$, the remainder $\OrderP(\sigma_i^{-2}(1+m_6))$ below is that of \eqref{eq_secondorder_green}. Note that $\omega_{n,i}=\orderP(1)$ when $\alpha>2$ or \eqref{ass_xi_exp} holds, then $m_4$ converges to $\phi_n\E\tilde\xi^4<\infty$ almost surely. If $\alpha\in(1,2]$, then $m_4\le\xi^2_{(1)}p^{-1}\sum_j\xi_j^2\le C\mathsf T=\mathrm o(\theta_i)$ on $\Omega$. Fix $1\le i\le m$. The argument follows the perturbation framework of \cite{xie2026limiting,yu2024testing}, we only illustrate the steps relevant to the heavy-tailed radii enter and the remainder.

\textbf{Step 1: determinant equation.} Let $\mathcal F_n:=\Omega\cap\Xi_0\cap\{\lambda_i\ge\theta_i/2\}\cap\{\|\mathcal S_2\|\le\theta_i/4\}$, with $\Xi_0$ as in the proof of Lemma \ref{lem_prf_realsignals_secondratio2_1}. Then $\mathbb P(\mathcal F_n)\to1$ by Lemma \ref{lem_goodconfiguration}, \cite{Wen2021}, Theorem \ref{lem_realsignal_preratio}, Lemma \ref{lem_realsignal_secondratio1}, and the bound $\|\mathcal S_2\|\le C\xi^2_{(1)}\le C\mathsf T$ on $\Omega\cap\Xi_0$. 

Next, we work on $\mathcal F_n$. For $x\in\{\lambda_i,\theta_i\}$, $\mathscr K(x)=(I-x^{-1}\mathcal S_2)^{-1}$ is well defined, symmetric, commutes with $\mathcal S_2$, and
\begin{equation}\label{eq_green_K}
\|\mathscr K(x)\|\le2,\quad \mathscr K(x)-I=x^{-1}\mathcal S_2^{1/2}\mathscr K(x)\mathcal S_2^{1/2}.
\end{equation}
As in \cite{xie2026limiting,yu2024testing}, Sylvester's determinant identity and the resolvent identity
\begin{equation}\label{eq_prf_realsignals_secondratio2_matrixinverseformula}
(\lambda_iI-\mathcal S_2)^{-1}=\frac{\mathscr K(\theta_i)}{\theta_i}-\frac{\lambda_i-\theta_i}{\lambda_i\theta_i}\mathscr K(\lambda_i)\mathscr K(\theta_i)
\end{equation}
together with $\Gamma_1=(\mathbf e_1,\dots,\mathbf e_m)$ and $\Lambda_{\Sigma}=\diag(\sigma_1,\dots,\sigma_m)$ imply that $\det\mathscr M=0$, where
\begin{equation}\label{eq_green_M}
\mathscr M:=\theta_i\Lambda_{\Sigma}^{-1}-H+\delta_iG,\quad H:=\Gamma_1'UD\mathscr K(\theta_i)DU'\Gamma_1,\quad G:=\Gamma_1'UD\mathscr K(\lambda_i)\mathscr K(\theta_i)DU'\Gamma_1,
\end{equation}
and $\widetilde{\mathcal{D}}_i:=(\lambda_i-\theta_i)/\lambda_i$. The matrices $H$, $G$ and $\mathscr M$ are symmetric. Note that the coefficient of $G$ is $\widetilde{\mathcal{D}}_i=\mathcal{D}_i/(1+\mathcal{D}_i)$ with $\mathcal{D}_i:=(\lambda_i-\theta_i)/\theta_i$.

\textbf{Step 2: quadratic bounds.} Let $Q:=\Gamma_1'UD^2U'\Gamma_1-(p^{-1}\tr D^2)I_m$.
\begin{lemma}\label{lem_prf_realsignals_secondratio2_1}
Under Assumptions \ref{ass_sigma} and \ref{ass_xi}, for $1\le i\le m$, $\|H-Q-\zeta_iI_m\|_\infty=\orderP(\varrho_n)$.
\end{lemma}
\begin{lemma}\label{lem_prf_realsignals_secondratio2_2}
Under Assumptions \ref{ass_sigma} and \ref{ass_xi}, for $1\le i\le m$, $\|Q\|_\infty=\OrderP(\varrho_n)$ and $\|G-I_m\|=\OrderP(\omega_{n,i})$.
\end{lemma}
The proofs of Lemma \ref{lem_prf_realsignals_secondratio2_1} and Lemma \ref{lem_prf_realsignals_secondratio2_2} are given in Sections \ref{subsec_proof_realsignals_secondratio2_1} and \ref{subsec_proof_realsignals_secondratio2_2}. Write $H=Q+\zeta_iI_m+E_H$ and $G=I_m+E_G$, the above two lemmas imply $\|E_H\|_\infty=\orderP(\varrho_n)$ and $\|E_G\|=\OrderP(\omega_{n,i})$.

\textbf{Step 3: reduction to the $(i,i)$ entry.} Denote $A_i:=Q_{ii}+\zeta_i-\theta_i/\sigma_i$ be the leading term in \eqref{eq_secondorder_green}. We first show that $A_i=\OrderP(\omega_{n,i})$. 

By \eqref{eq_prf_zeta_exact}, 
\[
\zeta_i-\theta_i/\sigma_i=p^{-1}\sum_j(\xi_j^2-\phi_n)+y_im_4-x_i+R_i,\quad x_i,y_i\in[0,C\theta_i^{-1}]
\]
 where  $|R_i|\le C\theta_i^{-2}(1+m_6)$. On $\Omega$, 
 \[
 m_6\le\xi^2_{(1)}m_4\le C\mathsf Tm_4
 \]
this gives
\[
|\zeta_i-\theta_i/\sigma_i|\le|p^{-1}\sum_j(\xi_j^2-\phi_n)|+C(1+m_4)/\theta_i.
\]
Moreover, by Chebyshev's inequality for $\Var(\tilde\xi^2)<\infty$, and by the truncation argument in the proof of Theorem \ref{thm_realsignal_extremeheavy} when $\alpha\in(1,2]$, we have
\[
p^{-1}\sum_j(\xi_j^2-\phi_n)=\OrderP(\varrho_n).
\]
Together with $Q_{ii}=\OrderP(\varrho_n)$, this proves
\[
A_i=\OrderP(\omega_{n,i}).
\]
In particular, it implies that $\zeta_i\to1$ in probability.

Next, by \eqref{eq_green_M},
\begin{equation*}
\mathscr M_{ii}=-A_i-(E_H)_{ii}+\delta_i\big(1+(E_G)_{ii}\big),\quad \mathscr M_{kl}=-Q_{kl}-(E_H)_{kl}+\delta_i(E_G)_{kl}=\OrderP(\omega_{n,i})\quad(k\ne l),
\end{equation*}
where we used $|\widetilde{\mathcal{D}}_i|\le1$ on $\mathcal E_n$. For $k\ne i$, since $\widetilde{\mathcal{D}}_i=\orderP(1)$ by Theorem \ref{lem_realsignal_preratio}, we have
\[
\mathscr M_{kk}=\theta_i/\sigma_k-\zeta_i+\orderP(1)
\]
By Assumption \ref{ass_sigma}(2), 
\[
|\theta_i/\sigma_k-1|\ge c_0/2
\]
with $c_0:=c_2/(2+c_2)$, so $|\mathscr M_{kk}|\ge c_0/4$ with probability tending to one.

If $m=1$, then $\mathscr M_{ii}=0$. If $m\ge2$, let $\mathscr M_{-i}$ be $\mathscr M$ with its $i$-th row and column removed, and $\mathbf q_i$ the $i$-th column of $\mathscr M$ without its $i$-th entry. By Gershgorin's theorem, 
\[
\|\mathscr M_{-i}^{-1}\|\le8/c_0
\]
with probability tending to one. Using the Schur complement formula
\[
\det\mathscr M=\det(\mathscr M_{-i})(\mathscr M_{ii}-\mathbf q_i'\mathscr M_{-i}^{-1}\mathbf q_i)
\]
and $\det\mathscr M=0$ then give
\begin{equation}\label{eq_green_Mii}
|\mathscr M_{ii}|\le\frac{8}{c_0}\|\mathbf q_i\|^2=\OrderP(\omega_{n,i}^2).
\end{equation}
This replaces the Leibniz expansion used in \cite{xie2026limiting} and keeps the quadratic order of the error explicit.

\textbf{Step 4: solving for $\mathcal{D}_i$.} By Step 3,
\begin{equation}\label{eq_green_delta}
\widetilde{\mathcal{D}}_i\big(1+(E_G)_{ii}\big)=A_i+(E_H)_{ii}+\OrderP(\omega_{n,i}^2).
\end{equation}
Firstly, since $(E_G)_{ii}=\orderP(1)$, we have $|1+(E_G)_{ii}|\ge1/2$ with probability tending to one. With $A_i=\OrderP(\omega_{n,i})$, we have $E_H=\orderP(\varrho_n)$ and $\omega_{n,i}^2\le\omega_{n,i}$. Therefore, it follows that $\widetilde{\mathcal{D}}_i=\OrderP(\omega_{n,i})$. Consequently, we obtain $\mathcal{D}_i=\widetilde{\mathcal{D}}_i/(1-\widetilde{\mathcal{D}}_i)=\OrderP(\omega_{n,i})$.

Secondly, moving $\widetilde{\mathcal{D}}_i(E_G)_{ii}=\OrderP(\omega_{n,i}^2)$ to the right-hand side of \eqref{eq_green_delta}, and using $\mathcal{D}_i-\widetilde{\mathcal{D}}_i=\mathcal{D}_i\widetilde{\mathcal{D}}_i=\OrderP(\omega_{n,i}^2)$, we obtain
\begin{equation*}
\mathcal{D}_i=A_i+\orderP(\varrho_n)+\OrderP(\omega_{n,i}^2).
\end{equation*}

Finally, notice
\[
\omega_{n,i}^2\le2\varrho_n^2+2(1+m_4)^2/\theta_i^2
\]
and $\varrho_n^2=\mathrm o(\varrho_n)$. By the Cauchy--Schwarz inequality, on $\Omega$,
\[
m_4^2\le(p^{-1}\sum_j\xi_j^2)(p^{-1}\sum_j\xi_j^6)\le Cm_6,
\]
so that by $\theta_i\ge\sigma_i$
\[
(1+m_4)^2/\theta_i^2\le C(1+m_6)/\sigma_i^2
\]
 Therefore,
 \[
 \lambda_i/\theta_i-1=\mathcal{D}_i=A_i+\OrderP(\sigma_i^{-2}(1+m_6))+\orderP(\varrho_n),
 \]
which is \eqref{eq_secondorder_green}. This concludes the proof of Theorem \ref{thm_realsignal_secondratio2}.

\subsection{Proof of Theorem \ref{thm_realsignal_clt}}\label{subsec_proof_realsignal_clt}
We use the original radii of $\xi^2$ with mean one in this proof, and write $Z_t:=p\xi_t^2u_{it}^2$. By Theorem \ref{thm_realsignal_secondratio2} and \eqref{eq_zeta_expansion}, the representation \eqref{eq_realsignal_representation} holds. Writing its second term as $b_i+E_1$, it reads
\begin{equation*}
\frac{\lambda_i}{\theta_i}-1-b_i=\frac1n\sum_{t=1}^n(Z_t-1)+E_1+\mathcal E_i,\quad E_1:=\frac{\bar\sigma\phi_n}{\theta_i}\,\frac1n\sum_{t=1}^n(\xi_t^4-\E\xi^4).
\end{equation*}

\textbf{Leading term.} Notice that $Z_t$, $1\le t\le n$, are i.i.d., because the $\mathbf u_t$ are i.i.d.\ and independent of the i.i.d.\ radii. Since $\E u_{it}^2=1/p$ and $\E u_{it}^4=3/(p(p+2))$, we have 
\[
\E Z_t=1,\quad \operatorname{Var}(Z_t)=\frac{3p}{p+2}\E\xi^4-1\to3\E\xi^4-1.
\]
For $\alpha>2$, choose $\delta>0$ with $2+\delta<\alpha$. Then,
\[
\E|Z_t|^{2+\delta}\le\E\xi^{2(2+\delta)}\sup_p\E(pu_{it}^2)^{2+\delta}<\infty
\]
and the same holds under \eqref{ass_xi_exp}. By the Lyapunov central limit theorem for triangular arrays \citep[Theorem 27.3]{billingsley1995probability}, 
\[
\sqrt n\cdot n^{-1}\sum_t(Z_t-1)\overset{d}{\to}\mathrm N(0,3\E\xi^4-1).
\]

\textbf{The term $E_1$.} If $\E\xi^8<\infty$ (in particular if $\alpha>4$ or \eqref{ass_xi_exp} holds), then
\[
E_1=\OrderP(\sigma_i^{-1}n^{-1/2})=\orderP(n^{-1/2}).
\]
 If $\alpha\in(2,4]$, then $\E|\xi^4|^r<\infty$ for every $r<\alpha/2$, where $\alpha/2\in(1,2]$. By the Marcinkiewicz--Zygmund strong law \citep{gut2013probability}, 
 \[
 n^{-1}\sum_t(\xi_t^4-\E\xi^4)=\mathrm o(n^{1/r-1})
 \]
 almost surely. Since $\sigma_i\gg\mathsf T\ge n^{1/\alpha}$, it follows that 
 \[
 E_1=\orderP(n^{1/r-1-1/\alpha}).
 \]
 Choosing $r$ close enough to $\alpha/2$ gives 
 \[
 1/r-1-1/\alpha<-1/2
 \]
 because $2/\alpha-1-1/\alpha=1/\alpha-1<-1/2$). Therefore, 
 \[
 E_1=\orderP(n^{-1/2}).
 \]

\textbf{The term $\mathcal E_i$.}
By \eqref{eq_secondorder_green} and \eqref{eq_zeta_expansion}, 
\[
\mathcal E_i=\OrderP(\sigma_i^{-2}(1+n^{-1}\sum_t\xi_t^6))+\orderP(\varrho_n),
\]
 and $\varrho_n=n^{-1/2}$ in the present case, so the second part is $\orderP(n^{-1/2})$. For the first part, note that on $\{\xi^2_{(1)}\le C\mathsf T\}$ with probability $1-\mathrm o(1)$), 
\[
n^{-1}\sum_t\xi_t^6=n^{-1}\sum_t\xi_t^6\mathbbm 1(\xi_t^2\le C\mathsf T).
\]
 By Karamata's theorem \citep[Section 1.6]{bingham1987regular}, 
 \[
 \E[\xi^6\mathbbm 1(\xi^2\le C\mathsf T)]=\mathrm O(1+\mathsf T^{3-\alpha+\varepsilon})
 \]
 for every $\varepsilon>0$. Hence,
 \[
    \sigma_i^{-2}(1+n^{-1}\sum_t\xi_t^6)=\OrderP(\sigma_i^{-2}(1+\mathsf T^{3-\alpha+\varepsilon})).
 \]
 If $\alpha\in(2,3]$, this is $\orderP(\mathsf T^{1-\alpha+\varepsilon})=\orderP(n^{-1/2})$ for small $\varepsilon$. If $\alpha\in(3,4]$, it is $\OrderP(\sigma_i^{-2})=\orderP(\mathsf T^{-2})=\orderP(n^{-1/2})$. If $\alpha>4$ or \eqref{ass_xi_exp} holds, it is $\OrderP(\sigma_i^{-2})=\orderP(n^{-1/2})$ by the additional assumption $\sigma_m\gg n^{1/4}$. The conclusion follows from Slutsky's lemma.

\subsection{Proof of Theorem \ref{thm_realsignal_extremeheavy}}
Again we use the original radii and the representation \eqref{eq_realsignal_representation}. Let $\alpha\in(1,2]$. By Lemma \ref{lem_goodconfiguration}, $\mathbb P(\xi^2_{(1)}\le C\mathsf T)\to1$; we enlarge $\mathsf T$ by the constant $C$ without changing notation. On this event, $\sum_tZ_t=\sum_tZ_t\mathbbm 1(\xi_t^2\le\mathsf T)$.

\textbf{Leading term.} We have 
\[
\E[Z_t\mathbbm 1(\xi_t^2\le\mathsf T)]-1=-\E[\xi^2\mathbbm 1(\xi^2>\mathsf T)]=-\mathrm O(\mathsf T^{1-\alpha}L(\mathsf T))
\]
by Karamata's theorem. For the variance,
\begin{equation*}
\Var\Big(\frac1n\sum_tZ_t\mathbbm 1(\xi_t^2\le\mathsf T)\Big)\le\frac1n\E[\xi^4\mathbbm 1(\xi^2\le\mathsf T)]\,\E(pu_{it}^2)^2\le\frac{3\mathsf T}{n}\,\E\xi^2=\frac{3\mathsf T}{n}.
\end{equation*}
This proves the second part of the theorem. By Chebyshev's inequality, 
\[
n^{-1}\sum_t(Z_t-1)=\OrderP(\sqrt{\mathsf T/n}+\mathsf T^{1-\alpha}L(\mathsf T))=\OrderP(\sqrt{\mathsf T/n}),
\]
 because $\mathsf T^{1-\alpha}\le\sqrt{\mathsf T/n}$ is equivalent to $n^{1/(2\alpha)-1}\lesssim n^{-1/2}$ up to logarithmic factors.

\textbf{Shift term.} By Karamata's theorem, 
\[
\E[\xi^4\mathbbm 1(\xi^2\le\mathsf T)]=\mathrm O(\mathsf T^{2-\alpha+\varepsilon})
\]
for every $\varepsilon>0$. Hence 
\[
\frac{\bar\sigma}{\theta_i}\phi_n n^{-1}\sum_t\xi_t^4=\OrderP(\mathsf T^{2-\alpha+\varepsilon}/\sigma_i)=\orderP(\mathsf T^{1-\alpha+\varepsilon})=\orderP(\sqrt{\mathsf T/n}).
\]
Moreover, $\bar\sigma/\theta_i=\mathrm O(\sigma_i^{-1})=\mathrm o(\mathsf T^{-1})$ and $\mathsf T^{-1}\le\sqrt{\mathsf T/n}$, because $\mathsf T^3\ge n$ when $\alpha\le2$.

\textbf{Remainder.} According to
\[
\mathcal E_i=\OrderP(\sigma_i^{-2}(1+n^{-1}\sum_t\xi_t^6))+\orderP(\varrho_n),
\]
we have 
\[
\varrho_n=n^{-1/2}+\sqrt{\mathsf T/n}\le2\sqrt{\mathsf T/n},
\]
 because $\mathsf T\ge1$. Hence, it follows that
 \[
    \orderP(\varrho_n)=\orderP(\sqrt{\mathsf T/n}).
 \]
 As above, 
 \[
 \sigma_i^{-2}n^{-1}\sum_t\xi_t^6=\OrderP(\sigma_i^{-2}\mathsf T^{3-\alpha+\varepsilon})=\orderP(\mathsf T^{1-\alpha+\varepsilon}),
 \]
 and $\sigma_i^{-2}=\mathrm o(\sqrt{\mathsf T/n})$. Collecting the bounds proves \eqref{eq_realsignal_extremeheavy}.

\subsection{Proof of Lemma \ref{lem_prf_realsignals_secondratio2_1}}\label{subsec_proof_realsignals_secondratio2_1}
Before we go to the details of the proof, we summarize some prior observation. First, recall that we have the global event $\Omega$ which gives a rough estimation of the typical order of $\{\xi^2_i\}$. Next, by \cite{Wen2021}, there exists a large constant $C>0$ such that $\mathbb{P}(\|U^{\prime}U\|>C)\le n^{-c}$, for any constant $c>0$. We define the high probability event $\Xi_0:=\{\|U^{\prime}U\|\le C(1+\sqrt{\phi})^2\}$. In the sequel, we always take $\Omega$ and $\Xi_0$ as given, since $1-\mathbb{P}(\Omega\bigcap\Xi_0)\le 1-\mathbb{P}(\Omega)+1-\mathbb{P}(\Xi_0)=\mathrm o(1)$ by Lemma \ref{lem_goodconfiguration}. As all statements in this section are in probability, such an argument only contributes negligible error.

\textbf{Step 1: Limiting representation.}\newline
In this part, we provide an asymptotic representation for $\mathbf{e}_i^{\prime}UD\mathscr{K}(\theta_i)DU^{\prime}\mathbf{e}_i$ for any $1\le i\le m$. To this end, we approximate $\mathbf{e}_i^{\prime}UD\mathscr{K}(\theta_i)DU^{\prime}\mathbf{e}_i$ by its conditional expectation $\mathbb{E}\big[\mathbf{e}_i^{\prime}UD\mathscr{K}(\theta_i)DU^{\prime}\mathbf{e}_i|D^2\big]$ under the events $\Omega,\Xi_0$. Define the conditional expectation $\mathbb{E}_j:=\mathbb{E}(\cdot|\mathbf{u}_1,\dots,\mathbf{u}_j,D^2),1\le j\le n$ and $U_j$ as the matrix replacing $U$'s $j$-th column with zero vector. Then 
\begin{align*}
        &\mathbf{e}_i^{\prime}UD\mathscr{K}(\theta_i)DU^{\prime}\mathbf{e}_i-\mathbb{E}[\mathbf{e}_i^{\prime}UD\mathscr{K}(\theta_i)DU^{\prime}\mathbf{e}_i]\\
        &=\sum_{j=1}^n(\mathbb{E}_j-\mathbb{E}_{j-1})\mathbf{e}_i^{\prime}U_jD\mathscr{K}(\theta_i)DU_j^{\prime}\mathbf{e}_i+\sum_{j=1}^n(\mathbb{E}_j-\mathbb{E}_{j-1})\mathbf{e}^{\prime}_iU_jD\mathscr{K}(\theta_i)D\mathbf{e}_j\mathbf{u}_j^{\prime}\mathbf{e}_i\\
        &+\sum_{j=1}^n(\mathbb{E}_j-\mathbb{E}_{j-1})\mathbf{e}_i^{\prime}\mathbf{u}_j\mathbf{e}_j^{\prime}D\mathscr{K}(\theta_i)DU_j^{\prime}\mathbf{e}_i+\sum_{j=1}^n(\mathbb{E}_j-\mathbb{E}_{j-1})\mathbf{e}_i^{\prime}\mathbf{u}_j\mathbf{e}_j^{\prime}D\mathscr{K}(\theta_i)D\mathbf{e}_j\mathbf{u}_j^{\prime}\mathbf{e}_i\\
        &:=I_1+I_2+I_3+I_4.
\end{align*}
We will show that the error is negligible if we replace $\mathscr{K}(\theta_i)$ with $I_{n\times n}$ in each $I_k,\;k=1,\dots,4$.

We first look at $I_1$. Note that we still have the prior bound $\|U_jU_j^{\prime}\|\le C(1+\sqrt{\phi})^2$ by the same reason as in $\Xi_0$. We decompose $I_1$ into,
\begin{equation}\label{eq_prf_realsignals_secondratio2_step1_decompositionI1}
    I_1=\sum_{j=1}^n(\mathbb{E}_j-\mathbb{E}_{j-1})\mathbf{e}_i^{\prime}U_jD[\mathscr{K}(\theta_i)-\mathscr{K}_j(\theta_i)]DU_j^{\prime}\mathbf{e}_i+\sum_{j=1}^n(\mathbb{E}_j-\mathbb{E}_{j-1})\mathbf{e}_i^{\prime}U_jD\mathscr{K}_j(\theta_i)DU_j^{\prime}\mathbf{e}_i,
\end{equation}
where $\mathscr{K}_j(\theta_i):=(I-\theta_i^{-1}DU_j^{\prime}\Sigma_2U_jD)^{-1}$. Notice that the second term is zero under the conditional expectation $(\mathbb{E}_j-\mathbb{E}_{j-1})$. We only need to concern with the first term. Rewrite 
\begin{align*}
    \mathcal{S}_2
    &=DU_j^{\prime}\Sigma_2U_jD+D\mathbf{e}_j\mathbf{u}_j^{\prime}\Sigma_2U_jD+DU_j^{\prime}\Sigma_2\mathbf{u}_j\mathbf{e}_j^{\prime}D+D\mathbf{e}_j\mathbf{u}_j^{\prime}\Sigma_2\mathbf{u}_j\mathbf{e}_j^{\prime}D\\
    &:=\mathcal{S}_{2j}+\mathcal{S}_{2r_1}+\mathcal{S}_{2r_2}+\mathcal{S}_{2r_3}.
\end{align*}
By matrix inverse formula, one has $\mathscr{K}(x)=\mathscr{K}_j(x)+x^{-1}\mathscr{K}(x)(\mathcal{S}_{2r_1}+\mathcal{S}_{2r_2}+\mathcal{S}_{2r_3})\mathscr{K}_j(x)$. Applying this relationship for the first term in \eqref{eq_prf_realsignals_secondratio2_step1_decompositionI1}, we have
\begin{align*}
        I_1&=\sum_{k=1}^3\frac{1}{\theta_i}\sum_{j=1}^n(\mathbb{E}_j-\mathbb{E}_{j-1})\mathbf{e}_i^{\prime}U_jD\mathscr{K}(\theta_i)\mathcal{S}_{2r_k}\mathscr{K}_j(\theta_i)DU_j^{\prime}\mathbf{e}_i:=\sum_{k=1}^3I_{1k}.
\end{align*}
Now, we consider $I_{1k}, k=1,2,3$ one by one. Notice that $D,\Sigma_2$ are diagonal matrices and the $j$-th column of $U_j$ is zero. Then the $j$-th row and column of $\mathcal{S}_{2j}$ are all equal to zero, which results that the $j$-th row and column of $\mathscr{K}_j(\theta_i)$ are equal to $\mathbf{e}_j$. Therefore, $\mathbf{e}_j^{\prime}D\mathscr{K}_j(\theta_i)DU^{\prime}_j\mathbf{e}_i=0$, which implies that $I_{12}=I_{13}=0$. It remains to calculate $I_{11}$. We rewrite $I_{11}$ as the following,
\begin{align*}
    I_{11}
    &=\frac{1}{\theta_i}\sum_{j=1}^n(\mathbb{E}_j-\mathbb{E}_{j-1})[\mathfrak{i}_{e_ie_j}\times\mathfrak{i}_{u_je_i}],
\end{align*}
where $\mathfrak{i}_{e_ie_j}=\mathbf{e}_i^{\prime}U_jD\mathscr{K}(\theta_i)D\mathbf{e}_j$ and $\mathfrak{i}_{u_je_i}=\mathbf{u}_j^{\prime}\Sigma_2U_jD\mathscr{K}_j(\theta_i)DU_j^{\prime}\mathbf{e}_i$. We aim to change $\mathscr{K}(\theta_i)$ with $\mathscr{K}_j(\theta_i)$ in $\mathfrak{i}_{e_ie_j}$. Before that, using the fact $\mathbf{e}_i^{\prime}U_jD\mathscr{K}(\theta_i)(I-\theta_i^{-1}\mathcal{S}_2)D\mathbf{e}_j=\mathbf{e}_i^{\prime}U_jD^2\mathbf{e}_j=0$, we obtain the relationship
\begin{align*}
   & \mathfrak{i}_{e_ie_j}=\frac{1}{\theta_i}\mathbf{e}_i^{\prime}U_jD\mathscr{K}(\theta_i)DU^{\prime}\Sigma_2UD^2\mathbf{e}_j\\
    =&\frac{\xi^2_j}{\theta_i}\big(\mathbf{e}_i^{\prime}U_jD\mathscr{K}(\theta_i)DU^{\prime}_j\Sigma_2\mathbf{u}_j+\mathbf{e}_i^{\prime}U_jD\mathscr{K}(\theta_i)D\mathbf{e}_j\mathbf{u}_j^{\prime}\Sigma_2\mathbf{u}_j\big)=\frac{\xi^2_j}{\theta_i}(\mathfrak{i}_{e_iu_j}+\mathfrak{i}_{e_ie_j}\mathbf{u}_j^{\prime}\Sigma_2\mathbf{u}_j),
\end{align*}
where we used the fact $U_j\mathbf{e}_j=0$ and $\mathfrak{i}_{e_iu_j}:=\mathbf{e}_i^{\prime}U_jD\mathscr{K}(\theta_i)DU_j\Sigma_2\mathbf{u}_j$. Besides, replacing $\mathscr{K}(\theta_i)$ with $\mathscr{K}_j(\theta_i)$ in $\mathfrak{i}_{e_iu_j}$, we have 
\begin{align*}       \mathfrak{i}_{e_iu_j}&=\mathbf{e}_i^{\prime}U_jD\mathscr{K}_j(\theta_i)DU^{\prime}_j\Sigma_2\mathbf{u}_j+\frac{1}{\theta_i}\mathfrak{i}_{e_ie_j}\mathbf{u}_j^{\prime}\Sigma_2U_jD\mathscr{K}_j(\theta_i)DU_j^{\prime}\Sigma_2\mathbf{u}_j:=\mathfrak{i}^{(j)}_{e_iu_j}+\frac{1}{\theta_i}\mathfrak{i}_{e_ie_j}\mathfrak{i}_{u_ju_j}.
\end{align*}
Therefore, we can calculate that 
\begin{equation*}
    \mathfrak{i}_{e_ie_j}=a_j^{-1}\times\frac{\xi^2_j}{\theta_i}\mathfrak{i}_{e_iu_j}^{(j)},\quad a_j=1-\frac{\xi^2_j}{\theta_i^2}\mathfrak{i}_{u_ju_j}-\frac{\xi^2_j}{\theta_i}\mathbf{u}_j^{\prime}\Sigma_2\mathbf{u}_j.
\end{equation*}
Since under $\Omega$ and $\Xi_0$, we have $\|\mathscr{K}_j(\theta_i)\|\le\mathrm{O}(1)$ and $|\mathfrak{i}_{u_ju_j}|\le C\|U_j^{\prime}\Sigma_2^2U_j\|\times\|D^2\|\times\|\mathscr{K}_j(\theta_i)\|\le\mathrm{O}(\xi^2_{(1)})$, $\mathbf{u}_j^{\prime}\Sigma_2\mathbf{u}_j\sim\Bar{\sigma}$. Then, $|a_j^{-1}|\le\mathrm{O}(1)$ uniformly over $1\le j\le n$. The summands of $I_{11}$ are martingale differences with respect to $(\mathbb E_j)$. Applying Burkholder's inequality \citep{burkholder1973distribution} conditionally on $D$,
\begin{align*}
        \mathbb{E}\big[|I_{11}|^2\,\big|\,D\big]&\le C\sum_{j=1}^n\mathbb{E}\Big[\big|\frac{1}{\theta_i}\mathfrak{i}_{e_ie_j}\mathfrak{i}_{u_je_i}\big|^2\,\Big|\,D\Big]
        \le \frac{C}{\theta_i^4}\sum_{j=1}^n\xi_j^4\,\|D\|^8\,\mathbb{E}\big[(\mathbf{u}_j^{\prime}\mathbf{e}_i)^4\big]\cdot\|U_j^{\prime}\Sigma_2 U_j\|^4\\
        &\le\frac{C\xi^8_{(1)}}{\theta_i^4}\cdot\frac{1}{p^2}\sum_{j=1}^n\xi_j^4=C\Big(\frac{\xi^2_{(1)}}{\theta_i}\Big)^4\frac{m_4}{p},
\end{align*}
where $m_4=p^{-1}\sum_j\xi_j^4$ and we used $\E(\mathbf u_j'\mathbf e_i)^4\le3/p^2$. By the conditional Chebyshev inequality, for every $M>0$,
\begin{equation*}
\mathbb P\Big(|I_{11}|>M\Big(\frac{\xi^2_{(1)}}{\theta_i}\Big)^2\sqrt{\frac{m_4}{p}}\ \Big|\ D\Big)\le\frac{C}{M^2}.
\end{equation*}
On $\Omega$, $\xi^2_{(1)}\le C\mathsf T$. Moreover, $\sqrt{m_4/p}=\OrderP(\varrho_n)$ with $\varrho_n:=n^{-1/2}+\sqrt{\mathsf T/n}\,\mathbbm 1(\alpha\in(1,2])$, by part (a) of the proof of Lemma \ref{lem_prf_realsignals_secondratio2_2}. Hence, by $\mathsf T/\theta_i\to0$, we obtain
\begin{equation*}
I_{11}=\OrderP\Big(\Big(\frac{\mathsf T}{\theta_i}\Big)^2\varrho_n\Big)=\orderP(\varrho_n).
\end{equation*}

Since $I_{12}=I_{13}=0$, we have $I_1=I_{11}=\orderP(\varrho_n)$. Moreover, $\sum_j(\mathbb E_j-\mathbb E_{j-1})\mathbf e_i'U_jD^2U_j'\mathbf e_i=0$, because $U_j$ does not depend on $\mathbf u_j$. So replacing $\mathscr K(\theta_i)$ by $I_{n\times n}$ in $I_1$ changes it by $\orderP(\varrho_n)$.

For $I_k$, $k=2,3,4$, the same computation applies. In each case, the difference between $I_k$ and its counterpart with $\mathscr K(\theta_i)$ replaced by $I_{n\times n}$ is a sum of martingale differences. Its conditional second moment is bounded by $C(\xi^2_{(1)}/\theta_i)^2m_4/p$, where the power is lower because only one resolvent expansion is needed. By the conditional Chebyshev inequality, the difference is $\OrderP((\mathsf T/\theta_i)\varrho_n)=\orderP(\varrho_n)$. We conclude that
\begin{gather*}
    I_2-\sum_{j=1}^n(\mathbb{E}_j-\mathbb{E}_{j-1})\mathbf{e}_i^{\prime}U_jD^2\mathbf{e}_j\mathbf{u}_j^{\prime}\mathbf{e}_i=\orderP(\varrho_n),\quad
    I_3-\sum_{j=1}^n(\mathbb{E}_j-\mathbb{E}_{j-1})\mathbf{e}_i^{\prime}\mathbf{u}_j\mathbf{e}_j^{\prime}D^2U_j^{\prime}\mathbf{e}_i=\orderP(\varrho_n),\\
    I_4-\sum_{j=1}^n(\mathbb{E}_j-\mathbb{E}_{j-1})\mathbf{e}_i^{\prime}\mathbf{u}_j\mathbf{e}_j^{\prime}D^2\mathbf{e}_j\mathbf{u}_j^{\prime}\mathbf{e}_i=\orderP(\varrho_n).
\end{gather*}
Summing the four expansions, and noting that the martingale decomposition of $\mathbf e_i'UD^2U'\mathbf e_i-\E[\mathbf e_i'UD^2U'\mathbf e_i\,|\,D]$ consists exactly of the four sums on the left-hand sides, we obtain
\begin{equation*}
\mathbf{e}_i^{\prime}UD\mathscr{K}(\theta_i)DU^{\prime}\mathbf{e}_i-\mathbb{E}\big[\mathbf{e}_i^{\prime}UD\mathscr{K}(\theta_i)DU^{\prime}\mathbf{e}_i\,\big|\,D\big]=\mathbf{e}_i^{\prime}UD^2U^{\prime}\mathbf{e}_i-\frac{1}{p}\operatorname{tr}D^2+\orderP(\varrho_n),
\end{equation*}
where we used $\E[\mathbf e_i'UD^2U'\mathbf e_i\,|\,D]=\sum_t\xi_t^2\E u_{it}^2=p^{-1}\tr D^2$.
It remains to calculate $\mathbb{E}[\mathbf{e}_i^{\prime}UD\mathscr{K}(\theta_i)DU^{\prime}\mathbf{e}_i]$.

\textbf{Step 2: Changing randomness $U$ with Gaussian randomness.}\newline
We denote $W$ as the $p\times n$ random matrix whose entries are i.i.d. Gaussian random variables scaled by $p^{-1/2}$, and $\mathbf{w}_j, 1\le j\le n$ as the columns of $W$. We need to show that 
\begin{equation*}
    \sqrt{n}\mathbb{E}\big[\mathbf{e}_i^{\prime}UD(I-\theta_i^{-1}DU^{\prime}\Sigma_2UD)^{-1}DU^{\prime}\mathbf{e}_i-\mathbf{e}_i^{\prime}WD(I-\theta_i^{-1}DW^{\prime}\Sigma_2WD)^{-1}DW^{\prime}\mathbf{e}_i\big]\rightarrow0.
\end{equation*}
The reason why we change $U$ into $W$ will be illustrated at the beginning of \textbf{Step 3}. To handle the above error, we use the Lindeberg replacement strategy. Denote 
\begin{gather*}
    Z^k:=\sum_{i=1}^k\mathbf{u}_i\mathbf{e}_i^{\prime}+\sum_{j=k+1}^n\mathbf{w}_j\mathbf{e}_j^{\prime},\quad Z_0^k:=\sum_{i=1}^{k-1}\mathbf{u}_i\mathbf{e}_i^{\prime}+\sum_{j=k+1}^n\mathbf{w}_j\mathbf{e}_j^{\prime}\\
    \mathscr{K}^k(\theta_i):=(I-\theta_i^{-1}D(Z^k)^{\prime}\Sigma_2Z^kD)^{-1},\quad \mathscr{K}^k_0(\theta_i):=(I-\theta_i^{-1}D(Z^k_0)^{\prime}\Sigma_2Z^k_0D)^{-1},
\end{gather*}
with $Z^n=U$, $Z^0=W$. Similarly to \textbf{Step 1}, we define the event $\Xi^0:=\{\|ZZ^{\prime}\|\le C(1+\sqrt{\phi})^2\}$, $\Xi^0_k:=\{\|Z^k_0(Z^k_0)^{\prime}\|\le C(1+\sqrt{\phi})^2\}$. It is easy to see that $\Xi^0$ and $\Xi^0_k$, for $1\le k\le n$, hold with high probability. Then, in the sequel, we always assume that $\Omega$, $\Xi^0$, $\Xi^0_k$ hold without further explanation. By the Lindeberg replacement argument, it suffices to consider
\begin{align*}
        &\mathbb{E}\big[\mathbf{e}_i^{\prime}UD(I-\theta_i^{-1}DU^{\prime}\Sigma_2UD)^{-1}DU^{\prime}\mathbf{e}_i-\mathbf{e}_i^{\prime}WD(I-\theta_i^{-1}DW^{\prime}\Sigma_2WD)^{-1}DW^{\prime}\mathbf{e}_i\big]\\
        &=\sum_{k=1}^n\Big(\mathbb{E}[\mathbf{e}_i^{\prime}Z^kD\mathscr{K}^k(\theta_i)D(Z^k)^{\prime}\mathbf{e}_i]-\mathbb{E}[\mathbf{e}_i^{\prime}Z^{k-1}D\mathscr{K}^{k-1}(\theta_i)D(Z^{k-1})^{\prime}\mathbf{e}_i]\Big).
\end{align*}
Noticing that $Z^k=Z^k_0+\mathbf{u}_k\mathbf{e}_k^{\prime}$, we have 
\begin{align*}
    &\mathbb{E}[\mathbf{e}_i^{\prime}Z^kD\mathscr{K}^k(\theta_i)D(Z^k)^{\prime}\mathbf{e}_i]=\mathbb{E}[\mathbf{e}_i^{\prime}(Z_0^k+\mathbf{u}_k\mathbf{e}^{\prime}_k)D\mathscr{K}^k(\theta_i)D(Z_0^k+\mathbf{u}_k\mathbf{e}^{\prime}_k)^{\prime}\mathbf{e}_i]\\
        &=\mathbb{E}[\mathbf{e}_i^{\prime}Z_0^kD\mathscr{K}^k(\theta_i)D(Z_0^k)^{\prime}\mathbf{e}_i+\mathbf{e}_i^{\prime}Z_0^kD\mathscr{K}^k(\theta_i)D\mathbf{e}_k\mathbf{u}_k^{\prime}\mathbf{e}_i+\mathbf{e}_i^{\prime}\mathbf{u}_k\mathbf{e}_k^{\prime}D\mathscr{K}^k(\theta_i)D(Z_0^k)^{\prime}\mathbf{e}_i\\
        &+\mathbf{e}_i^{\prime}\mathbf{u}_k\mathbf{e}_k^{\prime}D\mathscr{K}^k(\theta_i)D\mathbf{e}_k\mathbf{u}_k^{\prime}\mathbf{e}_i]:=II^k_1+II^k_2+II^k_3+II^k_4.
\end{align*}
Meanwhile, $Z^{k-1}=Z^k_0+\mathbf{w}_k\mathbf{e}_k^{\prime}$, so one has
\begin{align*}
        &\mathbb{E}[\mathbf{e}_i^{\prime}Z^{k-1}D\mathscr{K}^{k-1}(\theta_i)D(Z^{k-1})^{\prime}\mathbf{e}_i]=\mathbb{E}[\mathbf{e}_i^{\prime}(Z^{k}_0+\mathbf{w}_k\mathbf{e}_k^{\prime})D\mathscr{K}^{k-1}(\theta_i)D(Z^k_0+\mathbf{w}_k\mathbf{e}_k^{\prime})^{\prime}\mathbf{e}_i]\\
        &=\mathbb{E}[\mathbf{e}_i^{\prime}Z_0^kD\mathscr{K}^{k-1}(\theta_i)D(Z_0^k)^{\prime}\mathbf{e}_i+\mathbf{e}_i^{\prime}Z_0^kD\mathscr{K}^{k-1}(\theta_i)D\mathbf{e}_k\mathbf{w}_k^{\prime}\mathbf{e}_i\\
        &+\mathbf{e}_i^{\prime}\mathbf{w}_k\mathbf{e}_k^{\prime}D\mathscr{K}^{k-1}(\theta_i)D(Z_0^k)^{\prime}\mathbf{e}_i+\mathbf{e}_i^{\prime}\mathbf{w}_k\mathbf{e}_k^{\prime}D\mathscr{K}^k(\theta_i)D\mathbf{e}_k\mathbf{w}_k^{\prime}\mathbf{e}_i]\\
        &:=II^{k-1}_1+II^{k-1}_2+II^{k-1}_3+II^{k-1}_4.
\end{align*}
We are going to show that $\sum_{k=1}^n\mathbb{E}(II^k_j-II^{k-1}_j)=\mathrm{o}_{\mathbb{P}}(n^{-1/2}),\quad j=1,\dots,4$. In what follows, we only show the case $j=1$, while the other cases can be handled similarly. We start with $II^k_1$, one has to replace $\mathscr{K}^k(\theta_i)$ with $\mathscr{K}^k_0(\theta_i)$. As in \textbf{Step 1}, using matrix inverse formula, $\mathscr{K}^k(x)=\mathscr{K}^k_0(x)+x^{-1}\mathscr{K}^k(x)D\mathbf{e}_k\mathbf{u}_k^{\prime}\Sigma_2Z^k_0D\mathscr{K}^k_0(x)$, where we used again the fact that the $k$-th row and column of $\mathscr{K}^k_0(x)$ are equal to $\mathbf{e}_k$ and therefore $\mathbf{e}_k^{\prime}D\mathscr{K}^k_0(x)D(Z^k_0)^{\prime}\mathbf{e}_i=0$. Then, one has
\begin{align*}
    &\sum_{k=1}^nII_1^k-\sum_{k=1}^n\mathbb{E}[\mathbf{e}_i^{\prime}Z_0^kD\mathscr{K}^k_0(\theta_i)D(Z^k_0)^{\prime}\mathbf{e}_i]\\
    &=\sum_{k=1}^n\frac{1}{\theta_i}\mathbb{E}[\mathbf{e}_i^{\prime}Z_0^kD\mathscr{K}^k(\theta_i)D\mathbf{e}_k\mathbf{u}_k^{\prime}\Sigma_2Z_0^kD\mathscr{K}^k_0(\theta_i)D(Z_0^k)^{\prime}\mathbf{e}_i]=\frac{1}{\theta_i}\sum_{k=1}^n\mathbb{E}[\mathfrak{ii}_{e_ie_k}\times \mathfrak{ii}_{u_ke_i}].
\end{align*}
Then applying a similar argument as in \textbf{Step 1}, we have $\mathfrak{ii}_{e_ie_k}=b_k^{-1}\times\frac{\xi^2_k}{\theta_i}\mathfrak{ii}_{e_iu_k}^{(k)}$, $b_k=1-\frac{\xi^2_k}{\theta_i^2}\mathfrak{ii}_{u_ku_k}-\frac{\xi^2_k}{\theta_i}\mathbf{u}_k^{\prime}\Sigma_2\mathbf{u}_k,$ where $\mathfrak{ii}_{e_iu_k}^{(k)}:=\mathbf{e}_i^{\prime}Z_0^kD\mathscr{K}^k_0(\theta_i)D(Z^k_0)^{\prime}\Sigma_2\mathbf{u}_k$ and $\mathfrak{ii}_{u_ku_k}:=\mathbf{u}_k^{\prime}\Sigma_2Z^k_0D\mathscr{K}^k_0(\theta_i)D(Z^k_0)^{\prime}\Sigma_2\mathbf{u}_k$. Then we can write 
\begin{equation*}
    \frac{1}{\theta_i}\sum_{k=1}^n\mathbb{E}[\mathfrak{ii}_{e_ie_k}\times\mathfrak{ii}_{u_ke_i}]=\frac{1}{\theta_i}\sum_{k=1}^n\mathbb{E}[\frac{\xi^2_kb_k^{-1}}{\theta_i}\mathfrak{ii}^{(k)}_{e_iu_k}\mathfrak{ii}_{u_ke_i}].
\end{equation*}
Define $\Bar{b}_k=1-\frac{\xi^2_k}{\theta_i^2}\Bar{\mathfrak{ii}}_{u_ku_k}-\frac{\xi^2_k}{\theta_ip}\operatorname{tr}\Sigma_2$, where $\Bar{\mathfrak{ii}}_{u_ku_k}:=p^{-1}\operatorname{tr}(\Sigma_2Z_0^kD\mathscr{K}^k_0(\theta_i)D(Z^k_0)^{\prime}\Sigma_2)$. 
Thus, $b_k^{-1}-\Bar{b}_k^{-1}=-b_k^{-1}\Bar{b}_k^{-1}\big(\frac{\xi^2_k}{\theta_i^2}(\mathfrak{ii}_{u_ku_k}-\Bar{\mathfrak{ii}}_{u_ku_k})-\frac{\xi^2_k}{\theta_i}(\mathbf{u}_k^{\prime}\Sigma_2\mathbf{u}_k-\operatorname{tr}\Sigma_2)\big),$ and $\mathbb{E}[\frac{\xi^2_kb_k^{-1}}{\theta_i}\mathfrak{ii}^{(k)}_{e_iu_k}\mathfrak{ii}_{u_ke_i}]=\mathbb{E}\big[\frac{\xi^2_k}{\theta_i}(\Bar{b}_k^{-1}+b_k^{-1}-\Bar{b}_k^{-1})\mathfrak{ii}^{(k)}_{e_iu_k}\mathfrak{ii}_{u_ke_i}\big].$
It is easy to see that $|b_k^{-1}|\le C$ and $|\Bar{b}_k^{-1}|\le C$ under the events $\Omega$, $\Xi_0$, $\Xi^0_k$. It follows that 
\begin{align*}
        &\mathbb{E}\big[(b_k^{-1}-\Bar{b}_k^{-1})\mathfrak{ii}^{(k)}_{e_iu_k}\mathfrak{ii}_{u_ke_i}\big]\le C\sqrt{\mathbb{E}|b_k^{-1}-\Bar{b}_k^{-1}|^2\times\mathbb{E}(\mathfrak{ii}^{(k)}_{e_iu_k}\mathfrak{ii}_{u_ke_i})^2}\le C\sqrt{\frac{\xi^8_{(1)}}{\theta_i^4p^3}}\le \frac{C\xi^4_{(1)}}{\theta_i^2p^{3/2}}.
\end{align*}
Therefore, $\sum_{k=1}^n\mathbb{E}[\frac{\xi^2_kb_k^{-1}}{\theta_i^2}\mathfrak{ii}^{(k)}_{e_iu_k}\mathfrak{ii}_{u_ke_i}]=\sum_{k=1}^n\mathbb{E}[\frac{\xi^2_k\Bar{b}_k^{-1}}{\theta_i^2}\mathfrak{ii}^{(k)}_{e_iu_k}\mathfrak{ii}_{u_ke_i}]+\mathrm{o}_{\mathbb{P}}(n^{-1/2}).$ It implies that 
\begin{equation*}
    \sum_{k=1}^nII_1^k-\sum_{k=1}^n\mathbb{E}[\mathbf{e}_i^{\prime}Z^k_0D\mathscr{K}^k_0(\theta_i)D(Z^k_0)^{\prime}\mathbf{e}_i]=\sum_{k=1}^n\mathbb{E}[\frac{\xi^2_k\Bar{b}_k^{-1}}{\theta_i^2}\mathfrak{ii}^{(k)}_{e_iu_k}\mathfrak{ii}_{u_ke_i}]+\mathrm{o}_{\mathbb{P}}(n^{-1/2}).
\end{equation*}
A parallel argument will show that 
\begin{equation*}
    \sum_{k=1}^nII_1^{k-1}-\sum_{k=1}^n\mathbb{E}[\mathbf{e}_i^{\prime}Z^k_0D\mathscr{K}^k_0(\theta_i)D(Z^k_0)^{\prime}\mathbf{e}_i]=\sum_{k=1}^n\mathbb{E}[\frac{\xi^2_k\Bar{b}_k^{-1}}{\theta_i^2}\mathfrak{ii}^{(k)}_{e_iw_k}\mathfrak{ii}_{w_ke_i}]+\mathrm{o}_{\mathbb{P}}(n^{-1/2}).
\end{equation*}
Notice that $\mathbb{E}\mathbf{u}_k^{\prime}A\mathbf{u}_k=\mathbb{E}\mathbf{w}_k^{\prime}A\mathbf{w}_k,$ for any $p\times p$ matrix independent with $\mathbf{u}_k,\mathbf{w}_k$. We find that 
$ \sum_{k=1}^n\mathbb{E}[\frac{\xi^2_k\Bar{b}_k^{-1}}{\theta_i^2}\mathfrak{ii}^{(k)}_{e_iu_k}\mathfrak{ii}_{u_ke_i}]=\sum_{k=1}^n\mathbb{E}[\frac{\xi^2_k\Bar{b}_k^{-1}}{\theta_i^2}\mathfrak{ii}^{(k)}_{e_iw_k}\mathfrak{ii}_{w_ke_i}].$ Combining the above analysis, we finally reach that $\sum_{k=1}^nII^k_1-\sum_{k=1}^nII^{k-1}_1=\mathrm{o}_{\mathbb{P}}(n^{-1/2}).$ Applying the same strategies to $j=2,3,4$, we can actually obtain that $\sum_{k=1}^nII^k_j-\sum_{k=1}^nII^{k-1}_j=\mathrm{o}_{\mathbb{P}}(n^{-1/2}),\quad j=1,\dots,4,$ which indicates that changing randomness $U$ with Gaussian randomness $W$ only introduces negligible error. We omit further details here.

\textbf{Step 3: Calculating limiting expectation.}\newline
In \textbf{Step 1}, we deferred the computation of $\mathbb{E}[\mathbf{e}_i^{\prime}UD\mathscr{K}(\theta_i)DU^{\prime}\mathbf{e}_i],\;1\le i\le m$. In \textbf{Step 2}, we successfully reformulated the target expression as $\mathbb{E}[\mathbf{e}_i^{\prime}WD\mathscr{K}^w(\theta_i)DW^{\prime}\mathbf{e}_i]$ where $\mathscr{K}^w(\theta_i):=(I-\theta_i^{-1}DW^{\prime}\Sigma_2WD)^{-1}$. The advantage that we replace $U$ with $W$ is that $\mathbf{e}_i^{\prime}WD^2W^{\prime}\mathbf{e}_i$ is independent with $\mathscr{K}^w(\theta_i)$ since $\mathbf{e}_i\;1\le i\le m$ are orthogonal to the column space of $\Sigma_2$. Then we observe that $\mathbb{E}[\mathbf{e}_i^{\prime}WD\mathscr{K}^w(\theta_i)DW^{\prime}\mathbf{e}_i]=\frac{1}{p}\mathbb{E}[\operatorname{tr}(D^2\mathscr{K}^w(\theta_i))]=\frac{1}{p}\sum_{j=1}^n\xi^2_j\mathbb{E}[\mathscr{K}^w(\theta_i)]_{jj}.$ In order to obtain $\mathbb{E}[\mathscr{K}^w(\theta_i)]_{jj}$, we consider the complement matrix of $\mathscr{K}^w(\theta_i)$ defined as
\begin{equation*}
    \mathcal{K}^w(\theta_i):=(I-\theta_i^{-1}\Sigma_2^{1/2}WD^2W^{\prime}\Sigma_2^{1/2})^{-1}.
\end{equation*}
Then by Schur's complement formula, we have for $1\le j\le n$,
\begin{align*}
       & [\mathscr{K}^w(\theta_i)]_{jj}
        =
        \begin{pmatrix}
            I_{p\times p}\quad &\frac{1}{\sqrt{\theta_i}}DW^{\prime}\Sigma_2^{1/2}\\
            \frac{1}{\sqrt{\theta_i}}\Sigma_2^{1/2}WD\quad &I_{n\times n}
        \end{pmatrix}^{-1}_{jj}\\
        =&1+\big(\frac{1}{\theta_i}DW^{\prime}\Sigma^{1/2}_2\mathcal{K}^w(\theta_i)\Sigma^{1/2}_2WD\big)_{jj}=1+\frac{\xi^2_j}{\theta_i}\mathbf{w}_j^{\prime}\Sigma_2^{1/2}\mathcal{K}^w(\theta_i)\Sigma^{1/2}_2\mathbf{w}_j.
\end{align*}
Again, to take the expectation of $\mathbf{w}_j\Sigma_2^{1/2}\mathcal{K}^w(\theta_i)\Sigma^{1/2}_2\mathbf{w}_j$, we first replace $\mathcal{K}^w(\theta_i)$ with $\mathcal{K}_j^w(\theta_i)$ via
\begin{equation}\label{eq_prf_realsignals_secondratio2_step3_matrixinverse}
    \mathcal{K}^w(\theta_i)=\mathcal{K}_j^w(\theta_i)+\mathcal{K}^w(\theta_i)\frac{\xi^2_j}{\theta_i}\Sigma_2^{1/2}\mathbf{w}_j\mathbf{w}_j^{\prime}\Sigma_2^{1/2}\mathcal{K}_j^w(\theta_i).
\end{equation}
That is, we change $W$ into $W_j$ with $j$-th column of $W$ being zero in $\mathcal{K}^w_j(\theta_i)$. Then one has
\begin{align*}
      &  \mathbf{w}_j^{\prime}\Sigma^{1/2}_2\mathcal{K}^w(\theta_i)\Sigma^{1/2}_2\mathbf{w}_j
        =\mathbf{w}_j^{\prime}\Sigma_2^{1/2}\mathcal{K}^w_j(\theta_i)\Sigma_2^{1/2}\mathbf{w}_j(1+\frac{\xi^2_j}{\theta_i}\mathbf{w}_j^{\prime}\Sigma_2^{1/2}\mathcal{K}^w(\theta_i)\Sigma^{1/2}_2\mathbf{w}_j)\\
        =&\frac{{\theta_i}}{{\theta_i}-{\xi^2_j}\mathbf{w}_j^{\prime}\Sigma_2^{1/2}\mathcal{K}_j^w(\theta_i)\Sigma^{1/2}_2\mathbf{w}_j}\mathbf{w}_j^{\prime}\Sigma^{1/2}_2\mathcal{K}_j^w(\theta_i)\Sigma^{1/2}_2\mathbf{w}_j.
\end{align*}
Under $\Omega$, we find that $\|\mathcal{K}^w_j(\theta_i)\|\le C$ and consequently $ \frac{\xi^2_j}{\theta_i}\mathbf{w}_j^{\prime}\Sigma_2^{1/2}\mathcal{K}^w_j(\theta_i)\Sigma_2^{1/2}\mathbf{w}_j\rightarrow0,$ as $n,p\rightarrow\infty$. Therefore, one can see that 
\begin{align*}
        &\mathbb{E}\Big[\frac{\xi^2_j\mathbf{w}_j^{\prime}\Sigma_2^{1/2}\mathcal{K}^w_j(\theta_i)\Sigma^{1/2}_2\mathbf{w}_j-p^{-1}\operatorname{tr}\big(\Sigma_2\mathcal{K}^w_j(\theta_i)\big)}{\theta_i-{\xi^2_j}\mathbf{w}_j^{\prime}\Sigma_2^{1/2}\mathcal{K}_j^w(\theta_i)\Sigma^{1/2}_2\mathbf{w}_j}\Big]
        \le C\frac{\xi^2_j\|\Sigma_2\|_F\|\mathcal{K}^2_j(\theta_i)\|}{p\theta_i}\le\mathrm{o}(p^{-1/2}),
\end{align*}
where we used the large deviation inequality Lemma \ref{lem_largedeviationbound} for Gaussian random vectors. Furthermore, we rewrite the denominator as
\begin{align*}
        &\mathbb{E}\Big|\frac{{\xi^2_j}}{{p\theta_i}-p{\xi^2_j}\mathbf{w}_j^{\prime}\Sigma^{1/2}_2\mathcal{K}^w_j(\theta_i)\Sigma^{1/2}_2\mathbf{w}_j}-\frac{{\xi^2_j}}{{p\theta_i}-{\xi^2_j}\operatorname{tr}(\Sigma_2\mathcal{K}^w_j(\theta_i))}\Big|\operatorname{tr}\big(\Sigma_2\mathcal{K}^w_j(\theta_i)\big)\\
        &\le\frac{C\xi^4_j}{\theta^2_i}\mathbb{E}\Big|\mathbf{w}_j^{\prime}\Sigma^{1/2}_2\mathcal{K}^w_j(\theta_i)\Sigma^{1/2}_2\mathbf{w}_j-p^{-1}\operatorname{tr}(\Sigma_2\mathcal{K}^w_j(\theta_i))\Big|\le\mathrm{o}_{\mathbb{P}}(n^{-1/2}).
\end{align*}
Consequently, we have for all $1\le j\le n$ that 
\begin{align*}
        \mathbb{E}[\mathscr{K}^w(\theta_i)]_{jj}
        &=\mathbb{E}\Big[\frac{{p\theta_i}}{{p\theta_i}-{\xi^2_j}\operatorname{tr}(\Sigma_2\mathcal{K}^w_j(\theta_i))}\Big]+\mathrm{o}_{\mathbb{P}}(n^{-1/2}).
\end{align*}
Further, notice that
\begin{align*}
    &\mathbb{E}\Big|\frac{{p\theta_i}}{{p\theta_i}-{\xi^2_j}\operatorname{tr}(\Sigma_2\mathcal{K}^w_j(\theta_i))}-\frac{{p\theta_i}}{{p\theta_i}-{\xi^2_j}\operatorname{tr}(\Sigma_2\mathcal{K}^w(\theta_i))}\Big|\\
    &=\mathbb{E}\Big|{p{\xi^4_j} \mathbf{w}_j^{\prime}\mathcal{K}^w(\theta_i)\Sigma_2\mathcal{K}^w_j(\theta_i)\mathbf{w}_j}{\big[{p\theta_i}-{\xi^2_j}\operatorname{tr}(\Sigma_2\mathcal{K}^w_j(\theta_i))\big]\big[{p\theta_i}-{\xi^2_j}\operatorname{tr}(\Sigma_2\mathcal{K}^w(\theta_i))\big]}\Big|\le\mathrm{o}_{\mathbb{P}}(n^{-1/2}).
\end{align*}
Therefore, we obtain that $ \mathbb{E}[\mathbf{e}_i^{\prime}WD\mathscr{K}^w(\theta_i)DW^{\prime}\mathbf{e}_i]=\frac{1}{p}\sum_{j=1}^n\mathbb{E}\Big[\frac{{p\theta_i}\xi^2_j}{{p\theta_i}-{\xi^2_j}\operatorname{tr}(\Sigma_2\mathcal{K}^w(\theta_i))}\Big]+\mathrm{o}_{\mathbb{P}}(n^{-1/2}).$ Moreover, to calculate the expectation on the right-hand-side of the above equation, we observe that 
\begin{align}\label{eq_prf_realsignals_secondratio2_step3_traceexpectation}
        &\frac{1}{p}\sum_{j=1}^n\mathbb{E}\Big|\frac{{p\theta_i}\xi^2_j}{{p\theta_i}-{\xi^2_j}\operatorname{tr}(\Sigma_2\mathcal{K}^w(\theta_i))}-\frac{{p\theta_i}\xi^2_j}{{p\theta_i}-{\xi^2_j}\mathbb{E}\operatorname{tr}(\Sigma_2\mathcal{K}^w(\theta_i))}\Big| \nonumber\\
        &\le C\sqrt{\mathbb{E}\big|\frac{1}{p\theta_i}\big(\operatorname{tr}\Sigma_2\mathcal{K}^w(\theta_i)-\mathbb{E}[\operatorname{tr}\Sigma_2\mathcal{K}^w(\theta_i)]\big)\big|^2}.
\end{align}
Now, in order to handle $(p\theta_i)^{-1}\big(\operatorname{tr}\Sigma_2\mathcal{K}^w(\theta_i)-\mathbb{E}[\operatorname{tr}\Sigma_2\mathcal{K}^w(\theta_i)]\big)$, we regard it as a martingale sequence. Define $\mathbb{E}_j^w:=\mathbb{E}(\cdot|\mathbf{w}_1,\dots,\mathbf{w}_j,D^2)$. Write
\begin{align}\label{eq_prf_realsignals_secondratio2_step3_martingale}
        &\frac{1}{p\theta_i}\big(\operatorname{tr}\Sigma_2\mathcal{K}^w(\theta_i)-\mathbb{E}[\operatorname{tr}\Sigma_2\mathcal{K}^w(\theta_i)]\big)=\sum_{j=1}^n(\mathbb{E}_j^w-\mathbb{E}_{j-1}^w)\frac{1}{p\theta_i}\operatorname{tr}\Sigma_2\mathcal{K}^w(\theta_i) \nonumber\\
        &=\sum_{j=1}^n(\mathbb{E}^w_j-\mathbb{E}^w_{j-1})\frac{\xi^2_j}{p\theta_i^2}\mathbf{w}_j^{\prime}\mathcal{K}^w_j(\theta_i)\Sigma_2\mathcal{K}^w(\theta_i)\mathbf{w}_j,
\end{align}
where in the second step, we used \eqref{eq_prf_realsignals_secondratio2_step3_matrixinverse}. Then, plugging \eqref{eq_prf_realsignals_secondratio2_step3_martingale} back to \eqref{eq_prf_realsignals_secondratio2_step3_traceexpectation}, by Burkholder's inequality, we have
\begin{align*}
        &\mathbb{E}\big|\frac{1}{p\theta_i}(\operatorname{tr}\Sigma_2\mathcal{K}^w(\theta_i)-\mathbb{E}\operatorname{tr}(\Sigma_2\mathcal{K}^w(\theta_i)))\big|^2\\
        &\qquad\le\sum_{j=1}^n\mathbb{E}^w_j\big|\frac{\xi^2_j}{p\theta_i^2}\mathbf{w}_j^{\prime}\mathcal{K}^w_j(\theta_i)\Sigma_2\mathcal{K}^w(\theta_i)\mathbf{w}_j\big|^2\le\frac{C\xi^4_{(1)}}{p\theta_i^4}\le\frac{C\mathsf{T}^2}{p\theta_i^4}.
\end{align*}
Therefore, we conclude that $\mathbb{E}[\mathbf{e}_i^{\prime}WD\mathscr{K}^w(\theta_i)DW^{\prime}\mathbf{e}_i]=\frac{1}{p}\sum_{j=1}^n\frac{\xi^2_j}{1-\frac{\xi^2_j}{p\theta_i}\mathbb{E}[\operatorname{tr}(\Sigma_2\mathcal{K}^w(\theta_i)])}+\mathrm{o}_{\mathbb{P}}(n^{-1/2}).$ Finally, it only remains to calculate $p^{-1}\mathbb{E}[\operatorname{tr}(\Sigma_2\mathcal{K}^w(\theta_i))]$. We observe that $\Sigma_2$ plays the same structural role in $p^{-1}\mathbb{E}[\operatorname{tr}\Sigma_2\mathcal{K}^w(\theta_i)]$ as $D^2$ does in $p^{-1}\mathbb{E}[\operatorname{tr}(D^2\mathscr{K}^w(\theta_i))]$. Consequently, a parallel argument to give
\[
p^{-1}\mathbb{E}[\operatorname{tr}(\Sigma_2\mathcal{K}^w(\theta_i))]=\frac{1}{p}\sum_{k=1}^{p-m}\frac{\sigma_{m+k}}{1-\frac{\sigma_{m+k}}{p\theta_i}\mathbb{E}[\operatorname{tr}(D^2\mathscr{K}^w(\theta_i))]}+\mathrm{o}_{\mathbb{P}}(n^{-1/2}).
\]
Thereafter, we have 
$$
\frac{1}{p}\mathbb{E}[\operatorname{tr}(\Sigma_2\mathscr{K}^w(\theta_i))]=\frac{1}{p}\sum_{j=1}^n\frac{\xi^2_j}{1-{\xi^2_j}\sum_{k=1}^{p-m}\frac{\sigma_{m+k}}{{p\theta_i}-{\sigma_{m+k}}\mathbb{E}[\operatorname{tr}(\Sigma_2\mathscr{K}^w(\theta_i))]}}+\mathrm{o}_{\mathbb{P}}(n^{-1}).
$$ 
Recall the definition of $\zeta_i$ in \eqref{eq_def_zetai}. Such equation is stable in the sense that 
\begin{align*}
        &\zeta_i-\frac{1}{p}\mathbb{E}[\operatorname{tr}(\Sigma_2\mathscr{K}^w(\theta_i))]-\mathrm{o}_{\mathbb{P}}(n^{-1/2})\\
        &=\frac{1}{p}\sum_{j=1}^n\frac{{\xi^4_j}\sum_{k=1}^{p-m}\frac{{\sigma_{m+k}^2}}{({\theta_i}-\sigma_{m+k}\zeta_i)({p\theta_i}-{\sigma_{m+k}}\mathbb{E}[\operatorname{tr}(\Sigma_2\mathscr{K}^w(\theta_i))])}}{(1-\frac{\xi^2_j}{p}\sum_{k=1}^{p-m}\frac{\sigma_{m+k}}{\theta_i-\sigma_{m+k}\zeta_i})(1-{\xi^2_j}\sum_{k=1}^{p-m}\frac{\sigma_{m+k}}{{p\theta_i}-{\sigma_{m+k}}\mathbb{E}[\operatorname{tr}(\Sigma_2\mathscr{K}^w(\theta_i))]})}\\
        &\times(\zeta_i-\frac{1}{p}\mathbb{E}[\operatorname{tr}(\Sigma_2\mathscr{K}^w(\theta_i))])=\mathrm{o}_{\mathbb{P}}(1)\times(\zeta_i-\frac{1}{p}\mathbb{E}[\operatorname{tr}(\Sigma_2\mathscr{K}^w(\theta_i))]).
\end{align*}
It gives that $\zeta_i-\frac{1}{p}\mathbb{E}[\operatorname{tr}(\Sigma_2\mathscr{K}^w(\theta_i))]=\mathrm{o}_{\mathbb{P}}(n^{-1/2}).$ Finally, summarising the results in \textbf{Step 1} to \textbf{Step 3}, we have 
\begin{align*}
    &\Gamma_1^{\prime}UD\mathscr{K}(\theta_i)DU^{\prime}\Gamma_1
    =\Gamma_1^{\prime}UD^2U^{\prime}\Gamma_1-\frac{1}{p}\operatorname{tr}D^2I_{m\times m}+\zeta_iI_{m\times m}+\mathrm{o}_{\mathbb{P}}\Big(\frac{1}{n}+\frac{\mathsf{T}}{n}\times\mathbbm{1}(\alpha\in(1,2])\Big).
\end{align*}
For the diagonal entries, Step 1 contributes $\orderP(\varrho_n)$, while Steps 2 and 3 contribute $\orderP(n^{-1/2})$, which is $\orderP(\varrho_n)$ since $\varrho_n\ge n^{-1/2}$. For the off-diagonal entries $k\ne l$, following the martingale argument of Step 1 gives 
\[
\mathbf e_k'UD\mathscr K(\theta_i)DU'\mathbf e_l-\E[\mathbf e_k'UD\mathscr K(\theta_i)DU'\mathbf e_l\,|\,D]=Q_{kl}+\orderP(\varrho_n),
\]
 where we used $\E[Q_{kl}\,|\,D]=0$. Moreover, 
 \[
 \E[\mathbf e_k'UD\mathscr K(\theta_i)DU'\mathbf e_l\,|\,D]=0
 \]
 by symmetry. Indeed, the law of $U$ is invariant under the reflection that changes the sign of its $k$-th row. Since $\Sigma_2=\diag(0,\dots,0,\sigma_{m+1},\dots,\sigma_p)$, the matrix $\mathscr K(\theta_i)$ depends on $U$ only through its $m+1,\dots,p$ rows, so it remains unchanged by this reflection, while $\mathbf e_k'UD\mathscr K(\theta_i)DU'\mathbf e_l$ will change the sign. We finish the proof of this lemma.

\subsection{Proof of Lemma \ref{lem_prf_realsignals_secondratio2_2}}\label{subsec_proof_realsignals_secondratio2_2}
We work on the event $\mathcal F_n$ of Step 1 in the proof of Theorem \ref{thm_realsignal_secondratio2}, and condition on $D$. Recall that $\mathbf u_1,\dots,\mathbf u_n$ are i.i.d.\ uniform on $\mathbb S^{p-1}$. For such $\mathbf u=(u_1,\dots,u_p)'$, it holds that 
\[
\begin{aligned}
&\E u_k^2=1/p,\quad \E u_ku_l=0,\quad \E u_k^4=3/(p(p+2))\le3/p^2,\\
&\E u_k^2u_l^2=1/(p(p+2))\le1/p^2,\quad k\ne l.
\end{aligned}
\]

With above preparation, we consider the matrix $Q$ first. For $1\le k,l\le m$, 
\[
Q_{kl}=\sum_{t=1}^n\xi_t^2\big(u_{kt}u_{lt}-\delta_{kl}/p\big).
\]
 Since it is a sum of independent, conditionally centered terms, it holds
\begin{equation*}
\E\big[Q_{kl}^2\,\big|\,D\big]=\sum_{t=1}^n\xi_t^4\Var(u_{kt}u_{lt})\le\frac3{p^2}\sum_{t=1}^n\xi_t^4=\frac{3m_4}{p}.
\end{equation*}
By the conditional Chebyshev inequality, $Q_{kl}=\OrderP(\sqrt{m_4/p})$. Moreover, one may observe that $\sqrt{m_4/p}=\OrderP(\varrho_n)$. If $\alpha>2$ or \eqref{ass_xi_exp} holds, $m_4=\OrderP(1)$ (see the beginning of the proof of Theorem \ref{thm_realsignal_secondratio2}) and $p\asymp n$. If $\alpha\in(1,2]$, $m_4\le C\mathsf T$ on $\Omega$, so that $\sqrt{m_4/p}\le C\sqrt{\mathsf T/n}$. As $m$ is fixed, $\|Q\|_\infty=\OrderP(\varrho_n)$ and $\|Q\|\le m\|Q\|_\infty=\OrderP(\varrho_n)$.

Second, we consider the quantity $B_n:=\|\Gamma_1'UD\mathcal S_2DU'\Gamma_1\|$. Since $\mathcal S_2=DU'\Sigma_2UD$, the matrix $\Gamma_1'UD\mathcal S_2DU'\Gamma_1=\Gamma_1'UD^2U'\Sigma_2UD^2U'\Gamma_1$ is positive semidefinite. Hence its norm can be bounded by
\begin{equation*}
B_n\le\sum_{k=1}^m\mathbf e_k'UD^2U'\Sigma_2UD^2U'\mathbf e_k\le c_1\sum_{k=1}^m\|UD^2U'\mathbf e_k\|^2.
\end{equation*}
Now $UD^2U'\mathbf e_k=\sum_t\xi_t^2u_{kt}\mathbf u_t$, so $\|UD^2U'\mathbf e_k\|^2=\sum_{s,t}\xi_s^2\xi_t^2u_{ks}u_{kt}\,\mathbf u_s'\mathbf u_t$. For $s=t$, $\mathbf u_t'\mathbf u_t=1$ and $\E u_{kt}^2=1/p$. For $s\ne t$, $\E[u_{ks}u_{kt}\mathbf u_s'\mathbf u_t]=\sum_{l=1}^p\E[u_{ks}u_{ls}]\,\E[u_{kt}u_{lt}]=1/p^2$ by independence. Therefore
\begin{equation*}
\E\big[\|UD^2U'\mathbf e_k\|^2\,\big|\,D\big]=\frac1p\sum_t\xi_t^4+\frac1{p^2}\sum_{s\ne t}\xi_s^2\xi_t^2\le m_4+\Big(\frac1p\sum_t\xi_t^2\Big)^2\le m_4+C
\end{equation*}
on $\Omega$. Applying the conditional Markov inequality, we obtain that $B_n=\OrderP(1+m_4)$.

Now, we turn to the matrix $G$. We decompose 
\[
\mathscr K(\lambda_i)\mathscr K(\theta_i)-I=(\mathscr K(\lambda_i)-I)\mathscr K(\theta_i)+(\mathscr K(\theta_i)-I).
\]
By \eqref{eq_green_K} and since $\mathcal S_2$, $\mathscr K(\lambda_i)$ and $\mathscr K(\theta_i)$ commute,
\begin{equation*}
\mathscr K(\lambda_i)\mathscr K(\theta_i)-I=\lambda_i^{-1}\mathcal S_2^{1/2}\mathscr K(\lambda_i)\mathscr K(\theta_i)\mathcal S_2^{1/2}+\theta_i^{-1}\mathcal S_2^{1/2}\mathscr K(\theta_i)\mathcal S_2^{1/2}.
\end{equation*}
For any $n\times n$ matrix $N$, 
\[
\|\Gamma_1'UD\mathcal S_2^{1/2}N\mathcal S_2^{1/2}DU'\Gamma_1\|\le\|N\|\,\|\mathcal S_2^{1/2}DU'\Gamma_1\|^2=\|N\|B_n.
\]
 Using the facts $\|\mathscr K(\lambda_i)\mathscr K(\theta_i)\|\le4$, $\|\mathscr K(\theta_i)\|\le2$ and $\lambda_i\ge\theta_i/2$ on $\mathcal F_n$, we get
\begin{equation*}
\big\|G-\Gamma_1'UD^2U'\Gamma_1\big\|\le\Big(\frac{4}{\lambda_i}+\frac{2}{\theta_i}\Big)B_n\le\frac{10B_n}{\theta_i}=\OrderP\Big(\frac{1+m_4}{\theta_i}\Big).
\end{equation*}
Finally, it holds that
\[
\Gamma_1'UD^2U'\Gamma_1-I_m=Q+(p^{-1}\operatorname{tr}D^2-1)I_m
\]
with 
\[
p^{-1}\operatorname{tr}D^2-1=p^{-1}\sum_j(\xi_j^2-\phi_n)=\OrderP(\varrho_n),
\]
as shown in Step 3 of the proof of Theorem \ref{thm_realsignal_secondratio2}. Combining with (a),
\begin{equation*}
\|G-I_m\|\le\|Q\|+|p^{-1}\sum_j(\xi_j^2-\phi_n)|+\frac{10B_n}{\theta_i}=\OrderP\Big(\varrho_n+\frac{1+m_4}{\theta_i}\Big)=\OrderP(\omega_{n,i}).
\end{equation*}
This proves the lemma. 

\section{Proof of the results in Section \ref{sec_main_fakesignals}}\label{sec_prf_fake}
In this section, we prove Lemma \ref{lem_fakesignal_compare_nonspikes}, Theorem \ref{thm_fakesignal_limit} and Lemma \ref{lem_radius_separation}. At the very beginning, we prepare
several preliminary results.
\subsection{Technique preparation}
Recalling the sample covariance matrices $S_2$ and $\mathcal{S}_2$ in \eqref{eq_def_S2}, we denote their empirical spectral distributions (ESDs) as
        \begin{equation*}
            \mu_{S_2}:=\frac{1}{p}\sum_{i=1}^p\delta_{\lambda_i(S_2)},\quad \mu_{\mathcal{S}_2}:=\frac{1}{n}\sum_{j=1}^n\delta_{\lambda_j(\mathcal{S}_2)}.
        \end{equation*}
        For any $z=E+\mathrm{i}\eta\in\mathbb{C}_{+}$, denote their resolvents by
        \begin{equation*}
            G(z):=(S_2-zI)^{-1}\in\mathbb{C}^{p\times p},\quad \mathcal{G}(z):=(\mathcal{S}_2-zI)^{-1}\in\mathbb{C}^{n\times n}.
        \end{equation*}
        Correspondingly, we can define their Stieltjes transforms as
        \begin{equation}\label{eq_def_mS2}
            m_{S_2}:=\int\frac{1}{x-z}\mu_{S_2}=\frac{1}{p}\operatorname{tr}\big(G(z)\big),\quad m_{\mathcal{S}_2}:=\int\frac{1}{x-z}\mu_{\mathcal{S}_2}=\frac{1}{n}\operatorname{tr}\big(\mathcal{G}(z)\big).
        \end{equation}

        In order to characterize the limits of $\mu_{S_2}$ or $\mu_{\mathcal{S}_2}$, we define the following equation system,
        \begin{definition}[System of consistent equations]\label{def_prf_fakesignals_systemequation} 
        For $z \in \mathbb{C}_+,$ we define the triplets $(m_{1n}, m_{2n}, m_n) \in \mathbb{C}^3_+,$  via the following system of equations.
\begin{align}\label{eq_systemequations}
    &m_{1n}(z)=\frac{1}{p}\sum_{i=1}^{p-m}\frac{\sigma_{m+i}}{-z(1+\sigma_{m+i}  m_{2n}(z))},\quad m_{2n}(z)=\frac{1}{p}\sum_{i=1}^n\frac{\xi_i^2}{-z(1+\xi^2_im_{1n}(z))},\nonumber\\
   & m_{n}(z)=\frac{1}{p}\sum_{i=1}^{p-m}\frac{1}{-z(1+\sigma_{m+i}  m_{2n}(z))}. 
\end{align}
\end{definition}

For sufficiently large $n,$ we find that $\mu_{S_2}$ has a nonrandom deterministic equivalent and can be uniquely characterized by the above consistent equations. This is summarized by the following theorem.  
\begin{theorem}\label{lem_solutionsystem}
Suppose Assumptions \ref{ass_xi} and \ref{ass_sigma}. Then conditional on event $\Omega$, for any $z \in \mathbb{C}_+,$ when $n$ is sufficiently large, there exists a unique solution $(m_{1n}(z), m_{2n}(z), m_{n}(z)) \in \mathbb{C}_+^3$ to the system of equations in \eqref{eq_systemequations}. Moreover, $m_n(z)$ is the Stieltjes transform of some probability density function $\rho \equiv \rho_n$ defined on $\mathbb{R}$ which can be obtained using the inversion formula.
\end{theorem}
\begin{proof}
The proofs can be obtained by following lines of the arguments  of \cite[Theorem 2]{Karoui2009} and \cite[Theorem 1]{paul2009no}, or \cite[Theorem 2.4]{ding2021spiked} verbatim. We omit the details. 
\end{proof}
\begin{remark}
    we can prove an unconditional counterpart for Theorem \ref{lem_solutionsystem} by integrating out $\xi^2.$ Denote $F(x)$ as the CDF of $\xi^2.$ We find the unconditional system of equations as follows
\begin{align}\label{eq_systemequationsintergrate}
   & m_{1n,c}(z)=\frac{1}{n}\sum_{i=1}^{p-m}\frac{\sigma_{m+i}}{-z(1+\sigma_im_{2n,c}(z))},\quad m_{2n,c}(z)=\int_0^{+\infty} \frac{s}{-z(1+s m_{1n,c}(z))} \mathrm{d} F(s),\nonumber\\
   & m_{n,c}(z)=\frac{1}{p}\sum_{i=1}^{p-m}\frac{1}{-z(1+\sigma_{m+i}m_{2n,c}(z))}. 
\end{align}
In this setting, $(m_{1n,c}, m_{2n,c}, m_{n,c})$ is always deterministic.
\end{remark}

Thanks to Theorem \ref{lem_solutionsystem}, it is easy to see that the study of \eqref{eq_systemequations} can be reduced to the analysis of 
\begin{equation}\label{eq_functionFequal}
F_n(m_{1n}(z),z)=0, \quad z\in\mathbb{C}_{+},
\end{equation}
where $F_n(\cdot, \cdot)$ is defined as 
\begin{equation}\label{eq_F(m,z)1}
    F_n(m_{1n}(z),z)
    =\frac{1}{p}\sum_{i=1}^{p-m}{\sigma_{m+i}}\left({-z+\frac{\sigma_{m+i}}{p}\sum_{j=1}^n\frac{\xi^2_j}{1+\xi^2_j m_{1n}(z)}}\right)^{-1}-m_{1n}(z).
\end{equation} 
By the form of $m_{1n}(z)$, $m_{2n}(z)$ in \eqref{eq_systemequations}, it is more convenient to study the following quantities
\begin{equation}\label{eq_def_m12}
   m_1(z):=\frac{1}{p}{\rm tr}\left(G(z)\Sigma_2\right), \quad  m_2(z):=\frac{1}{p}\sum_{i=1}^n\xi^2_i\mathcal G_{ii}(z).
\end{equation}
Below, we will see that $m_1(z)$ and $m_2(z)$ approximate $m_{1n}(z)$ and $m_{2n}(z)$ well, respectively.

        Throughout this section, we will frequently use the minors of a matrix.  For the sample covariance matrix $S_2$, we denote $S_2=Y_2Y_2^{\prime}$ where $Y_2:=\Sigma_2^{1/2}UD$. Denote the index set ${\mathcal I}=\{1,\dots,n\}$ and given an index set $\mathcal{T}\subset{\mathcal I}$, we introduce
	the notation $Y_2^{(\mathcal{T})}$ to denote the $p \times(n-|\mathcal{T}|)$ minor of $Y_2$
	obtained from removing all the $i$-th columns of $Y_2$ for $i\in \mathcal{T}$ and keeping the original indices of $Y_2$. In particular, $Y_2^{(\emptyset)}=Y_2$. For convenience, we briefly write $(\{i\})$, $(\{i,j\})$ and $\{i,j\}\cup \mathcal{T}$ as $(i)$, $(i,j)$
	and $(ij\mathcal{T})$ respectively. Correspondingly, we denote sample covariance matrices and resolvents of the minors as 
	\begin{equation}\label{eq_defnminor}
	S_2^{(\mathcal{T})}=(Y_2^{(\mathcal{T})}) (Y_2^{(\mathcal{T})})^{\prime},\ {\mathcal S}_2^{(\mathcal{T})}=(Y_2^{(\mathcal{T})})^{\prime} (Y_2^{(\mathcal{T})}). 
	\end{equation}
	and
	\begin{equation}\label{eq_defnminorG}
	G^{(\mathcal{T})}(z)=(S_2^{(\mathcal{T})}-zI)^{-1}, \quad{\mathcal G}^{(\mathcal{T})}(z)=({\mathcal S}_2^{(\mathcal{T})}-zI)^{-1}.
	\end{equation}
    Moreover, we can define $m_{S_2}^{(\mathcal{T})}(z)$, $m_{\mathcal{S}_2}^{(\mathcal{T})}(z)$, $m_{1n}^{(\mathcal{T})}(z)$, $m_{2n}^{(\mathcal{T})}(z)$, $m_{n}^{(\mathcal{T})}(z)$, $m_1^{(\mathcal{T})}(z)$, $m_2^{(\mathcal{T})}(z)$ accordingly, via $G^{(\mathcal{T})}(z)$ and $\mathcal{G}^{(\mathcal{T})}(z)$.

Next, we collect several identities which are useful in our discussion.

 \begin{lemma}[Resolvent identities]\label{lem_resolvent}
        Let $\{\mathbf{y}_{i,2}\} \subset \mathbb{R}^p$ be the columns of $Y_2$. Then we have that 
        \begin{align*}
                \mathcal{G}_{ii}(z)&=-\frac{1}{z+z\mathbf{y}_{i,2}^{*}G^{(i)}(z)\mathbf{y}_{i,2}},\ ~~
                \mathcal{G}_{ij}(z)=z\mathcal{G}_{ii}(z)\mathcal{G}^{(i)}_{jj}(z)\mathbf{y}_{i,2}^{*}G^{(ij)}(z)\mathbf{y}_{i,2}\quad (i\neq j)\\
                \mathcal G_{ij}(z) & =\mathcal G_{ij}^{(k)}(z)+\frac{\mathcal G_{ik}(z)\mathcal G_{kj}(z)}{\mathcal G_{kk}(z)}\quad (i,j\ne k).
        \end{align*}
        \end{lemma}
        \begin{lemma}[Matrix identities]\label{lem_matrixidentities}
            For any finite subset $\mathcal{T}\subset\{1,\dots,n\}$, we have that 
        \begin{equation}
        \left\|G^{(\mathcal{T})}\Sigma_2^{1/2} \right \|^2_F=\eta^{-1}\operatorname{Im}\operatorname{tr} \left(G^{(\mathcal{T})} \Sigma_2 \right). 
        \end{equation}
        Moreover, we have that 
	\begin{gather}
			\left|{\rm tr}(\mathcal{G}^{(i)}-\mathcal{G}) \right|  \le  \eta^{-1}, \ 			\left|{\rm tr}({ G}^{(i)}-{ G}) \right| \le  |z|^{-1}+\eta^{-1}, \ 
			\left|\operatorname{Im}{\rm tr}({G}^{(i)}-{\mathcal G})\right|  \le  \eta|z|^{-2}+\eta^{-1}. 
		\end{gather}
        \end{lemma}
        The proof of the above two lemmas can be found in \cite[Lemma A.18-A.19]{xie2026multiplier}. We omit them here.

\subsection{Proof of Lemma \ref{lem_fakesignal_compare_nonspikes}}
The proof of the Lemma \ref{lem_fakesignal_compare_nonspikes} follows closely from a discussion similar to the proof of \cite[Theorem 3.7]{ding2021spiked}, or \cite[Theorem 2.7]{knowles2013isotropic}, or \cite[Theorem 2.7]{bloemendal2016principal}. Due to similarity, we only sketch the proof strategies, provide the key ingredients and point out the main differences. In fact, our proof will be easier since the spikes are much larger than the edges and we only consider the first few extremal eigenvalues. As discussed in \cite[Appendix D]{ding2021spiked}, or \cite[Section 6]{knowles2013isotropic}, or \cite[Section 4]{bloemendal2016principal}, the proof consists of the following three steps.
\begin{enumerate}
\item[(i).] We first find the permissible regions in which contain the eigenvalues of $S$ with high probability.
\item[(ii).] Then we apply a counting argument to a special case (where all the spikes are well-separate), and show that the results hold under this special case. 
\item[(iii).] Finally we use a continuity argument to extend the results in (ii) to the general case using the gaps in the permissible regions.   
\end{enumerate}

\subsection{Proof of Theorem \ref{thm_fakesignal_limit}}
The proof of Theorem \ref{thm_fakesignal_limit} consists of two steps. Our discussion will focus on the largest eigenvalue $\lambda_{1,2}$ of $S_2$ while the rest leading $k$-th eigenvalues $\lambda_{k,2}$'s can be handled similarly. We begin from the equation system \eqref{eq_systemequations}, where we introduce a real auxiliary quantity $\mu_1>0$ governed by $F_n(m_{1n}(\mu_1),\mu_1)$. More specifically, since the pair $(\mu_1,m_{1n}(\mu_1))$ satisfies $F_n(m_{1n}(\mu_1), \mu_1)=0$, we set $\mu_1$ to be the largest solution of the following equation
\begin{equation}\label{eq_prf_fakesignals_defmu1}
    1+(\xi^2_{(1)}+\mathsf{q})m_{1n}(\mu_{1})=0,
\end{equation}
where for fixed sufficiently small constant $\epsilon>0,$ $\mathsf{q}$ is defined in \eqref{eq_def_q}. In the first step, we give several prior estimations of $m_{1n}(z)$ (also $m_{2n}(z)$, $m_{n}(z)$) around the target region. Also, we need to confirm the existence of $\mu_1$ in $F_n(m_{1n}(\mu_1), \mu_1)=0$. After that, we use a perturbation argument to characterize $\lambda_{1,2}$ via the equation $M(\lambda_{1,2})=0$ defined in \textbf{Step 2}. We wish to establish $M(\mu_1)\approx0$ from \eqref{eq_prf_fakesignals_defmu1} to link $\mu_1$ with $\lambda_{1,2}$. To this end, we need to control the fluctuation of $m_{1}(z)$ as well as $m_{2}(z)$. Then, we have to establish a kind of local law for $m_{1}(z)$, $m_{2}(z)$ to control their randomness around $m_{1n}(z)$ and $m_{2n}(z)$. We present such results also in the second step. In the second step, we conduct perturbation argument to show that $1+(\xi^2_{(1)}+\mathsf{q})m_{1n}(\lambda_{1,2})\approx 0$. Using \eqref{eq_prf_fakesignals_defmu1}, by a continuity
and stability analysis, we conclude the proof of this theorem. Below, we give the details of each step.
\textbf{Step 1: Asymptotic properties}\newline
Recalling the definition of $\mu_1$ in \eqref{eq_prf_fakesignals_defmu1}, we conduct the discussion on the following domain,
\begin{equation}\label{eq_def_spectraldomain}
    \mathbf{D}_{\mu}\equiv \mathbf{D}_{\mu}( C):=\left\{z=E+\mathrm{i}\eta \in \mathbb{C}_+: |E-\mu_1| \le C \mathsf{T}_{2}, \ n^{-2/3}\le\eta\le C \mu_1\right\},
\end{equation}
for some large constant $ C>0$. Define the control parameter $\mathsf{c}$ as follows
\begin{equation}\label{eq_def_controlparameter}
\mathsf{c}:=
\begin{cases}
\frac{\log n}{n^{1/\alpha}}, & \ \text{if (\ref{ass_xi_poly}) holds}; \\
\frac{1}{\log^{1/\beta-\epsilon} n}, & \ \text{if (\ref{ass_xi_exp}) holds}, 
\end{cases}
\end{equation} 
for some sufficient small constant $0<\epsilon<1/\beta$. The following asymptotic properties for $m_{1n}(z)$ and $m_{2n}(z)$ hold on $\mathbf{D}_{\mu}$.
\begin{lemma}\label{lem_prf_fakesignals_asymptoticproperties}
Suppose Assumptions \ref{ass_xi}, \ref{ass_sigma} hold. For any fixed realization $\{\xi_i^2\} \in \Omega$, we have  
\begin{enumerate}
\item For $z \in \mathbf{D}_{\mu}$ and some constants $C_1, C_2>0,$ 
\begin{equation*}
\operatorname{Re} m_{1n}(z) \asymp -E^{-1}, \ C_1 \eta E^{-2} \le \operatorname{Im} m_{1n}(z) \le C_2 \eta E^{-1}.  
\end{equation*}
\item When $|E-\mu_1| \leq C \mathsf{q}$ for some sufficiently large constant $C>0,$ letting $m_{1n}(E)=\lim_{\eta \downarrow 0} m_{1n}(E+\mathrm{i} \eta),$ we have that 
\begin{equation*}
    m_{1n}(E) \asymp -E^{-1}.
\end{equation*}
\item For $z \in \mathbf{D}_{u}$ and $\mathsf{c}$ defined in \eqref{eq_def_controlparameter}, we have that 
\begin{gather*}
    |m_{2n}(z)|=\mathrm{O}(\mathsf{c}),\quad |m_n(z)|=\mathrm{O}(E^{-1}),\\
    \operatorname{Im}m_{2n}(z)=\mathrm{O}(\eta E^{-1}),\quad \operatorname{Im}m_n(z)=\mathrm{O}(\eta E^{-2}).
\end{gather*}
\end{enumerate} 
\end{lemma}

The following theorem gives the averaged local law for $m_1(z)$ and $m_2(z)$ on $\mathbf{D}_{\mu}$,
\begin{theorem}[Averaged local laws] \label{thm_prf_fakesignals_locallaw} 
Suppose Assumptions \ref{ass_xi}, \ref{ass_sigma} hold. For any fixed realization $\{\xi_i^2\} \in \Omega$ where $\Omega$ is introduced in Definition \ref{def_goodconfiguration}, let $m_1^{(1)}(z)$ and $m_{1n}^{(1)} (z)$ be defined by removing the column or entries associated with $\xi_{(1)}^2.$ The following holds uniformly for $z \in \mathbf{D}_{u}$ in \eqref{eq_def_spectraldomain}
\begin{enumerate}
\item If \eqref{ass_xi_poly} of  Assumption \ref{ass_xi} holds, we have that 
\begin{equation*}
m_{1}^{(1)}(z)=m_{1n}^{(1)}(z)+\mathrm{O}_{\prec}\left(n^{-1/2-2/\alpha}\right). 
\end{equation*}
\item If \eqref{ass_xi_exp} of Assumption \ref{ass_xi} holds, we have that 
\begin{equation*}
m_{1}^{(1)}(z)=m_{1n}^{(1)}(z)+\mathrm{O}_{\prec}\left(n^{-1/2}\right).
\end{equation*}
\end{enumerate}  
\end{theorem}

The proof of Lemma \ref{lem_prf_fakesignals_asymptoticproperties} can be found in \cite{xie2026multiplier} and the proof of Theorem \ref{thm_prf_fakesignals_locallaw} will be postponed in Section \ref{sec_proof_thmfake_locallaw}.

\textbf{Step 2: Locations for largest eigenvalue of $S_2$}\newline
In this part, we study the convergent limits of the largest eigenvalues of $S_2,$ i.e., $\lambda_{k,2}.$ For simplification, we only focus on $\lambda_{1,2}$ while the rest can be handled similarly.

In order to quantify the location of $\lambda_{1,2} \equiv \lambda_1(S_2),$ we need to introduce several auxiliary quantities. Recall $\mu_1$ defined in \eqref{eq_prf_fakesignals_defmu1}. Similarly, we denote $\mu_2$ by replacing $\xi^2_{(1)}$ with $\xi^2_{(2)}$ in \eqref{eq_prf_fakesignals_defmu1}.  Moreover,  for $\mathsf{q}$ and the sufficiently small constant $\epsilon>0$ in \eqref{eq_def_q}, we denote 
\begin{equation}\label{eq_prf_fakesignals_locationdefinition}
\mu_1^{\pm}:=\mu_1 \pm n^{-1/2+2 \epsilon} \mathsf{q} ,
\end{equation}
and notice that 
\begin{equation}\label{eq_qremoveonecolumn}
S_2^{(1)}:=S_2-\mathbf{y}_{(1),2}\mathbf{y}^{\prime}_{(1),2},
\end{equation}
where $\mathbf{y}_{(1),2}$ is the column of $Y_2$ associated with $\xi^2_{(1)}$. Accordingly, we denote the largest eigenvalue of $S_2^{(1)}$ as $\lambda_{1,2}^{(1)} \equiv \lambda_1(S_2^{(1)}).$ The main result of this part is summarised in Theorem \ref{thm_prf_fakesignals_eigenvaluerigidity} below.
\begin{theorem}\label{thm_prf_fakesignals_eigenvaluerigidity}
Suppose Assumptions \ref{ass_xi} and \ref{ass_sigma} hold. For some sufficiently small constant $\epsilon>0$ and $\mu_1^{\pm}$ defined in \eqref{eq_prf_fakesignals_locationdefinition}, conditioned on $\Omega$ in Lemma \ref{lem_goodconfiguration}, with high probability, we have that 
$
\lambda_{1,2} \in [\mu_1^-, \mu_1^+]. 
$
\end{theorem}
\begin{proof}[\bf Proof of Theorem \ref{thm_prf_fakesignals_eigenvaluerigidity}]
In what follows, we restrict the discussion on the probability event $\Omega$ in Lemma \ref{lem_goodconfiguration}. By Weyl's inequality, we have that $\lambda_{2,2} \leq \lambda_{1,2}^{(1)}.$ Moreover, by (\ref{eq_keyused}), we see that with high probability $\lambda_{1,2}^{(1)}<\mu_2.$ The rest of the proof focuses on the following two claims: 
\begin{equation}\label{eq_mu1mu2bound}
\mu_1-\mu_2 \geq C^{-1}n^{1/\alpha} \log^{-C}n,
\end{equation}
and for $S_2^{(1)}$ in (\ref{eq_qremoveonecolumn}),
\begin{equation}\label{eq_defnmlambda}
M(x)=1+\mathbf{y}^{\prime}_{(1),2}G^{(1)}(x)\mathbf{y}_{(1),2}, \ G^{(1)}(x)=(S_2^{(1)}-x I)^{-1},
\end{equation}
$M(x)$ changes sign with high probability at $\mu_1^-$ and $\mu_1^+.$  In fact, $\lambda_{1,2}$ satisfies the following equation with high probability
\begin{equation}\label{eq_masterequation}
{\rm det}(\lambda_{1,2}I-\mathbf{y}_{(1),2}\mathbf{y}^{\prime}_{(1),2}-S_2^{(1)})=0\Rightarrow M(\lambda_{1,2})=0,
\end{equation}
as long as $\lambda_{1,2}>\lambda_{1,2}^{(1)}.$ On the other hand, if $M(x)$ changes sign at $\mu_1^{\pm},$ then by continuity, there must at least be one eigenvalue of $S_2$ within the interval $[\mu_1^-, \mu_1^+].$ If (\ref{eq_mu1mu2bound}) holds, then combining the above arguments, the only possible eigenvalue of $S_2$ in this interval is $\lambda_{1,2}$ and moreover, $\lambda_{1,2}>\lambda_{1,2}^{(1)}.$  

We first justify \eqref{eq_mu1mu2bound}. Recall that $\mu_1$ is defined in \eqref{eq_prf_fakesignals_defmu1} according to 
\begin{equation*}
1+(\xi^2_{(1)}+\mathsf{q})m_{1n}(\mu_1)=0.
\end{equation*}
Together with \eqref{eq_functionFequal} and \eqref{eq_F(m,z)1}, we readily obtain that 
\begin{equation*}
1=\frac{1}{p}\sum_{i=1}^{p-m}\frac{\sigma_{m+i}}{\frac{\mu_1}{\xi^2_{(1)}+\mathsf{q}}-\frac{\sigma_{m+i}}{n}\sum_{j=1}^n\frac{\xi_{(j)}^2}{\xi^2_{(1)}+\mathsf{q}-\xi^2_{(j)}}}.
\end{equation*}
Recall \eqref{eq_def_controlparameter}. Using the definition of $\mathsf{q}$, we see that on $\Omega,$ for some constant $C>0$

\begin{align}\label{eq: 2.2 beta=2}
&\frac{1}{p}\sum_{j=1}^n\frac{\xi_{(j)}^2}{\xi^2_{(1)}+\mathsf{q}-\xi^2_{(j)}}=\frac{1}{p}\frac{\xi_{(1)}^2}{\mathsf{q}}+\frac{1}{p}\sum_{j=2}^{n}\frac{\xi_{(j)}^2}{\xi^2_{(1)}+\mathsf{q}-\xi^2_{(j)}}\nonumber\\
   &\le \frac{Cn^{\epsilon}\log n}{p}+\frac{1}{p}\sum_{j=2}^{n}\frac{\xi_{(j)}^2}{\xi^2_{(1)}+\mathsf{q}-\xi^2_{(2)}}
   \le\frac{Cn^{\epsilon}\log n}{n}+\frac{C\log n}{n^{1/\alpha}}=C \mathsf{c}.
\end{align}

The definition of $\Bar{\sigma}$ and
the above arguments imply that on $\Omega$
\begin{equation}\label{eq_mu1part}
{\mu_1}/({\xi_{(1)}^2+\mathsf{q}})=\Bar{\sigma}+\mathrm{O}(\mathsf{c}).  
\end{equation}

By an analogous argument, we have that for some constants $C_k>0, k=1,2,3,$ 
\begin{align*}
&\frac{1}{p}\left(\sum_{j=1}^n \frac{\xi^2_{(j)}}{\xi^2_{(1)}+\mathsf{q}-\xi^2_{(j)}}-\sum_{j=2}^n\frac{\xi^2_{(j)}}{\xi^2_{(2)}+\mathsf{q}-\xi^2_{(j)}}\right)
\geq C_1\mathsf{c}-\frac{1}{p}\sum_{j=2}^n\frac{\xi^2_{(j)}(\xi^2_{(1)}-\xi^2_{(2)})}{(\xi^2_{(1)}+\mathsf{q}-\xi^2_{(j)})(\xi^2_{(2)}+\mathsf{q}-\xi^2_{(j)})}\\&\ge C_1 \mathsf{c}-\frac{1}{n}\sum_{j=2}^n \frac{\xi^2_{(j)}}{\xi^2_{(2)}+\mathsf{q}-\xi^2_{(j)}}\geq C_2 \mathsf{c},
\end{align*}
and on the other hand
\begin{align*}
    &\frac{1}{p}\left(\sum_{j=1}^n \frac{\xi^2_{(j)}}{\xi^2_{(1)}+\mathsf{q}-\xi^2_{(j)}}-\sum_{j=2}^n \frac{\xi^2_{(j)}}{\xi^2_{(2)}+\mathsf{q}-\xi^2_{(j)}}\right)\\
    &\le \frac{1}{p}\left(\frac{\xi^2_{(1)}}{d_{1}}+\sum_{j=2}\frac{\xi^2_{(j)}}{\xi^2_{(1)}+\mathsf{q}-\xi^2_{(j)}}+\sum_{j=2}\frac{\xi^2_{(j)}}{\xi^2_{(2)}+\mathsf{q}-\xi^2_{(j)}}\right) \leq C_3\mathsf{c}.
\end{align*}
Using the above control, the definition of $\mu_2$ and a discussion similar to (\ref{eq_mu1part}), we can prove that 
\begin{equation}\label{eq_mu2part}
{\mu_2}/({\xi_{(2)}^2+\mathsf{q}})=\phi \bar{\sigma}+\mathrm{O}(\mathsf{c}).  
\end{equation}

Combining (\ref{eq_mu1part}) and (\ref{eq_mu2part}), we immediately see that  
\begin{equation}\label{eq_bbbbb}
\frac{\mu_1}{\xi^2_{(1)}+\mathsf{q}}-\frac{\mu_2}{\xi^2_{(2)}+\mathsf{q}}=\mathrm{O}(\mathsf{c}).
\end{equation}
This implies that 
\begin{align*}
    \mu_1-\mu_2&=(\xi^2_{(1)}+\mathsf{q})\left(\frac{\mu_1}{\xi^2_{(1)}+\mathsf{q}}-\frac{\mu_2}{\xi^2_{(2)}+\mathsf{q}}\right)+\mu_2\left(\frac{\xi^2_{(1)}+\mathsf{q}}{\xi^2_{(2)}+\mathsf{q}}-1\right)\\
    &=(\xi^2_{(1)}+\mathsf{q})\left(\frac{\mu_1}{\xi^2_{(1)}+\mathsf{q}}-\frac{\mu_2}{\xi^2_{(2)}+\mathsf{q}}\right)+\frac{\mu_2}{\xi^2_{(2)}+\mathsf{q}}\left(\xi^2_{(1)}-\xi^2_{(2)}\right)\ge C^{-1}n^{1/\alpha} \log^{-C}n,
\end{align*}
where in the third step we used (\ref{eq_bbbbb}), Lemma \ref{lem_goodconfiguration} and the definition of $\mathsf{c}$ in (\ref{eq_def_controlparameter}). This completes the proof of (\ref{eq_mu1mu2bound}). 

Next, we will show that 
\begin{equation*}
M(\mu_1^-)<0,  \ M(\mu_1^+)>0. 
\end{equation*}
Due to similarity, in what follows, we focus on the first inequality. Note that $\mathbf{y}_{1,2}^{\prime} G^{(1)}(\mu_1^-) \mathbf{y}_{1,2}= \xi_{(1)}^2  \mathbf{u}_1^{\prime}\Sigma_2^{1/2} G^{(1)}(\mu_1^-) \Sigma_2^{1/2} \mathbf{u}_1.$ Moreover, recall that $m_1^{(1)}(\mu_1^-)=n^{-1} \operatorname{tr}(\Sigma_2^{1/2}G^{(1)}(\mu_1^-)\Sigma_2^{1/2}).$ Then according to Lemma \ref{lem_largedeviationbound}, we have that 
\begin{align}\label{eq_conconconone}
\mathbf{y}_{1,2}^{\prime} G^{(1)}(\mu_1^-) \mathbf{y}_{1,2}& 
 =\xi^2_{(1)} m_1^{(1)}(\mu_1^-)+\mathrm{O}_{\prec} ( {\xi^2_{(1)}}{n^{-1/2-1/\alpha}}), 
\end{align}
where we used (\ref{eq_mu1mu2bound}) and the fact $\mu_2>\lambda_1^{(1)}$.  Moreover, for some sufficiently small constant $\epsilon_0>0$ and $z_0=\mu_1^-+\mathrm{i} n^{-1/2-\epsilon_0},$ we can obtain that  

\begin{align}\label{eq_bigexpansion}
m_1^{(1)}(\mu_1^-)& =\left[m_1^{(1)}(\mu_1^-)-m_1^{(1)}(z_0) \right]+\left[m_1^{(1)}(z_0)-m_{1n}^{(1)}(z_0)\right]\\
&+\left[m_{1n}^{(1)}(z_0)-m_{1n}^{(1)}(\mu_1^-)\right]+m_{1n}^{(1)}(\mu_1^-) =\mathsf{P}_1+\mathsf{P}_2+\mathsf{P}_3+m_{1n}^{(1)}(\mu_1^-).
\end{align}
First, by Theorem \ref{thm_prf_fakesignals_locallaw}, we have that $\mathsf{P}_2 \prec n^{-1/2-2/\alpha}$. Second, let $\{\mathbf{v}^{(1)}_i\}$ be the eigenvectors of $S_2^{(1)}$ associated with the eigenvalues $\{\lambda_{i,2}^{(1)}\}.$ Then we have that
\begin{align*}
|\mathsf{P}_1|&\le\frac{1}{p}\sum_{i=1}^n|T\mathbf{v}^{(1)}_i|^2\left|\frac{1}{\lambda^{(1)}_{i,2}-\mu_1^-}-\frac{1}{\lambda^{(1)}_{i,2}-z_0}\right|=\frac{1}{p}\sum_{i=1}^n|T\mathbf{v}^{(1)}_i|^2\left|\frac{\mathrm{i}n^{-1/2-\epsilon_0}}{(\lambda^{(1)}_{i,2}-\mu_1^-)(\lambda^{(1)}_{i,2}-z_0)}\right|\\
  &\le\frac{1}{p}\sum_{i=1}^n|T\mathbf{v}^{(1)}_i|^2\Big|\frac{\mathrm{i}n^{-1/2-\epsilon_0}+\mathrm{O}_{\prec}(n^{-2/\alpha-1/2}\log^2n)}{|\lambda^{(1)}_{i,2}-z_0|^2}\Big|\\
  &\prec\operatorname{Im}(m_1^{(1)}(z_0))\times \mathrm{O}_{\prec}(n^{-1/2-\epsilon_0})
\prec n^{-1-1/\alpha},
\end{align*}
where in the third step we used (\ref{eq_mu1mu2bound}) and the fact $\mu_2>\lambda_1^{(1)}$ and in the last step we used Lemmas \ref{lem_prf_fakesignals_asymptoticproperties}, \ref{lem_goodconfiguration} and (\ref{eq_mu1part}). Third, according to Lemma \ref{lem_solutionsystem},  we can obatin that 
\begin{align*}
  \mathsf{P}_3 =\mathcal{L}_1(z_0,z_0)-\mathcal{L}_1(\mu_1^-,\mu_1^-)&=[\mathcal{L}_1(z_0,z_0)-\mathcal{L}_1(z_0,\mu_1^-)]+[\mathcal{L}_1(z_0,\mu_1^-)-\mathcal{L}_1(\mu_1^-,\mu_1^-)]\\
    &:=\mathcal{M}^{(1)}_{31}+\mathcal{M}^{(1)}_{32},
\end{align*}
where binary function $\mathcal{L}_1(z_0,\mu_1^-)=\frac{1}{p}\sum_{i=1}^{p-m}\frac{\sigma_{m+i}}{-z_0(1-\frac{\sigma_{m+i}}{n}\sum_{j=2}^n{\xi^2_{(j)}}[{\mu_1^-(1+m^{(1)}_{1n}(\mu_1^-)\xi^2_{(j)})}]^{-1})}$.
According to (\ref{eq_mu1part}), (\ref{eq_mu1mu2bound}), and Lemmas \ref{lem_goodconfiguration}, \ref{lem_prf_fakesignals_asymptoticproperties}, we conclude that with high probability $|1+\sigma_{m+i} m_{2n}^{(1)}(z_0)|, |1+\sigma_{m+i} m_{2n}^{(1)}(\mu_1^-)|  \sim 1.$ For $\mathcal{M}^{(1)}_{31}$, using the definition in \eqref{eq_systemequations} and the above bounds,  we have that with high probability
\begin{align} \label{eq_m11control}
  &  \mathcal{M}^{(1)}_{31}=\frac{1}{p}\sum_{i=1}^{p-m} \frac{\sigma^2_{m+i}}{-z_0(1+\sigma_{m+i}m^{(1)}_{2n}(z))(1+\sigma_{m+i}m^{(1)}_{2n}(\mu_1^-))}\big(m^{(1)}_{2n}(\mu_1^-)-m^{(1)}_{2n}(z)\big) \nonumber  \\
    &=\mathrm{O}(|z_0|^{-1})\times\frac{1}{n}\sum_{j=2}^n \left(\frac{\xi^2_{(j)}}{-\mu_1^-(1+\xi^2_{(j)}m_{1n}^{(1)}(\mu_1^-))}-\frac{\xi^2_{(j)}}{-z(1+\xi^2_{(j)}m_{1n}^{(1)}(z_0))}\right) \nonumber \\
    &=\mathrm{O}(|z_0|^{-1})\times\frac{1}{n}\sum_{j=2}^n \left(\frac{\xi^2_{(j)}}{-\mu_1^-(1+\xi^2_{(j)}m^{(1)}_{1n}(\mu_1^-))}-\frac{\xi^2_{(j)}}{-\mu_1^-(1+\xi^2_{(j)}m^{(1)}_{1n}(z_0))}\right) \nonumber \\
    &+\mathrm{O}(|z_0|^{-1})\times\frac{1}{n}\sum_{j=2}^n \left(\frac{\xi^2_{(j)}}{-\mu_1^-(1+\xi^2_{(j)}m^{(1)}_{1n}(z))}-\frac{\xi^2_{(j)}}{-z(1+\xi^2_{(j)}m_{1n}^{(1)}(z_0))}\right) \nonumber \\ 
    &=\mathrm{O}(|z_0|^{-1})\times\frac{1}{n}\sum_{j=2}^n \frac{\xi^4_{(j)}(m^{(1)}_{1n}(z_0)-m_{1n}^{(1)}(\mu_1^-))}{-\mu_1^-(1+\xi^2_{(j)}m^{(1)}_{1n}(\mu_1^-))(1+\xi^2_{(j)}m^{(1)}_{1n}(z_0))}+\mathrm{O}(|z_0|^{-1}) \nonumber\\
    &\times\frac{1}{n}\sum_{j=2}^n \frac{\xi^2_{(j)}(\mu_1^--z_0)}{z\mu_1^-(1+\xi^2_{(j)}m^{(1)}_{1n}(z_0))} \nonumber \\
    &=\mathrm{o}(1)\times(m^{(1)}_{1n}(z_0)-m^{(1)}_{1n}(\mu_1^-))+\mathrm{O}(n^{-2/\alpha-1/2-\epsilon_0}), 
\end{align}
where in the second to last step we used Lemma \ref{lem_prf_fakesignals_asymptoticproperties} and in the last step we used Lemmas \ref{lem_prf_fakesignals_asymptoticproperties}  and \ref{lem_goodconfiguration}. This implies with high probability 
\begin{equation*}\label{control_m11}
\mathcal{M}_{31}^{(1)}=\mathrm{O} \left( n^{-2/\alpha-1/2-\epsilon_0} \right).  
\end{equation*}

Similarly, for $\mathcal{M}^{(1)}_{32}$, we have that 
\begin{align}\label{eq_controlm12}
    \mathcal{M}^{(1)}_{32}&=\frac{1}{p}\sum_{i=1}^{p-m} \frac{\sigma_{m+i}(z_0-\mu_1^-)}{z_0\mu_1^-(1+\sigma_{m+i}m^{(1)}_{2n}(\mu_1^-))}=\mathrm{O}(n^{-2/\alpha-1/2-\epsilon_0}).
\end{align}
Combining the above arguments, we have that $\mathsf{P}_3=\mathrm{O}\left( n^{-2/\alpha-1/2-\epsilon_0} \right). $

Inserting the bounds of $\mathsf{P}_k, 1 \leq k \leq 3$ into (\ref{eq_bigexpansion}), we conclude that $|m^{(1)}_1(\mu_1^-)-m^{(1)}_{1n}(\mu_1^-)|\prec n^{-1/\alpha-1/2-\epsilon_0}.$ Together with (\ref{eq_conconconone}) and Lemma \ref{lem_goodconfiguration}, this yields that  
\begin{equation}\label{eq_M1reducedcase}
  M(\mu_1^-)=1+(\xi^2_{(1)}+\mathsf{q})m^{(1)}_{1n}(\mu_1^-)+\mathrm{O}_{\prec}(n^{-1/2-\epsilon_0}).
\end{equation}

In what follows, we study $1+(\xi^2_{(1)}+\mathsf{q})m^{(1)}_{1n}(\mu_1^-)$. We rewrite that
\begin{align}\label{eq_reducedcontrol}
    (\xi^2_{(1)}+\mathsf{q})m^{(1)}_{1n}(\mu_1^-)&=(\xi^2_{(1)}+\mathsf{q})m^{(1)}_{1n}(\mu_1)-(\xi^2_{(1)}+\mathsf{q})\big(m^{(1)}_{1n}(\mu_1)-m^{(1)}_{1n}(\mu_1^-)\big).
\end{align}
Using the definition for $\mu_1$ that $1+(\xi^2_{(1)}+\mathsf{q})m_{1n}(\mu_1)=0$ and the definitions in \eqref{eq_systemequations}, by a discussion similar to (\ref{eq_m11control}),  we have that for some constant $C>0$
\begin{align}\label{eq: m^s_1_1n-m_1n}
    &1+(\xi^2_{(1)}+\mathsf{q})m^{(1)}_{1n}(\mu_1) \nonumber\\
    &=1+(\xi^2_{(1)}+\mathsf{q})m_{1n}(\mu_1)+(\xi^2_{(1)}+\mathsf{q})\big(m_{1n}^{(1)}(\mu_1)-m_{1n}(\mu_1)\big)\nonumber\\
    &=(\xi^2_{(1)}+\mathsf{q})\frac{1}{p}\sum_{i=1}^{p-m}\left(\frac{\sigma_{m+i}}{-\mu_1(1+\sigma_{m+i} m_{2n}^{(1)}(\mu_1))}-\frac{\sigma_{m+i}}{-\mu_1(1+\sigma_{m+i} m_{2n}(\mu_1))}\right)\nonumber\\
    &=(\xi^2_{(1)}+\mathsf{q}) \left[\frac{1}{p}\sum_{i=1}^{p-m}\frac{\sigma^2_{m+i}}{-\mu_1(1+\sigma_{m+i} m_{2n}^{(1)}(\mu_1))(1+\sigma_{m+i} m_{2n}(\mu_1))}\right]\big(m_{2n}(\mu_1)-m_{2n}^{(1)}(\mu_1)\big)\nonumber\\
    & \leq C \frac{(\xi^2_{(1)}+\mathsf{q})}{\mu_1}\times \left|\frac{1}{p}\sum_{j=1}^n\frac{\xi^2_{(j)}}{-\mu_1(1+\xi^2_{(j)}m_{1n}(\mu_1))}-\frac{1}{p}\sum_{j=2}^n\frac{\xi^2_{(j)}}{-\mu_1(1+\xi^2_{(j)}m^{(1)}_{1n}(\mu_1))}\right|\nonumber\\
    &\leq C \frac{(\xi^2_{(1)}+\mathsf{q})}{\mu_1}\times\xi^2_{(2)}|m_{1n}^{(1)}(\mu_1)-m_{1n}(\mu_1)|+n^{-1},
\end{align}
where in the last step we again used Lemma \ref{lem_goodconfiguration}. This yields that 
\begin{equation*}
|(\xi^2_{(1)}+\mathsf{q})\big(m_{1n}^{(1)}(\mu_1)-m_{1n}(\mu_1)\big)| \leq C \frac{(\xi^2_{(1)}+\mathsf{q})}{\mu_1}\times\xi^2_{(2)}|m_{1n}^{(1)}(\mu_1)-m_{1n}(\mu_1)|+n^{-1}, 
\end{equation*}
which implies that $(\xi^2_{(1)}+\mathsf{q})\big(m^{(1)}_{1n}(\mu_1)-m_{1n}(\mu_1)\big)=\mathrm{O}(n^{-1})$. Together with (\ref{eq_reducedcontrol}), we have 
\begin{equation}\label{eq_finalpartone}
1+(\xi^2_{(1)}+\mathsf{q})m^{(1)}_{1n}(\mu_1^-)=-(\xi^2_{(1)}+\mathsf{q})\big(m^{(1)}_{1n}(\mu_1)-m^{(1)}_{1n}(\mu_1^-)\big)+\mathrm{O}(n^{-1}).
\end{equation} 

Recall that we have established $\mu_1, \mu_1^->\lambda_1^{(1)}.$ Then by Theorem \ref{thm_prf_fakesignals_locallaw} and the monotonicity of $m_1^{(1)}$ outside the bulk, the first term on the right-hand side of (\ref{eq_finalpartone}) is negative. In order to show $M(\mu_1^-)<0,$ in light of (\ref{eq_M1reducedcase}), it suffices to show that its magnitude is much larger than $\mathrm{O}(n^{-1/2-\epsilon_0}).$  To see this, we decompose that
\begin{align*}
    &m^{(1)}_{1n}(\mu_1)-m^{(1)}_{1n}(\mu_1^-)    =\frac{1}{p}\sum_{i=1}^{p-m}\frac{\sigma_{m+i}}{-\mu_1(1+\sigma_{m+i}m^{(1)}_{2n}(\mu_1))}-\frac{1}{p}\sum_{i=1}^{p-m}\frac{\sigma_{m+i}}{-\mu_1^-(1+\sigma_{m+i}m^{(1)}_{2n}(\mu_1^-))}\\
 &   =[\mathcal{L}_2(\mu_1,\mu_1)-\mathcal{L}_2(\mu_1,\mu_1^-)]+[\mathcal{L}_2(\mu_1,\mu_1^-)-\mathcal{L}_2(\mu_1^-,\mu_1^-)]
    :=\tilde{\mathcal{M}}_{11}^{(1)}+\tilde{\mathcal{M}}_{12}^{(1)},
\end{align*}
where binary function  $\mathcal{L}_2(\mu_1,\mu_1^-)=\frac{1}{p}\sum_{i=1}^{p-m}\frac{\sigma_{m+i}}{-\mu_1(1-\frac{\sigma_{m+i}}{n}\sum_{j=2}^n{\xi^2_{(j)}}[{\mu_1(1+\xi^2_{(j)}m_{1n}^{(1)}(\mu_1^-))}]^{-1})}$.

Similarly to the discussion of (\ref{eq_m11control}), we have that,
\begin{align*}
    &\tilde{\mathcal{M}}_{11}^{(1)}
    =\mathrm{O}(\frac{1}{\mu_1})\times\frac{1}{n}\sum_{j=2}^n\frac{\xi^4_{(j)}(m^{(1)}_{1n}(\mu_1)-m^{(1)}_{1n}(\mu_1^-))}{-\mu_1(1+\xi^2_{(j)}m_{1n}^{(1)}(\mu_1^-))(1+\xi^2_{(j)}m_{1n}^{(1)}(\mu_1))}\\
    =&\mathrm{o}(1)\times(m^{(1)}_{1n}(\mu_1)-m^{(1)}_{1n}(\mu_1^-)).
\end{align*}
Moreover, similarly to (\ref{eq_controlm12}), for $\tilde{\mathcal{M}}_{12}^{(1)}$ we have that with high probability
\begin{align*}
    \tilde{\mathcal{M}}_{12}^{(1)}
    =&\sum_{i=1}^{p-m}\frac{\frac{\sigma_{m+i}(\mu_1-\mu_1^-)}{p\mu_1\mu_1^-}}{(1+\frac{\sigma_{m+i}}{n}\sum_{j=2}^n\frac{\xi^2_{(j)}}{-\mu_1(1+\xi^2_{(j)}m_{1n}^{(1)}(\mu_1^-))})(1+\frac{\sigma_{m+i}}{n}\sum_{j=2}^n\frac{\xi^2_{(j)}}{-\mu_1^-(1+\xi^2_{(j)}m_{1n}^{(1)}(\mu_1^-))})}\\
     \asymp& \frac{\mu_1-\mu_1^-}{\mu_1^-\mu_1}  \asymp n^{-1/2-1/\alpha+\epsilon}.
\end{align*}
This implies that $m^{(1)}_{1n}(\mu_1)-m^{(1)}_{1n}(\lambda) \asymp n^{-1/2-1/\alpha+\epsilon}.$ Together with (\ref{eq_finalpartone}), the definition of $\mathsf{q}$ and Lemma \ref{def_goodconfiguration}, we readily see that 
\begin{equation}\label{eq: lambda-lambda_(1)}
   1+(\xi^2_{(1)}+\mathsf{q})m^{(1)}_{1n}(\mu_1^-)) \asymp -n^{-1/2+\epsilon},
\end{equation}
which concludes the proof of $M(\mu_1^-)<0$ when $n$ is sufficiently large. Similarly, we can prove that $M(\mu_1^+)>0.$ 
\end{proof}
\subsection{Proof Theorem \ref{thm_prf_fakesignals_locallaw}}\label{sec_proof_thmfake_locallaw}
In this section, we will prove Theorem \ref{thm_prf_fakesignals_locallaw}. The proof contains two steps. In the first step, we will establish the local scale results of $S_2$ on the domain $\widetilde{\mathbf{D}}_{\mu}$, where $\widetilde{\mathbf{D}}_{\mu}$ is denoted as follows   
\begin{equation}\label{eq_spectraldomainonereduced}
\widetilde{\mathbf{D}}_{\mu} \equiv \widetilde{\mathbf{D}}_{\mu}(\mathrm{C}_0):=\left\{z=E+\mathrm{i}\eta: 0<E-\mu_1 \leq  \mathrm{C}_0 d_{1}, \ n^{-2/3}\le\eta\le \mathrm{C_0} \mu_1 \right\},
\end{equation}
for some properly chosen constant $\mathrm{C}_0>0$. Then, we will establish the following proposition. 
\begin{proposition}\label{eq_propooursidebulk} 
Under the assumptions of Theorem \ref{thm_prf_fakesignals_locallaw}, the following results hold uniformly on the spectral domain $\widetilde{\mathbf{D}}_{\mu}$ in (\ref{eq_spectraldomainonereduced}) when conditional on the event $\Omega$ in Lemma \ref{lem_goodconfiguration}. 
\begin{enumerate}
\item[(1).] If \eqref{ass_xi_poly} holds, we have that 
\begin{align}\label{eq_entrywiselocallaw}
\mathcal{G}_{ij}(z)=-{\delta_{ij}}[{z(1+m_{1n}(z)\xi_i^2)}]^{-1}+\mathrm{O}_{\prec} \left( n^{-1/2-1/\alpha} \right),
\end{align}
where $\delta_{ij}$ is the Dirac delta function so that $\delta_{ij}=1$ when $i=j$ and $\delta_{ij}=0$ when $i \neq j.$ 
Moreover, we have that 
\begin{equation*}
m_1(z)=m_{1n}(z)+\mathrm{O}_{\prec}\left( n^{-1/2-2/\alpha}\right),\quad  m_2(z)=m_{2n}(z)+\mathrm{O}_{\prec}\left( n^{-1/2-1/\alpha}\right), 
\end{equation*}
and 
\begin{equation}\label{eq_averagedlawused}
m_{S_2}(z)=m_{n}(z)+\mathrm{O}_{\prec} \left( n^{-1/2-2/\alpha} \right).
\end{equation}
\item[(2).] If \eqref{ass_xi_exp} holds, we have that the results in part (1) hold by setting $\alpha=\infty.$
\end{enumerate}  
\end{proposition}
\quad Once Proposition \ref{eq_propooursidebulk} is proved, we can quantify the rough locations of the eigenvalues of $S_2$ as summarized in the following lemma. 

\begin{lemma}\label{lem: upper bound for eigenvalues}
Suppose Assumptions \ref{ass_sigma} and \ref{ass_xi} hold. For some sufficiently large constant $C>0,$ with high probability, for any fixed realization $\{\xi_i^2\} \in \Omega$ where $\Omega$ is introduced in Lemma \ref{lem_goodconfiguration}, for all $1 \leq i \leq \min\{p-m,n\},$ we have that 
\begin{equation}\label{eq_boundestimateone}
\lambda_i(S_2) \notin (\mu_1, C n^{1/\alpha} \log n), \ \text{if \eqref{ass_xi_poly} holds},  
\end{equation}
and 
\begin{equation}\label{eq_boundestimatetwo}
\lambda_i(S_2) \notin (\mu_1, C  \log^{1/\beta} n), \ \text{if \eqref{ass_xi_exp} holds}.  
\end{equation}
\end{lemma}
\begin{proof}
    Due to similarity, we focus our arguments on (\ref{eq_boundestimateone}). We prove the results by an indirect argument. Assume there is an eigenvalue of $S_2$ in the interval as in (\ref{eq_boundestimateone}), denoted as $\widehat{\lambda}$. Let $z=\hat{\lambda}+\mathrm{i} n^{-2/3}.$ Since $z\in\widetilde{\mathbf{D}}_{\mu} \subset \mathbf{D}_{\mu}$ as in (\ref{eq_def_spectraldomain}), by Lemma \ref{lem_prf_fakesignals_asymptoticproperties}, we obtain $\operatorname{Im}m_n(z)=\eta \widehat{\lambda}^{-2}$. According to \eqref{eq_mu1part} and Lemma \ref{lem_goodconfiguration}, we have that on the event $\Omega$
\begin{equation}\label{eq_lowerboundmu1}
\mu_1 \gtrsim n^{1/\alpha} \log^{-1} n.
\end{equation} 
Together with (\ref{eq_averagedlawused}), we readily see that 
\begin{align}\label{eq: value of m_W outside spectrum}
        \operatorname{Im}m_{S_2}(z)&=\operatorname{Im}m_n(z)+\operatorname{Im}(m_{S_2}(z)-m_n(z))\prec n^{-1/2-2/\alpha}.
      \end{align}
On the other hand, we have
\begin{equation*}
\operatorname{Im}m_{S_2}(z)=\frac{1}{n}\sum_i\frac{\eta}{(\lambda_i-\widehat{\lambda})^2+\eta^2}\ge\frac{1}{n\eta}=n^{-1/3},
\end{equation*}
which contradicts \eqref{eq: value of m_W outside spectrum}. Therefore, there is no eigenvalue in this interval. Similarly, we can prove (\ref{eq_boundestimatetwo}). The only difference is that (\ref{eq_lowerboundmu1}) should be replaced by $\mu_1 \gtrsim  \log^{1/\beta} n$ according to Lemma \ref{lem_goodconfiguration} so that the error rate in (\ref{eq: value of m_W outside spectrum}) should be updated to $n^{-1/2}.$ This completes the proof. 
\end{proof}

\quad Armed with the above lemma, we can proceed to the second step to conclude the proof of Theorem \ref{thm_prf_fakesignals_locallaw}.  In what follows, we first provide the proof of Proposition \ref{eq_propooursidebulk} in Section \ref{prop_subsubsubsubsubsub}. After that, we prove Theorem \ref{thm_prf_fakesignals_locallaw} in Section \ref{sec_proofa2}.

\subsubsection{Proof of Proposition \ref{eq_propooursidebulk}}\label{prop_subsubsubsubsubsub}

We first prepare two lemmas. The first one is to establish Proposition \ref{eq_propooursidebulk} for large $\eta.$

\begin{lemma}[Averaged local law for large $\eta$]\label{lem: average local law for large eta} Proposition \ref{eq_propooursidebulk} holds when $\eta=\mathrm{C} \mu_1.$ for some $\mathrm{C}>0$.
\end{lemma}
\begin{proof}
To simplify notation, in what follows, without loss of generality, we assume that $\xi_1^2 \geq \xi_2^2 \geq \cdots \geq \xi_n^2.$ 
 
According to (\ref{eq_mu1part}), Lemma \ref{lem_goodconfiguration} and the definition of $\mathsf{q}$ in \eqref{eq_def_q}, on the event $\Omega,$ we have that 
\begin{equation}\label{eq_Econtrol}
E \asymp \xi_1^2. 
\end{equation}
When $\eta=\mathrm{C} E$, $\max\left\{\|G^{(\mathcal{T})}\|,\|\mathcal{G}^{(\mathcal{T})}\|\right\}\le\eta^{-1}=\mathrm{C}^{-1} E^{-1}$ for any finite $\mathcal{T}\subset\{1,\dots, n\}$ by definition.  The main idea is to explore the relation of $m_1$ and $m_2$ using Lemma \ref{lem_resolvent}. We start with $m_2.$ By Lemma \ref{lem_resolvent} and definition of $m_2$ in \eqref{eq_def_m12}, we have
\begin{gather}\label{eq: decomp m_2}
    {
    m_2=\frac{1}{p}\sum_{i=1}^n\frac{\xi^2_i}{-z-z\mathbf{y}^{\prime}_{i,2} G^{(i)}\mathbf{y}_{i,2}}=\frac{1}{p}\sum_{i=1}^n\frac{\xi^2_i}{-z(1+\xi^2_ip^{-1}{\rm tr}G^{(i)}\Sigma_2+Z_i)},}\\
    Z_i=\mathbf{y}^{\prime}_{i,2} G^{(i)}\mathbf{y}_{i,2}-\xi^2_ip^{-1}{\rm tr}G^{(i)}\Sigma_2. \nonumber
\end{gather}
As $\mathbf{y}_{i,2}$ is independent of $G^{(i)}$, by $(2)$ of Lemma \ref{lem_largedeviationbound}, we see that 
\begin{equation}\label{eq: bound for Z}
    Z_i\prec\frac{\xi^2_i}{p}\|G^{(i)}\Sigma_2\|_F\le\frac{\xi^2_i}{p}\|G^{(i)}\|\|\Sigma_2\|_F\prec\frac{\xi^2_i}{\sqrt{n}\eta}.
\end{equation}
Moreover, using the definition of $m_1$ in \eqref{eq_def_m12} and the second resolvent identity, we readily obtain that for some constant $C>0$ 
\begin{equation}\label{eq_m1approximate}
    \frac{1}{p}{\rm tr}(G^{(i)}\Sigma_2)-m_1(z)=\frac{1}{p}\mathbf{y}_{i,2}^{\prime}G \Sigma_2 G^{(i)}\mathbf{y}_{i,2}\le C\frac{\xi^2_i}{p \eta^2}.
\end{equation}

Moreover, by (\ref{eq_Econtrol}) and the form of $\eta,$ we find that for some constant $C>0$, 
$
|1+\xi^2_i m_1|\ge 1-C \mathrm{C}^{-1}>0, 
$ 
when $\mathrm{C}>0$ is chosen to be sufficiently large. 
Together with (\ref{eq: decomp m_2}), we obtain
\begin{equation}\label{eq: m_2 by m_1}
    m_2=\frac{1}{p}\sum_{i=1}^n\frac{\xi^2_i}{-z(1+\xi^2_im_1+\mathrm{O}_{\prec}({\xi^2_i}{{n}^{-1/2}\eta^{-1}}))}=\frac{1}{p}\sum_{i=1}^n\frac{\xi^2_i}{-z(1+\xi^2_im_1)}+\mathrm{O}_{\prec}(n^{-1/2-1/\alpha}),
\end{equation}
where we used again Lemma \ref{lem_goodconfiguration}. 

Then we work with $m_1.$ Decompose $S_2-zI$ as
\[
S_2-zI=\sum_{i=1}^n\mathbf{y}_{i,}\mathbf{y}_{i,2}^{\prime}+zm_2(z)\Sigma_2-z(I+m_2(z)\Sigma_2).
\]
Applying resolvent expansion to the order one, we obtain that
\[
G=-z^{-1}(I+m_2(z)\Sigma)^{-1}+z^{-1}G(\sum_{i=1}^n\mathbf{y}_{i,2}\mathbf{y}_{i,2}^{\prime}+zm_2(z)\Sigma_2)(I+m_2(z)\Sigma_2)^{-1}.
\]
Furthermore, using the Shernman-Morrison formula, we have that
\begin{equation}\label{eq_usefulformula}
G\mathbf{y}_{i,2}={G^{(i)}\mathbf{y}_{i,2}}({1+\mathbf{y}_{i,2}^{\prime} G^{(i)}\mathbf{y}_{i,2}})^{-1}.
\end{equation}
Combining the above two identities and Lemma \ref{lem_resolvent}, we can further write
\begin{align}\label{eq: decomp of mathcal_G}
    G&=-z^{-1}(I+m_2(z)\Sigma_2)^{-1}+\left[z^{-1}\sum_{i=1}^n\frac{G^{(i)}(\mathbf{y}_{i,2}\mathbf{y}_{i,2}^{\prime}-p^{-1}\xi^2_i\Sigma_2)}{1+\mathbf{y}_{i,2}^{\prime}G^{(i)}\mathbf{y}_{i,2}}(I+m_2(z)\Sigma_2)^{-1} \right] \nonumber\\
    &+\left[z^{-1}\frac{1}{p}\sum_{i=1}^n\frac{(G^{(i)}-G)\xi^2_i\Sigma_2}{1+\mathbf{y}_{i,2}^{\prime}G^{(i)}\mathbf{y}_{i,2}}(I+m_2(z)\Sigma_2)^{-1} \right]\nonumber\\
    &:=-z^{-1}(I+m_2(z)\Sigma_2)^{-1}+R_1+R_2.
\end{align}
In what follows, we control the two error terms $R_1,R_2.$ For $R_1$, we notice that 
\begin{align}\label{eq: decomp R_1}
    &\frac{z}{p}{\rm tr}(R_1\Sigma_2)=\frac{1}{p}\sum_i{\rm tr}\left(\frac{G^{(i)}(\mathbf{y}_{i,2}\mathbf{y}_{i,2}^{\prime}p^{-1}\xi^2_i\Sigma_2)}{1+\mathbf{y}_{i,2}^{\prime} G^{(i)}\mathbf{y}_{i,2}}(I+m_2^{(i)}\Sigma_2)^{-1}\Sigma_2\right) \nonumber\\
    &+\frac{1}{p}\sum_i{\rm tr}\left(\frac{ G^{(i)}(\mathbf{y}_{i,2}\mathbf{y}_{i,2}^{\prime}-p^{-1}\xi^2_i\Sigma_2)}{1+\mathbf{y}_{i,2}^{\prime} G^{(i)}\mathbf{y}_{i,2}}(I+m_2\Sigma_2)^{-1}(m_2^{(i)}-m_2)\Sigma_2(I+m^{(i)}_2\Sigma_2)^{-1}\Sigma_2\right)\nonumber\\
    &:=\mathtt R_{11}+\mathtt R_{12}.
\end{align}
Since $\|\mathcal G^{(i)}\| \le\eta^{-1}$, using (\ref{eq_Econtrol}), with high probability, we have that for some constant $C>0,$ 
\begin{equation}\label{eq_eqcontrolcontrolcontrolcontrol}
|m_2^{(i)}(z)|\le\frac{1}{p}\sum_{j\neq i}\xi^2_j| \mathcal G^{(i)}_{jj}|\le C \frac{\log^2 n}{n^{1/\alpha}}. 
\end{equation}
Moreover, according to (\ref{eq: bound for Z}), with high probability, when $n$ is sufficiently large, we have that for some constant $c>0$
\begin{equation}\label{eq_controlower}
|1+\mathbf{y}_{i,2}^{\prime} G^{(i)}\mathbf{y}_{i,2}| \asymp  |1+\xi_i^2 p^{-1} {\rm tr} G^{(i)} \Sigma_2 | \geq 1- c C^{-1}>0,
\end{equation}
whenever $C>0$ is chosen sufficiently large. Consequently, for all $i$, we have that 
\begin{align}\label{eq_controncontrolcontrol}
&   {\rm tr}\left(\frac{ G^{(i)}(\mathbf{y}_{i,2}\mathbf{y}_{i,2}^{\prime}-p^{-1}\xi^2_i\Sigma_2)}{1+\mathbf{y}_{i,2}^{\prime} G^{(i)}\mathbf{y}_{i,2}}(I+m_2^{(i)}\Sigma_2)^{-1}\Sigma_2\right) \nonumber\\
&\asymp {\rm tr} \left(\xi^2_i G^{(i)}(\mathbf{u}_i\mathbf{u}_i^{\prime}-p^{-1}I) (I+m_2^{(i)}\Sigma_2)^{-1} \Sigma_2^2 \right)= \xi_i^2\mathbf{u}_i^{\prime} G^{(i)} (I+m_2^{(i)}\Sigma_2)^{-1} \Sigma_2^2 \mathbf{u}_i \nonumber\\&~~~~~~~~~~~~-\xi_i^2p^{-1} {\rm tr} \left( G^{(i)} (I+m_2^{(i)}\Sigma_2)^{-1} \Sigma_2^2 \right)  \prec \xi_i^2 {\eta^{-1}{p}^{-1/2}}, 
\end{align}
where in the third step we used (1) of Lemma \ref{lem_largedeviationbound}.  Together with (\ref{eq_Econtrol}), we find that 
\begin{equation*}
\mathtt{R}_{11} \prec n^{-1/2-1/\alpha},
\end{equation*}
where we used the fact that $p/n\rightarrow\phi<c^{-1}$ from Assumption \ref{ass_sigma}. For $\mathtt R_{12}$, using the definition in Lemma \ref{lem_resolvent} and the definition of $\mathcal{G}^{(i)}$, we see that  
\begin{equation}\label{eq_decompositionleavoneout}
m_2(z)-m^{(i)}_2(z)=\frac{1}{p}\sum_{j=1}^n \xi^2_j(\mathcal G_{jj}-\mathcal G^{(i)}_{jj})=\frac{1}{p}\sum_{j\neq i}\xi^2_j\frac{ \mathcal G_{ji} \mathcal G_{ij}}{\mathcal G_{ii}}+\frac{\xi_i^2(\mathcal{G}_{ii}-|z|^{-1})}{p}.
\end{equation}
In addition, using Lemma \ref{lem_resolvent} and a discussion similar to (\ref{eq_controlower}), we conclude that  
\begin{equation*}
 {\mathcal G_{ii}^{-1}(z)}=-z-z\mathbf{y}_{i,2}^{\prime} G^{(i)}\mathbf{y}_{i,2}\prec |z|.
\end{equation*}
Moreover, by Lemmas \ref{lem_resolvent} and \ref{lem_largedeviationbound}, we have that 
\begin{equation*}  
  \mathcal G_{ij}(z)=z\mathcal G_{ii}(z)\mathcal G^{(i)}_{jj}(z)\mathbf{y}_{i,2}^{\prime} G^{(ij)}\mathbf{y}_{j,2}\prec |z|\eta^{-2} |\xi_i\xi_j| p^{-1}\|\mathcal{G}^{(ij)}\|_F\prec p^{-1/2}|z|\eta^{-3}|\xi_i\xi_j| ,\quad i\neq j.  
\end{equation*}
Combining the above bounds with the deterministic bounds for $\xi_i$'s under $\Omega$, we see that $m_2(z)-m^{(i)}_2(z) \prec p^{-1-1/\alpha}.$ Together with (\ref{eq_controlower}), (\ref{eq_controncontrolcontrol}) and noticing that $p/n\rightarrow\phi$ we arrive at 
\begin{equation*}
\mathtt R_{12} \prec  {\eta^{-2} n^{-3/2}}.  
\end{equation*}
Using the above bounds, we see that $\frac{z}{p}{\rm tr}(R_1\Sigma)\prec n^{-1/2-1/\alpha}.$ For $R_2$, applying the Sherman--Morrison formula to $((G^{(i)})^{-1}+\mathbf{y}_{i,2} \mathbf{y}_{i,2}^{\prime})^{-1},$ we obtain that 
\begin{align}\label{eq_simimimimimi}
    \frac{1}{p}\left|{\rm tr}\left(\frac{(G^{(i)}-G)\Sigma_2(I+m_2\Sigma_2)^{-1}\Sigma_2}{1+\mathbf{y}_{i,2}^{\prime} G^{(i)}\mathbf{y}_{i,2}}\right) \right|&=\frac{1}{p}\left|\frac{\mathbf{y}_{i,2}^{\prime} G^{(i)}\Sigma_2(I+m_2\Sigma_2)^{-1}\Sigma_2 G \mathbf{y}_{i,2}}{1+\mathbf{y}_{i,2}^{\prime}G^{(i)}\mathbf{y}_{i,2}}\right|\prec \frac{\xi_i^2}{p \eta^2},
\end{align}
where in the second step we used Lemma \ref{lem_largedeviationbound} and (\ref{eq_controlower}) and a discussion similar to (\ref{eq_eqcontrolcontrolcontrolcontrol}). Together with the definition of $R_2$ in (\ref{eq: decomp of mathcal_G}), under $\Omega$, we find that
\[
\begin{split}
  \frac{z}{p} \left|{\rm tr}(R_2\Sigma_2) \right|
  \le\frac{1}{p^2}\sum_i\frac{\xi^2_i}{\eta^2 }\prec n^{-1-2/\alpha}.
\end{split}
\]
As a result, in light of the definition $m_1$ in \eqref{eq_def_m12},
we have
\begin{align}\label{eq: m_1 by m_2}
    m_1&=\frac{1}{p}{\rm tr}(G(z)\Sigma)=-z^{-1}\frac{1}{p}{\rm tr}((I+m_2(z)\Sigma_2)^{-1}\Sigma_2)+\mathrm{O}_{\prec}(n^{-1/2-2/\alpha})\nonumber\\
    &=-\frac{1}{p}\sum_{i=1}^{p-m}\frac{\sigma_{m+i}}{z(1+m_2\sigma_{m+i})}+\mathrm{O}_{\prec}(n^{-1/2-2/\alpha}).
\end{align}

We first control $m_2(z)-m_{2n}(z).$ Recall the definition of $m_{2n}(z)$ in \eqref{eq_systemequations}. Combing \eqref{eq: m_2 by m_1} and \eqref{eq: m_1 by m_2},  we have that 
\begin{align}\label{eq_closenessm2m2n}
    &m_2(z)-m_{2n}(z)\nonumber\\
    =&\frac{1}{p}\sum_{i=1}^n \left(\frac{\xi^2_i}{-z(1+\xi^2_im_1)}+\frac{\xi^2_i}{z(1+\xi^2_im_{1n})} \right)+\mathrm{O}_{\prec}(n^{-1/2-1/\alpha})\\
    =&\frac{1}{p}\sum_{i=1}^n\frac{\xi^4_i(m_1-m_{1n})}{z(1+\xi^2_im_1)(1+\xi^2_im_{1n})}+\mathrm{O}_{\prec}(n^{-1/2-1/\alpha})\nonumber \\
    =&\left(\frac{1}{p}\sum_{i=1}^n\frac{\xi^4_i}{z(1+\xi^2_im_1)(1+\xi^2_im_{1n})}\right)\left(\frac{1}{p}\sum_{i=1}^{p-m}\frac{\sigma_{m+i}^2(m_2-m_{2n})}{z(1+\sigma_{m+i}m_2)(1+\sigma_{m+i}m_{2n})}\right)\\
    &+\mathrm{O}_{\prec}(n^{-1/2-1/\alpha}). \nonumber
\end{align}
By a discussion similar to (\ref{eq_eqcontrolcontrolcontrolcontrol}) and (\ref{eq_controlower}) with high probability, when $n$ is sufficiently large, we have that  
\[
\begin{split}
  |m_2-m_{2n}|&=\mathrm{O}\left(\frac{1}{p}\sum_{i=1}^n\frac{\xi^4_i}{z^2}|m_2-m_{2n}| \right)+\mathrm{O}_{\prec}(n^{-1/2-1/\alpha})\\
  &=\mathrm{O}(n^{-2/\alpha}|m_2-m_{2n}|)+\mathrm{O}_{\prec}(n^{-1/2-1/\alpha}),
\end{split}
\]
where in the second step we used the upper bounds for $\xi^2_i$'s under $\Omega$ and the fact that $p/n\rightarrow\phi$. Then we can conclude that $m_2-m_{2n}\prec n^{-1/2-1/\alpha}$. By a similar procedure, we also have $m_1-m_{1n}\prec n^{-1/2-2/\alpha}$.

Armed with the above two results, we proceed to finish the rest of the proof. Recall $m_{S_2}$ in \eqref{eq_def_mS2}. Using (\ref{eq: decomp of mathcal_G}) and a discussion similar to \eqref{eq: m_1 by m_2}, one can see that 
\begin{align}\label{eq_verbalttt}
    m_{S_2}&=\frac{1}{p}{\rm tr}(G(z))=-\frac{1}{p}\sum_{i=1}^{p-m}\frac{1}{z(1+m_2\sigma_{m+i})}+\mathrm{O}_{\prec}(n^{-1/2-2/\alpha}) \nonumber\\
    &=-\frac{1}{p}\sum_{i=1}^{p-m}\frac{1}{z(1+m_{2n}\sigma_{m+i})}+\frac{1}{p}\sum_{i=1}^{p-m}\frac{(m_2-m_{2n})\sigma_{m+i}}{z(1+m_2\sigma_{m+i})(1+m_{2n}\sigma_{m+i})}+\mathrm{O}_{\prec}(n^{-1/2-2/\alpha}) \nonumber\\
    &=m_{n}+\mathrm{O}_{\prec}(n^{-1/2-2/\alpha}),
\end{align}
where in the last step we recall $m_n(z)$ in \eqref{eq_systemequations}. Finally, for the diagonal entries of $\mathcal{G}$, by Lemma \ref{lem_resolvent} and a discussion similar to  \eqref{eq: bound for Z} and \ref{eq_m1approximate},  we have
\begin{align*}
   1/ \mathcal G_{ii}&=-{z(1+\mathbf{y}^{\prime}_{i,2} G^{(i)}\mathbf{y}_{i,2})}=-{z(1+\xi^2_ip^{-1}{\rm tr} G^{(i)}\Sigma_2+\mathrm{O}_{\prec}(\frac{\xi^2_i}{\sqrt{p}\eta}))}\\
    &=-\frac{1}{z(1+\xi^2_im_1+\mathrm{O}_{\prec}(\frac{\xi^2_i}{\sqrt{q}\eta}))}=-\frac{1}{z(1+\xi^2_im_{1n})}+\mathrm{O}_{\prec}(p^{-1/2-1/\alpha}).
\end{align*}
For off-diagonal entries, together with Lemmas \ref{lem_resolvent} and \ref{lem_largedeviationbound}, we have
\[
\begin{split}
    |\mathcal G_{ij}|&\le|z||\mathcal G_{ii}||\mathcal G^{(i)}_{ii}||\mathbf{y}^{\prime}_{i,2} G^{(ij)}\mathbf{y}_{j,2}| \prec n^{-1/2-2/\alpha}. 
\end{split}
\]
This completes the proof. For the case \eqref{ass_xi_exp}, the main difference is to use the estimates $(b)$ in Definition \ref{def_goodconfiguration} instead of $(a)$ whenever it is needed, for example, (\ref{eq: decomp m_2}). We omit further details.
\end{proof}

The second component is to prove Proposition \ref{eq_propooursidebulk} under a priori control of the resolvent which is summarized in the following lemma. 

\begin{lemma} \label{lem: self improvement}
Proposition \ref{eq_propooursidebulk} holds if (\ref{eq_entrywiselocallaw}) holds uniformly for $z \in \widetilde{\mathbf{D}}_{\mu}.$ 
\end{lemma}
\begin{proof}
According to the priori control (\ref{eq_entrywiselocallaw}), we have that for $1 \leq i \neq j \leq n$ 
\begin{gather}\label{eq_expansionone}
\mathcal G_{ii}=\frac{1}{z(1+\xi^2_im_{1n}(z)}+\mathrm{O}_{\prec}(n^{-1/2-1/\alpha}), \quad \mathcal G_{ij} =\mathrm{O}_{\prec}(n^{-1/2-1/\alpha}).
\end{gather}
For the diagonal entries, when $i=1,$ using \eqref{eq_prf_fakesignals_defmu1}, we observe that
\begin{align}\label{eq_g11control}
        \mathcal G_{11}&=-\frac{1}{z(1+\xi^2_{1}m_{1n}(z))}+\mathrm{O}_{\prec}(n^{-1/2-1/\alpha}) \nonumber\\
        &=-\frac{1}{z(1+\xi^2_{1}m_{1n}(\mu_1))}+\frac{z\xi_{1}^2(m_{1n}(z)-m_{1n}(\mu_1))}{(z(1+\xi^2_{1}m_{1n}(\mu_1)))(z(1+\xi^2_{1}m_{1n}(z)))}+\mathrm{O}_{\prec}(n^{-1/2-1/\alpha}) \nonumber\\
        &=\frac{1}{z\mathsf{q}m_{1n}(\mu_1)}-\frac{z\xi_{1}^2(m_{1n}(z)-m_{1n}(\mu_1))}{z\mathsf{q}m_{1n}(\mu_1)}\left(\mathcal G_{11}+\mathrm{O}_{\prec}(n^{-1/2-2/\alpha})\right)+\mathrm{O}_{\prec}(n^{-1/2-1/\alpha}) \nonumber\\
        &=\frac{1}{z\mathsf{q}m_{1n}(\mu_1)}-\frac{z\xi_{1}^2}{z\mathsf{q}m_{1n}(\mu_1)}(\mathcal G_{11}+\mathrm{O}_{\prec}(n^{-1/2-1/\alpha}))\times \mathrm{O}_{\prec}(n^{-1/\alpha})+\mathrm{O}_{\prec}(n^{-1/2-1/\alpha}),
\end{align}
where in the fourth step we used Lemma \ref{lem_prf_fakesignals_asymptoticproperties}. By Lemma \ref{lem_prf_fakesignals_asymptoticproperties}, (\ref{eq_mu1part}) and under $\Omega$,  this yields that for some constant $C>0$
\begin{align}\label{eq_g11}
    |\mathcal G_{11}|&=\frac{1}{|z\mathsf{q}m_{1n}(\mu_1)|}+\mathrm{O}_{\prec}(n^{-1/2-1/\alpha})
    =\frac{C}{\mathsf{q}}+\mathrm{O}_{\prec}(n^{-1/2-1/\alpha}).
\end{align}

Similarly, when $2 \leq i \leq n,$ by the definition of $\mathsf{q},$ using Lemma \ref{lem_prf_fakesignals_asymptoticproperties}, we see that
\begin{align}\label{eq_expansiontwo}
\mathcal{G}_{ii}=\mathrm{O}_{\prec}(n^{-1/\alpha}). 
\end{align}

We also provide some basic controls for the matrix $\mathcal G^{(i)}$ for all $1 \leq i \leq n.$ By definition and an elementary calculation, it is not hard to see that
\begin{align}\label{eq_trivialcontrol}
\mathcal G_{ii}^{(i)}=-z^{-1}; \ \mathcal{G}_{i k}^{(i)}=0, \  1 \leq k \neq i \leq n. 
\end{align}  
Moreover, using (\ref{eq_expansionone}), (\ref{eq_expansiontwo}) and the third identity of Lemma \ref{lem_resolvent}, we find that for $1 \leq i \leq n,$ 
\begin{align}\label{eq_nontrivialcontrol}
\mathcal G_{kk}^{(i)}=\mathrm{O}_{\prec}(n^{-1/\alpha}),  \ k \neq i; \ \mathcal G_{kl}^{(i)}=\mathrm{O}_{\prec}(n^{-1/2-1/\alpha}), \ k,l \neq i.
\end{align}

With the above preparations in place, we now proceed to control $Z_i$ in (\ref{eq: decomp m_2}). Unlike the approach in (\ref{eq: bound for Z}), we exploit the fact that $S_2$ and $\mathcal{S}_2$ share the same non-zero eigenvalues, allowing us to control $Z_i$ using the bounds established above.
\begin{align}\label{eq_standarddiscussion}
     Z_i&\prec\frac{\xi^2_i}{p}\| G^{(i)}\Sigma_2\|_F\prec\frac{\xi^2_i}{p}\|G^{(i)}\|_F=\frac{\xi^2_i}{p}(\operatorname{tr}((G^{(i)})^2))^{1/2}\le\frac{\xi^2_i}{p}\operatorname{tr}((\mathcal G^{(i)})^2)^{1/2}+\frac{\xi_i^2}{p}\frac{\sqrt{|n-p|}}{|z|} \nonumber\\
    & \asymp \frac{\xi^2_i}{p}\|\mathcal G^{(i)}\|_F+\frac{\xi^2_i}{p}\frac{n^{1/2}}{p^{1/\alpha}}=\frac{\xi^2_i}{p} \left((\mathcal G_{ii}^{(i)})^2+\sum_{j\neq i}(\mathcal G^{(i)}_{jj})^2+\sum_{j\neq k\neq i}(\mathcal G^{(i)}_{jk})^2 \right)^{1/2}+\frac{\xi^2_i}{p}\frac{n^{1/2}}{p^{1/\alpha}}  \nonumber\\
& \prec        \frac{\xi_1^2}{p}\left( |z|^{-2}+n p^{-2/\alpha}+n^2 p^{-1-2/\alpha}\right)^{1/2}+\frac{\xi^2_i}{p}\frac{n^{1/2}}{p^{1/\alpha}}
      \prec\frac{\xi_i^2}{n^{1/2+1/\alpha}},
\end{align}
where in the third and fourth steps we used (\ref{eq_trivialcontrol}) and (\ref{eq_nontrivialcontrol}). 

Besides, by a discussion similar to (\ref{eq_standarddiscussion}),  we now have from Lemma \ref{lem_largedeviationbound} that 
\begin{align}\label{eq_cccccc11111}
   T_i:&=\frac{1}{p}{\rm tr} G^{(i)}\Sigma_2-m_1(z)=\frac{1}{p}\mathbf{y}_{i,2}^{\prime}G \Sigma_2 G^{(i)}\mathbf{y}_{i,2}=\frac{1}{p}\frac{\mathbf{y}^{\prime}_{i,2} G^{(i)}\Sigma_2 G^{(i)}\mathbf{y}_{i,2}}{1+\mathbf{y}^{\prime}_{i,2} G^{(i)}\mathbf{y}_{i,2}} \nonumber\\
    &\prec\frac{\xi^2_i}{p}\frac{p^{-1}\| G^{(i)}\|^2_F}{|1+\mathbf{y}^{\prime}_{i,2} G^{(i)}\mathbf{y}_{i,2}|}\prec\frac{\xi^2_i}{p^2}|z||\mathcal G_{ii}|\left(\|\mathcal G^{(i)}\|^2_F+\frac{n}{p^{2/\alpha}} \right) \prec \frac{\xi_i^2}{n^{1+2/\alpha}},
\end{align}
where in the second step we used the relation (\ref{eq_usefulformula}), in the third step we used Lemma \ref{lem_resolvent} and in the last two steps we used a discussion similar to (\ref{eq_standarddiscussion}). 

With the above control, we now use an idea similar to the proof of  Lemma \ref{lem: average local law for large eta} to conclude the proof. The key ingredient is to explore the relation of $m_1$ and $m_2.$ We start with $m_2.$ Using the above notations and Lemma \ref{lem_resolvent}, we find that  
\begin{equation}\label{eq_eeeone}
{-z(1+\xi_i^2 m_1(z))}={\mathcal{G}_{ii}^{-1}+z(Z_i+T_i)}.
\end{equation}
Consequently, by (\ref{eq_g11}), (\ref{eq_expansiontwo}), (\ref{eq_standarddiscussion}) and (\ref{eq_cccccc11111}), we see that 
\begin{equation}\label{eq_eeetwo}
\frac{1}{-z(1+\xi_i^2 m_1(z))} \prec n^{-1/\alpha}. 
\end{equation}
Then using the decomposition as in (\ref{eq: decomp m_2}), we have that 
\begin{align}\label{eq_modeltemp}
        m_2&=\frac{1}{p}\frac{\xi^2_1}{-z(1+\xi^2_1p^{-1}\operatorname{tr} G^{(1)}\Sigma_2+Z_1)}+\frac{1}{p}\sum_{i=2}^n\frac{\xi^2_i}{-z(1+\xi^2_ip^{-1}\operatorname{tr} G^{(i)}\Sigma_2+Z_i)} \nonumber\\
        &=\frac{1}{p}\frac{\xi^2_1}{-z(1+\xi^2_1m_1(z)+\xi_1^2 n^{-1-2/\alpha}+Z_1)}+\frac{1}{p}\sum_{i=2}^n\frac{\xi^2_i}{-z(1+\xi^2_im_1(z)+\xi_i^2 n^{-1-2/\alpha}+Z_i)}\nonumber\\
        &=\frac{1}{p}\sum_{i=1}^n\frac{\xi^2_i}{-z(1+\xi^2_im_1(z))}+\mathrm{O}_{\prec}\left(n^{-3/2}+\frac{1}{p}\sum_{i=2}^n\frac{\xi^4_i}{|z|n^{1/2+1/\alpha}}\right)\nonumber\\
        &=\frac{1}{p}\sum_{i=1}^n\frac{\xi^2_i}{-z(1+\xi^2_im_1(z))}+\mathrm{O}_{\prec}(n^{-1/2-1/\alpha}),
\end{align}
where in the second step we used (\ref{eq_cccccc11111}), in the third step we used a discussion similar to (\ref{eq_g11control}) and in the last step we used the bounds for $\{\xi^2_i\}$'s under $\Omega$. Then we study $m_1$ using the arguments between (\ref{eq: decomp of mathcal_G}) and (\ref{eq: m_1 by m_2}). We provide the key ingredients as follows. First, for $\mathtt{R}_{11}$ in (\ref{eq: decomp R_1}), by definition of $m_2^{(i)}$ and (\ref{eq_trivialcontrol}), we have that   
\begin{align*}
 &  \left |m^{(i)}_2 \right |\le\frac{1}{p} \left(\sum_{j\neq i}\xi^2_j| \mathcal G^{(i)}_{jj}|+|z|^{-1} \right)
    \le\frac{1}{p} \left[\sum_{j\neq 1}\xi^2_j\left( | \mathcal G_{jj}|+\frac{|\mathcal G_{j1}||\mathcal G_{1j}|}{|\mathcal G_{11}|} \right)+|z|^{-1} \right] \prec n^{-1/\alpha},
\end{align*}
where in the second step we used Lemma \ref{lem_resolvent} and in the last step we used (\ref{eq_expansionone}), (\ref{eq_expansiontwo}) and event $\Omega$. Moreover, by Lemma \ref{lem_resolvent} and (\ref{eq_expansiontwo}), we find that $(1+\mathbf{y}_{i,2}^{\prime} G^{(i)} \mathbf{y}_{i,2})^{-1} \prec 1.$ Therefore, we conclude that for all $1 \leq i \leq n,$
\begin{align*}
  &   {\rm tr}\left(\frac{ G^{(i)}(\mathbf{y}_{i,2}\mathbf{y}_{i,2}^{\prime}-n^{-1}\xi^2_i\Sigma)}{1+\mathbf{y}_{i,2}^{\prime} G^{(i)}\mathbf{y}_{i,2}}(I+m_2^{(i)}\Sigma_2)^{-1}\Sigma_2\right) \\
  &\asymp {\rm tr} \left(\xi^2_i G^{(i)}(\mathbf{u}_i\mathbf{u}_i^{\prime}-p^{-1}I) (I+m_2^{(i)}\Sigma_2)^{-1} \Sigma_2^2 \right)\\
   &= \xi_i^2 \left( \mathbf{u}_i^{\prime} G^{(i)} (I+m_2^{(i)}\Sigma_2)^{-1} \Sigma^2 \mathbf{u}_i -p^{-1} {\rm tr} \left( G^{(i)} (I+m_2^{(i)}\Sigma_2)^{-1} \Sigma_2^2 \right) \right)\\
 & \prec \frac{\xi_i^2}{p}\|G^{(i)}\|_F\prec\frac{\xi_i^2}{n^{1/2+1/\alpha}},
\end{align*}
where in the last step we used a discussion similar to (\ref{eq_standarddiscussion}) and $p/n\rightarrow\phi$. Under $\Omega$, we can conclude that $\mathtt R_{11} \prec n^{-1/2-1/\alpha}.$ For $\mathtt R_{12},$ using (\ref{eq_decompositionleavoneout}), (\ref{eq_expansionone}) and (\ref{eq_expansiontwo}) 
\begin{equation*}
m_2-m^{(i)}_2 \prec\frac{1}{p}\sum_{j\neq i}\frac{\xi^2_in^{-1-2/\alpha}}{n^{-1/\alpha}} + \frac{1}{p}\prec n^{-1}.
\end{equation*}
Then by an argument similar to (\ref{eq_simimimimimi}), we can conclude that $\mathtt R_{12} \prec n^{-1-1/\alpha}.$ Similarly, for $R_2,$ we have that
\begin{align*}
    &\frac{1}{p} \left|{\rm tr}\left(\frac{(G^{(i)}-G)\Sigma(I+m_2\Sigma_2)^{-1}\Sigma_2}{1+\mathbf{y}_{i,2}^{\prime} G^{(i)}\mathbf{y}_{i,2}}\right) \right|=\frac{1}{p}\left|\frac{\mathbf{y}_{i,2}^{\prime} G^{(i)}\Sigma_2(I+m_2\Sigma_2)^{-1}\Sigma_2 G\mathbf{y}_{i,2}}{1+\mathbf{y}_{i,2}^{\prime} G^{(i)}\mathbf{y}_{i,2}}\right|\\
    &\asymp\frac{1}{p}\left|\mathbf{y}_{i,2}^{\prime} G^{(i)}\Sigma_2(I+m_2\Sigma_2)^{-1}\Sigma_2 G^{(i)}\mathbf{y}_{i,2} \right | \prec\frac{\xi_i^2}{p^2}\| G^{(i)}\|_F \| G \|_F \prec\frac{\xi_i^2}{n^{1+2/\alpha}}.
\end{align*}
Consequently, we have that $\frac{z}{p} \left|{\rm tr}(R_2\Sigma_2) \right|\prec\frac{1}{p}\sum_i\frac{\xi_i^2}{n^{1+1/\alpha}}\prec n^{-1-1/\alpha}.$

Combining all the above arguments, we find that \eqref{eq: m_1 by m_2} still holds true. Armed with all the above controls, using an argument similar to the discussions between (\ref{eq_closenessm2m2n}) and (\ref{eq_verbalttt}), we can conclude the proof. 
\end{proof}

Via the above two lemmas, we now proceed to the proof of Proposition \ref{eq_propooursidebulk}.  We will use a continuity argument as in \cite[Lemma A.12]{Ding&Yang2018} or \cite[Section 4.1]{Alex2014}. In fact, our discussion is easier since the real part in the spectral domain $\widetilde{\mathbf{D}}_{\mu}$ is divergent so that the rate is independent of $\eta$. Due to similarity, we focus on explaining the key ingredients. 

\begin{proof}[\bf Proof of Proposition \ref{eq_propooursidebulk}]
 For each $z=E+\mathrm{i}\eta\in\widetilde{\mathbf{D}}_{\mu}$, we fix the real part and construct a sequence $\{\eta_j\}$ by setting $\eta_j=\mathrm{C} \mu_1-j n^{-3}$. Then it is clear that $\eta$ falls in an interval $[\eta_{j-1},\eta_j]$ for some $0\le j\le Cn^{1/\alpha+3}$, where $C>0$ is some constant.

In Lemma \ref{lem: average local law for large eta}, we have proved that the results hold  for $\eta_0$. Now we assume (\ref{eq_entrywiselocallaw}) holds for some $\eta_j$. Then according to Lemma \ref{lem: self improvement}, we have that
\begin{equation*}
|m_1(z_j)-m_{1n}(z_j)|+|m_{S_2}(z_j)-m_n(z_j)| \prec n^{-1/2-2/\alpha}, \ |m_2(z_j)-m_{2n}(z_j)| \prec n^{-1/2-1/\alpha}. 
\end{equation*}
For any $\eta^{\prime}$ lying in the interval $[\eta_{k-1},\eta_k]$, denote $z^{\prime}=E+\mathrm{i}\eta^{\prime}$ and $z_j=E+\mathrm{i}\eta_j$. According to the first resolvent identity, we have that 
 \begin{equation} \label{eq: Lip for mathcal G}
 \| \mathcal G(z')-\mathcal{G}(z_j) \| \leq n^{-3} \| \mathcal{G}(z') \| \| \mathcal{G}(z_j) \| \prec n^{-11/6-1/\alpha},
\end{equation}  
where in the second step we used the basic bound $\|\mathcal G(z^{\prime})\|\leq n^{2/3},$ (\ref{eq_expansionone}), (\ref{eq_expansiontwo}) and the Gershgorin circle theorem.  

On the one hand, according to the definitions in \eqref{eq_def_m12}, using the first resolvent identity, we have that  
\begin{align*}
    &m_1(z^{\prime})-m_1(z_j)=\frac{1}{p}{\rm tr}[(G(z^{\prime})-G(z_j))\Sigma_2]\\
    =&\frac{1}{p^4}{\rm tr}(G(z^{\prime})G(z_j)\Sigma_2) \prec\frac{1}{p^4 \eta_j}\|G(z_j)\|_F \prec n^{-17/6-1/\alpha},
\end{align*}
where in the last step we used a discussion similar to (\ref{eq_standarddiscussion}). Similarly, using (\ref{eq: Lip for mathcal G}) and under $\Omega$,  we have that  
\begin{equation*}
    m_2(z^{\prime})-m_2(z_j)=\frac{1}{p}\sum_{i=1}^n\xi^2_i(\mathcal G_{ii}(z^{\prime})-\mathcal G_{ii}(z_j))
\prec n^{-11/6-1/\alpha},
\end{equation*}
and by a discussion similar to (\ref{eq_standarddiscussion})
\begin{align*}
|m_{S_2}(z')-m_{S_2}(z_j)|& =\frac{1}{p} \left| \operatorname{tr} \left( G(z')-G(z_j) \right) \right|  \leq p^{-4} \|G(z') \|_F \| G(z_j) \|_F \\
& \prec p^{-4} (n^{1-2/\alpha})^{1/2} n^{1/2+2/3}=n^{-7/3-1/\alpha}. 
\end{align*}
 On the other hand, using the definitions in \eqref{eq_systemequations}, we decompose that 
\begin{align*}
    m_{1n}(z^{\prime})-m_{1n}(z_j)
    &=\frac{1}{p}\sum_{i=1}^{p-m}\left(\frac{\sigma_{m+i}}{-z^{\prime}(1+\sigma_{m+i}m_{2n}(z^{\prime}))}-\frac{\sigma_{m+i}}{-z^{\prime}(1+\sigma_{m+i}m_{2n}(z_j))}\right)\\
    &+\frac{1}{p}\sum_i^{p-m}\left(\frac{\sigma_{m+i}}{-z^{\prime}(1+\sigma_{m+i}m_{2n}(z_j))}-\frac{\sigma_{m+i}}{-z_j(1+\sigma_{m+i}m_{2n}(z_j))}\right)\\
    &:=\mathcal{M}_{11}+\mathcal{M}_{12}.
\end{align*}
For $\mathcal{M}_{11}$, according to Lemma \ref{lem_prf_fakesignals_asymptoticproperties} and \eqref{eq_systemequations}, we readily obtain that
\begin{align*}
    \mathcal{M}_{11}&=\frac{1}{p}\sum_{i=1}^{p-m}\frac{\sigma^2_{m+i}}{-z^{\prime}(1+\sigma_{m+i}m_{2n}(z^{\prime}))(1+\sigma_{m+i}m_{2n}(z_j))}\big(m_{2n}(z_j)-m_{2n}(z^{\prime})\big)\\
    &=\mathrm{O}(|z'|^{-1})\times\frac{1}{p}\sum_i\Big(\frac{\xi^2_i}{-z_j(1+\xi^2_im_{1n}(z_j))}-\frac{\xi^2_i}{-z_j(1+\xi^2_im_{1n}(z^{\prime}))}\Big)\\
    &+\mathrm{O}(|z'|^{-1})\times\frac{1}{p}\sum_i\left(\frac{\xi^2_i}{-z_j(1+\xi^2_im_{1n}(z^{\prime}))}-\frac{\xi^2_i}{-z^{\prime}(1+\xi^2_im_{1n}(z^{\prime}))}\right)\\
    &=\mathrm{O}(|z'|^{-1}) \times \left( \mathtt{M}_{11,1}+ \mathtt{M}_{11,2} \right). 
\end{align*}

For $\mathtt{M}_{11,1},$ by a discussion similar to (\ref{eq_g11}) and (\ref{eq_expansiontwo}), we find that 
\begin{align*}
   & \mathtt{M}_{11,1}=\frac{1}{p}\sum_{i=1}^n\frac{\xi_i^4}{-z_j (1+\xi^2_im_{1n}(z^{\prime}))(1+\xi^2_im_{1n}(z_j))}(m_{1n}(z')-m_{1n}(z_j))\\
    \prec&\left(\frac{\xi_1^4}{-pz_j \mathsf{q}^2 |m_{1n}(z_j)||m_{1n}(z')|}+\frac{1}{p}\sum_{i\ge2}\frac{\xi_i^4}{|z^{\prime}(1+\xi^2_1m_{1n}(z^{\prime}))(1+\xi^2_im_{1n}(z_j))|}\right)\\
    &\times|m_{1n}(z_j)-m_{1n}(z^{\prime})|\\
    \prec& \mathrm{o}(1) \times |m_{1n}(z_j)-m_{1n}(z^{\prime})|.
\end{align*}
Similarly, for $\mathtt{M}_{11,2},$ we have that 
\[
\begin{split}
    \mathtt{M}_{11,2}=\frac{1}{p}\sum_{i}\frac{\xi_i^2 (z_j-z^{\prime})}{z^{\prime}z_j(1+\xi_i^2 m_{1n}(z'))} \prec \frac{z^{\prime}-z_j}{z^{\prime}z_j} \prec n^{-3-2/\alpha}.
\end{split}
\]
Analogously, we can prove that $\mathcal{M}_{12} \prec n^{-3-4/\alpha}$. 
Therefore, under $\Omega$, we see that 
\begin{equation*}
    |m_{1n}(z^{\prime})-m_{1n}(z_j)|\prec  n^{-3-2/\alpha}.
\end{equation*}

\quad By similar procedures and arguments, we can also prove that 
\begin{equation*}
    |m_{2n}(z^{\prime})-m_{2n}(z_j)|\prec n^{-3-2/\alpha}, \  |m_{n}(z^{\prime})-m_{n}(z_j)|\prec n^{-3-2/\alpha},
\end{equation*}
and
\begin{gather*}
    \left\|(z^{\prime})^{-1}(I+m_{1n}(z^{\prime})D^2)^{-1}-(z_j)^{-1}(I+m_{1n}(z_j)D^2)^{-1}\right\|\prec n^{-3-2/\alpha}.
\end{gather*}

Therefore, combining all the above bounds with triangle inequality, we see that the results of part 1 of Theorem \ref{thm_prf_fakesignals_locallaw} hold for $z'.$ Using an induction procedure and a standard lattice argument (for example, see \cite{Alex2014,Ding&Yang2018}), we find that the results hold for all $z \in \widetilde{\mathbf{D}}_{\mu}$ and conclude the proof of Proposition \ref{eq_propooursidebulk}. 
\end{proof}

\subsubsection{Completion of the proof of Theorem \ref{thm_prf_fakesignals_locallaw}} \label{sec_proofa2}

Once Proposition \ref{eq_propooursidebulk} is proved, we can roughly locate the edge eigenvalues of $S_2$ as in Lemma \ref{lem: upper bound for eigenvalues} so that we can expand the spectral domain from $\widetilde{\mathbf{D}}_{\mu}$ to $\mathbf{D}_{\mu}$ for $S_2^{(1)}$ and conclude the proof of Theorem \ref{thm_prf_fakesignals_locallaw}.

\quad  Recall  the definitions of $\mu_2$ and $\lambda_1^{(1)}$. By Lemma \ref{lem: upper bound for eigenvalues} and an analogous argument, as well as  
Weyl's inequality, we find that conditional on the event $\Omega,$ with high probability, 
\begin{equation}\label{eq_keyused}
    \mu_1>\lambda_1>\mu_2>\lambda_1^{(1)}.
\end{equation}
By (\ref{eq_mu1part}) and a similar argument, we see that $\mu_k \asymp \xi_k^2, k=1,2.$ Under $\Omega$, we have that $\mu_1-\mu_2\ge C_1 n^{1/\alpha}\log^{-1}n$ for some constant $C_1>0$. This implies for some constant $C>0$ and $z \in \mathbf{D}_{\mu},$ 
\begin{gather}\label{eq: eigen gap}
    |\lambda_1^{(1)}-z|\ge Cn^{1/\alpha}\log^{-1}n,
\end{gather}

Now we proceed to the proof of Theorem \ref{thm_prf_fakesignals_locallaw}. Recall (\ref{eq_defnminor}) and (\ref{eq_defnminorG}). 
\begin{proof}[\bf Proof of Theorem \ref{thm_prf_fakesignals_locallaw}]  
Observe  by \eqref{eq: eigen gap} that it holds uniformly over $z\in\mathbf{D}_{\mu}$ and $\mathcal{T}\subset\{2,\dots,n\}$, for some constant $C_1>0$
\begin{gather}\label{eq: est of G^(1)}
    \|G^{(1\mathcal{T})}\|\leq C_1 n^{-1/\alpha} \log n.
\end{gather}
By the definition of $m_2^{(1)}$ and a decomposition similar to (\ref{eq: decomp m_2}),  we have that 
\begin{gather*}
        m_2^{(1)}=\frac{1}{p}\sum_{i=2}^n\frac{\xi^2_i}{-z-z\mathbf{y}_{i,2}^{\prime}G^{(1i)}\mathbf{y}_{i,2}}=\frac{1}{p}\sum_{i=2}^n\frac{\xi^2_i}{-z(1+\xi^2_ip^{-1}\operatorname{tr}G^{(1i)}\Sigma_2+Z_i^{(1)})},\\
        Z_i^{(1)}=\mathbf{y}_{i,2}^{\prime}G^{(1i)}\mathbf{y}_{i,2}-\xi^2_ip^{-1}\operatorname{tr}G^{(1i)}\Sigma_2.
\end{gather*}
By arguments similar to (\ref{eq: bound for Z}), (\ref{eq_m1approximate}) and using \eqref{eq: est of G^(1)}, we obtain that 
\begin{gather*}
    Z_i^{(1)}\prec\frac{\xi^2_i}{p}\|G^{(1i)}\Sigma_2\|_F\le\frac{\xi^2_i}{n}\|G^{(1i)}\|\|\Sigma_2\|_F\prec\frac{\xi^2_i}{p^{1/2}}n^{-1/\alpha},\\
    \frac{1}{p}\operatorname{tr}(G^{(1i)}\Sigma_2)-m_1^{(1)}(z)=\frac{1}{p}\mathbf{y}_{i,2}^{\prime}G^{(1)}\Sigma_2 G^{(1i)}\mathbf{y}_{i,2}\prec \frac{\xi^2_i}{p}n^{-2/\alpha}.
\end{gather*}
In addition, using (\ref{eq: est of G^(1)}) and a discussion similar to (\ref{eq_eeeone})--(\ref{eq_modeltemp}), we readily see that 
\begin{gather*}
    m^{(1)}_2=\frac{1}{p}\sum_{i=2}^n\frac{\xi^2_i}{-z(1+\xi^2_im^{(1)}_1)}+\mathrm{O}_{\prec}\left(n^{-1/2-1/\alpha}\right).
\end{gather*}
Using the decomposition 
\begin{gather*}
    S_2^{(1)}-zI=\sum_{i=2}^n\mathbf{y}_{i,2}\mathbf{y}_{i,2}^{\prime}+zm^{(1)}_2(z)\Sigma_2-z(I+m_2^{(1)}(z)\Sigma_2),
\end{gather*}
by arguments similar to (\ref{eq: decomp of mathcal_G})--(\ref{eq: m_1 by m_2}) with $\|\mathcal{G}^{(1\mathcal{T})}\|=\|G^{(1 \mathcal{T})}\|\prec n^{-1/\alpha}$, we conclude that 
\begin{gather*}
\begin{split}
    m_1^{(1)}&=-z^{-1}\frac{1}{p}\operatorname{tr}((I+m_2^{(1)}(z)\Sigma_2)^{-1}\Sigma_2)+\mathrm{O}_{\prec}\left(n^{-1/2-2/\alpha}\right)\\
    &=-\frac{1}{p}\sum_{i=1}^{p-m}\frac{\sigma_{m+i}}{z(1+m_2^{(1)}\sigma_{m+i})}+\mathrm{O}_{\prec}\left(n^{-1/2-2/\alpha}\right).
\end{split}
\end{gather*}
Combining the definitions in \eqref{eq_systemequations}, we see that
\begin{gather*}
    \begin{split}
      &  m_1^{(1)}(z)-m_{1n}^{(1)}(z)\\
    =&\frac{1}{p}\sum_{i=1}^{p-m}\frac{\sigma_{m+i}^2(m_2^{(1)}(z)-m_{2n}^{(1)}(z))}{z(1+\sigma_{m+i}m_{2n}^{(1)}(z))(1+\sigma_{m+i}m_2^{(1)}(z))}+\mathrm{O}_{\prec}(n^{-1/2-2/\alpha})\\
        =&\left(\frac{1}{p}\sum_{i=1}^{p-m}\frac{\sigma_{m+i}^2}{z(1+\sigma_{m+i}m_{2n}^{(1)}(z))(1+\sigma_{m+i}m_2^{(1)}(z))}\right)\\
        &\times\left(\frac{1}{p}\sum_{i=2}^n\frac{\xi^4_i(m_1^{(1)}(z)-m_{1n}^{(1)}(z))}{z(1+\xi^2_im_1^{(1)}(z))(1+\xi^2_im_{1n}^{(1)}(z))} \right)+\mathrm{O}_{\prec}(n^{-1/2-2/\alpha})\\
        =&\mathrm{o}(1)(m_1^{(1)}(z)-m_{1n}^{(1)}(z))+\mathrm{O}_{\prec}(n^{-1/2-2/\alpha}),
    \end{split}
\end{gather*} 
where in the second step we used a discussion similar to (\ref{eq_eeetwo}) under $\Omega$. This completes our proof. 
\end{proof}

\subsection{Proof of Lemma \ref{lem_radius_separation}}
\begin{proof}
As in the proof of Lemma \ref{lem_goodconfiguration}, 
\[
(\xi^2_{(1)},\xi^2_{(2)})/a_n\overset{d}{\to}(\Gamma_1^{-1/\alpha},\Gamma_2^{-1/\alpha}),
\]
 where $\Gamma_k=E_1+\dots+E_k$ with $E_1,E_2$ i.i.d.\ standard exponential \citep[Chapters 3--4]{resnick1987extreme}. By the continuous mapping theorem, $\xi^2_{(2)}/\xi^2_{(1)}\overset{d}{\to}(\Gamma_1/\Gamma_2)^{1/\alpha}$. Moreover, $\Gamma_1/\Gamma_2=E_1/(E_1+E_2)$ is uniformly distributed on $(0,1)$. Hence, for $\delta>0$,
\begin{equation*}
\mathbb P(r_n\le1+\delta)=\mathbb P\big(\xi^2_{(2)}/\xi^2_{(1)}\ge(1+\delta)^{-1}\big)\to\mathbb P\big(U\ge(1+\delta)^{-\alpha}\big)=1-(1+\delta)^{-\alpha},
\end{equation*}
where $\mathcal{U}$ is uniform on $(0,1)$, and $(1+\delta)^{-1}$ is a continuity point of the limit. Letting $\delta\downarrow0$ gives \eqref{eq_radius_nondegenerate}. For the last claim, $\xi^2_{(2)}/a_n$ converges in distribution to the almost surely positive variable $\Gamma_2^{-1/\alpha}$, so $\xi^2_{(2)}\ge a_n/\log n$ with probability tending to one. The error in Lemma \ref{lem_fakesignal_compare_nonspikes} is therefore negligible relative to $\phi_n\xi^2_{(2)}$. Lemma \ref{lem_fakesignal_compare_nonspikes} and Theorem \ref{thm_fakesignal_limit} with $k=1,2$ then give $\lambda_{m+k}=\phi_n\xi^2_{(k)}(\bar\sigma+\orderP(1))$, and the claim follows.
\end{proof}

\section{Proof of the results in Section \ref{sec_testing_procedures}}\label{sec_proof_testing procedures}
Throughout this section we use the notation of the main text, $\xi_t^2$ denotes the original radius with $\E\xi_t^2=1$, $\mathbf y_t=\sqrt p\,\xi_t\Sigma^{1/2}\mathbf u_t$, $S=n^{-1}YY'$ and $\phi_n=p/n$. We write $a_0:=1-\bar h\in(0,1)$, and $u_{it}$ for the $i$-th entry of $\mathbf u_t$.
\subsection{A preliminary lemma}
\begin{lemma}\label{lem_prf_test_prelim}
Suppose Assumptions \ref{ass_xi} and \ref{ass_sigma} hold. There are constants $C_1,C_2>0$ such that, with probability tending to one,
\begin{itemize}
\item[(i)] $\lambda_m\ge\sigma_m/2$ and $\lambda_1/\lambda_m\le2\sigma_1/\sigma_m\le C_1$;
\item[(ii)] $\xi^2_{(1)}\le C_1\mathsf T$ and $\lambda_{m+1}\le C_1\mathsf T$;
\item[(iii)] $\max_{1\le t\le n}\ell_t\le C_2\big(\mathsf T\log n/n+\mathsf T/\sigma_m\big)$, where $\ell_t$ is defined in \eqref{eq_def_leverage}.
\end{itemize}
\end{lemma}
\begin{proof}
(i) follows from Theorem \ref{lem_realsignal_preratio}, together with $\sigma_1/\sigma_m\le c_3^{m-1}(1+\mathrm o(1))$ as in Remark \ref{rem_sigma_diagonal}. 

(ii) The first claim is Lemma \ref{lem_goodconfiguration}. For the second, write $\mathcal S=\mathcal S_1+\mathcal S_2$ with $\mathcal S_\ell=DU'\Sigma_\ell UD$. Since $\mathcal S_1\succeq0$ has rank at most $m$, Weyl's inequality gives 
\[
\lambda_{m+1}=\lambda_{m+1}(\mathcal S)\le\lambda_1(\mathcal S_2)\le\|D\|^2\|\Sigma_2\|\|U'U\|\le\phi_n\xi^2_{(1)}c_1\|U'U\|.
\]
Moreover, $\|U'U\|\le C$ with probability tending to one \citep{Wen2021}. 

(iii) Write $\Sigma^{1/2}=\Sigma_1^{1/2}+\Sigma_2^{1/2}$, where $\Sigma_1^{1/2}=\diag(\sigma_1^{1/2},\dots,\sigma_m^{1/2},0,\dots,0)$ and $\Sigma_2^{1/2}=\diag(0,\dots,0,\sigma_{m+1}^{1/2},\dots,\sigma_p^{1/2})$. Let $P:=V\Lambda^{-1}V'\succeq0$, so that $\|P\|=\lambda_m^{-1}$. Using the inequality
\[
(\mathbf a+\mathbf b)'P(\mathbf a+\mathbf b)\le2\mathbf a'P\mathbf a+2\mathbf b'P\mathbf b,
\]
we have
\begin{equation*}
\ell_t=\frac pn\,\xi_t^2\,\mathbf u_t'\Sigma^{1/2}P\Sigma^{1/2}\mathbf u_t\le\frac{2p}{n}\,\xi_t^2\Big(\frac{\sigma_1}{\lambda_m}\sum_{i=1}^mu_{it}^2+\frac{c_1}{\lambda_m}\Big).
\end{equation*}
For $\mathbf u$ uniform on $\mathbb S^{p-1}$, we have
\[
\mathbb P(|u_i|\ge s)\le2e^{-ps^2/2},\quad s\in(0,1)
\]
according to \citep[Lemma 2.2]{ball1997elementary}. Hence
\[
\mathbb P(\max_{t\le n,i\le m}pu_{it}^2\ge4\log n)\le2mn\cdot n^{-2}\to0.
\]
 Combining this with (i) and (ii) gives 
 \[
 \max_t\ell_t\le\frac{2}{n}C_1\mathsf T\big(2C_1\cdot4m\log n+2pc_1/\sigma_m\big),
 \]
 which proves (iii).
\end{proof}

\subsection{Proof of Lemma \ref{lem_mc_error}}
\begin{proof}
Fix $Y$, $h$ and $k$, and write $X_j:=X^{(j)}_k=\lambda^{(j)}_k/\lambda_k$, $1\le j\le K$. Conditionally on $Y$, $X_j$ are i.i.d. and take values in $[1-h,1+h]$ by \eqref{eq_sandwich}. Write $\mu:=\E^*X_1\in[1-h,1+h]$ and $\sigma^2:=\Var^*(X_1)\le h^2$. Then, $\tau_k(h)=\sigma^2/\mu^2$. Moreover, $\mathbb T_k=s^2/\bar X^2$, where $\bar X=K^{-1}\sum_jX_j$ and $s^2=K^{-1}\sum_j(X_j-\bar X)^2$. Put $\varepsilon_j:=X_j-\mu\in[-2h,2h]$. By Hoeffding's inequality \citep{hoeffding1963probability}, each of the following holds with $\mathbb P^*$-probability at least $1-\eta/2$:
\begin{equation*}
\Big|\frac1K\sum_j\varepsilon_j\Big|\le2h\sqrt{\frac{2\log(4/\eta)}{K}},\qquad\Big|\frac1K\sum_j(\varepsilon_j^2-\sigma^2)\Big|\le4h^2\sqrt{\frac{\log(4/\eta)}{2K}}.
\end{equation*}
On the intersection of the two events, using $s^2=K^{-1}\sum_j\varepsilon_j^2-\bar\varepsilon^2$ and $K\ge \log(4/\eta)$,
\begin{equation*}
|s^2-\sigma^2|\le4h^2\sqrt{\log(4/\eta)/(2K)}+8h^2\log(4/\eta)/K\le12h^2\sqrt{\log(4/\eta)/K}.
\end{equation*}
Moreover, $\bar X,\mu\ge a_0$, so $|\bar X^{-2}-\mu^{-2}|\le|\mu^2-\bar X^2|/a_0^4\le4|\bar\varepsilon|/a_0^4$. Therefore
\begin{equation*}
|\mathbb T_k-\tau_k(h)|\le\frac{|s^2-\sigma^2|}{\bar X^2}+\sigma^2\big|\bar X^{-2}-\mu^{-2}\big|\le h^2\sqrt{\frac {\log(4/\eta)}{K}}\Big(\frac{12}{a_0^2}+\frac{8\sqrt2}{a_0^4}\Big).
\end{equation*}
Dividing by $v(h)=h^2/3$ gives the claim with $C_0:=3(12a_0^{-2}+12a_0^{-4})$.
\end{proof}

\subsection{Proof of Proposition \ref{prop_real_magnified}}
\begin{proof}
Fix $1\le k\le m$ and $h\in(0,\bar h]$, and write $W=W_1$ and $\widetilde S=\widetilde S^{(1)}$. Let $\mathcal F_n$ be the event on which the conclusions of Lemma \ref{lem_prf_test_prelim} hold, together with $\lambda_m\ge4\lambda_{m+1}/a_0$ and $8m\max_t\ell_t\le1$. By Lemma \ref{lem_prf_test_prelim} and $\mathsf T/\sigma_m\to0$, we have $\mathbb P(\mathcal F_n)\to1$. Therefore, we can restrict ourselves on the event $\mathcal F_n$. All statements below are deterministic on $\mathcal F_n$.

Let $V_\perp$ complete $V$ to an orthonormal eigenbasis of $S$, and decompose $\widetilde S$ in the basis $(V,V_\perp)$:
\begin{equation*}
\widetilde{S}_{11}:=V'\widetilde SV,\quad \widetilde{S}_{12}:=V'\widetilde SV_\perp,\quad \widetilde{S}_{22}:=V_\perp'\widetilde SV_\perp .
\end{equation*}

\textbf{Step 1: the leading block.} Let $\mathbf q_t:=n^{-1/2}\Lambda^{-1/2}V'\mathbf y_t\in\R^m$. Then 
\[
\sum_t\mathbf q_t\mathbf q_t'=\Lambda^{-1/2}V'SV\Lambda^{-1/2}=I_m,\quad \|\mathbf q_t\|^2=\ell_t.
\]
Moreover,
\begin{equation*}
\widetilde{S}_{11}=\Lambda^{1/2}(I_m+E_W)\Lambda^{1/2},\quad E_W:=\sum_{t=1}^n(w_t-1)\mathbf q_t\mathbf q_t',\quad -hI_m\preceq E_W\preceq hI_m .
\end{equation*}
The matrix $\widetilde{S}_{11}$ has the same eigenvalues as $(I+E_W)^{1/2}\Lambda(I+E_W)^{1/2}$. By Ostrowski's theorem \citep[Theorem 4.5.9]{horn2012matrix}, $\lambda_k(\widetilde{S}_{11})=\vartheta_k\lambda_k$ for some $\vartheta_k\in[1-\|E_W\|,1+\|E_W\|]$.

\textbf{Step 2: the off-diagonal block.} By Cauchy's interlacing theorem, $\lambda_k(\widetilde S)\ge\lambda_k(\widetilde{S}_{11})$. For an upper bound, note that $\lambda_m(\widetilde{S}_{11})\ge(1-h)\lambda_m$ because $\widetilde{S}_{11}\succeq(1-h)\Lambda$, and $\|\widetilde{S}_{22}\|\le(1+h)\lambda_{m+1}$ because $\widetilde{S}_{22}\preceq(1+h)V_\perp'SV_\perp$. Since $V'SV_\perp=0$, we have
\[
\widetilde{S}_{12}=n^{-1}V'Y(W-I)Y'V_\perp,\quad \|\widetilde{S}_{12}\|\le h\sqrt{\lambda_1\lambda_{m+1}}.
\]
 On $\mathcal E_n$, 
 \[
    \lambda_m(\widetilde{S}_{11})-\|\widetilde{S}_{22}\|\ge a_0\lambda_m-2\lambda_{m+1}\ge a_0\lambda_m/2>0.
 \]
 Set $\beta:=\|\widetilde{S}_{12}\|^2/(\lambda_m(\widetilde{S}_{11})-\|\widetilde{S}_{22}\|)$ and $\mu:=\lambda_k(\widetilde{S}_{11})+\beta>\|\widetilde{S}_{22}\|$. By the Haynsworth inertia additivity formula \citep{haynsworth1968determination} applied to $\mu I-\widetilde S$, the number of eigenvalues of $\widetilde S$ exceeding $\mu$ equals the number of eigenvalues of $\widetilde{S}_{11}+\widetilde{S}_{12}(\mu I-\widetilde{S}_{22})^{-1}\widetilde{S}_{12}'$ exceeding $\mu$. Since 
 \[
    \widetilde{S}_{12}(\mu I-\widetilde{S}_{22})^{-1}\widetilde{S}_{12}'\preceq\|\widetilde{S}_{12}\|^2(\mu-\|\widetilde{S}_{22}\|)^{-1}I\preceq\beta I,
 \]
  this number is at most the number of eigenvalues of $\widetilde{S}_{11}$ exceeding $\mu-\beta=\lambda_k(\widetilde{S}_{11})$, which is at most $k-1$. Hence $\lambda_k(\widetilde S)\le\lambda_k(\widetilde{S}_{11})+\beta$ (see also \cite{li2005note}). Dividing by $\lambda_k\ge\lambda_m$,
\begin{equation}\label{eq_prf_real_sandwich}
\vartheta_k\le X_k:=\frac{\lambda_k(\widetilde S)}{\lambda_k}\le\vartheta_k+h^2\delta_n,\quad \delta_n:=\frac{2\lambda_1\lambda_{m+1}}{a_0\lambda_m^2}.
\end{equation}

\textbf{Step 3: conditional variance.} By \eqref{eq_prf_real_sandwich}, $|X_k-1|\le\|E_W\|+h^2\delta_n$. Since the $w_t-1$ are independent, centered and have variance $v(h)$,
\begin{equation*}
\E^*\|E_W\|^2\le\E^*\|E_W\|_F^2=v(h)\sum_t\|\mathbf q_t\|^4\le v(h)\max_t\ell_t\sum_t\ell_t=m\,v(h)\max_t\ell_t .
\end{equation*}
Hence
\[
\Var^*(X_k)\le\E^*(X_k-1)^2\le2m\,v(h)\max_t\ell_t+2h^4\delta_n^2.
\]
 Moreover, on $\mathcal{F}_n$ we have
 \[
 \E^*X_k\ge1-\E^*\|E_W\|\ge1-(m\,v(h)\max_t\ell_t)^{1/2}\ge1/2.
 \]
Consequently,
\begin{equation*}
\frac{\tau_k(h)}{v(h)}\le\frac{4\big(2m\,v(h)\max_t\ell_t+2h^4\delta_n^2\big)}{v(h)}=8m\max_t\ell_t+24h^2\delta_n^2\le8m\max_t\ell_t+24\bar h^2\delta_n^2 .
\end{equation*}
The right-hand side does not depend on $h$. By Lemma \ref{lem_prf_test_prelim}(i), $\delta_n\le4C_1\lambda_{m+1}/(a_0\lambda_m)$, which gives \eqref{eq_real_bound}. Since $\lambda_{m+1}/\lambda_m=\OrderP(\mathsf T/\sigma_m)$ and $(\mathsf T/\sigma_m)^2\le\mathsf T/\sigma_m$, the rate follows from Lemma \ref{lem_prf_test_prelim}.
\end{proof}

\subsection{Proof of Proposition \ref{prop_fake_magnified}}
\begin{proof}
For $h\in[0,\bar h]$, write
\[
w_t(h)=1+h\gamma_t,\quad \widetilde S_h=n^{-1}Y\diag(w_t(h))Y',\quad X(h):=\lambda_{m+1}(\widetilde S_h)/\lambda_{m+1}
\]
 and $M(h):=\max_t\xi_t^2w_t(h)/\xi^2_{(1)}$. By \eqref{eq_sandwich}, we have $X(h)\in[a_0,2]$. And
 \[
    \xi^2_{(1)}w_{t^*}(h)\le\max_t\xi_t^2w_t(h)\le(1+h)\xi^2_{(1)},
 \]
 we have $M(h)\in[a_0,2]$.

\textbf{Step 1: fixed $h$.} By (F2), Lemma \ref{lem_fakesignal_compare_nonspikes} and Theorem \ref{thm_fakesignal_limit} apply to $\widetilde S_h$, whose radii $\xi_t^2w_t(h)$ satisfy \eqref{ass_xi_poly} with the same $\alpha$. On $\Omega$, we have
\[
\max_t\xi_t^2w_t(h)\ge a_0\xi^2_{(1)}\ge a_0C^{-1}n^{1/\alpha}\log^{-1}n.
\]
It gives
\[
\lambda_{m+1}(\widetilde S_h)/(\phi_n\max_t\xi_t^2w_t(h))\to\bar\sigma
\]
 in probability. The same argument applied to $S$ gives 
 \[
 \lambda_{m+1}/(\phi_n\xi^2_{(1)})\to\bar\sigma.
 \]
 Since $\bar\sigma\ge c_1^{-1}$, $X(h)/M(h)\to1$ in probability and both are bounded, $X(h)-M(h)\to0$ in probability.

\textbf{Step 2: uniformity in $h$.} For $h,h'\in[0,\bar h]$, 
\[
|w_t(h')-w_t(h)|\le|h-h'|,\quad w_t(h)\ge a_0.
\]
Hence $w_t(h')/w_t(h)\in[1-|h-h'|/a_0,1+|h-h'|/a_0]$, and so 
\[
(1-|h-h'|/a_0)\widetilde S_h\preceq\widetilde S_{h'}\preceq(1+|h-h'|/a_0)\widetilde S_h.
\]
It follows that 
\[
|X(h')-X(h)|\le2|h-h'|/a_0,
\]
and in the same way 
\[
|M(h')-M(h)|\le2|h-h'|/a_0.
\]
 Let $\mathcal{N}_\varepsilon:=\{0,\varepsilon,2\varepsilon,\dots\}\cap[0,\bar h]$. Then
\begin{equation*}
\sup_{h\in[0,\bar h]}|X(h)-M(h)|\le\max_{g\in \mathcal{N}_\varepsilon}|X(g)-M(g)|+4\varepsilon/a_0 .
\end{equation*}
By Step 1, the maximum over the finite set $\mathcal{N}_\varepsilon$ tends to zero in probability. Since $\varepsilon>0$ is arbitrary,
\[
\sup_h|X(h)-M(h)|\to0\quad \text{in probability.}
\]
 As the supremum is bounded by $2$,
 \[
 \E\big[\sup_h|X(h)-M(h)|^2\big]\to0.
 \]
  Thus,
  \[
  \E^*\big[\sup_h|X(h)-M(h)|^2\big]\to0
  \]
  in $L^1(\mathbb P)$, and hence in probability. Together with $\Delta_n^2\le\E^*[\sup_h|X(h)-M(h)|^2]$, this proves \eqref{eq_fake_approx}.

\textbf{Step 3: order-preserving $h$.} Let $h=h(Y)$. Conditionally on $Y$, the law of $(X(h),M(h))$ under the coupling coincides with the law obtained from an independent draw $W_1$ with i.i.d.\ $U[1-h,1+h]$ entries. Suppose \eqref{eq_order_preserving} holds. For $t\ne t^*$,
\[
\xi_t^2w_t(h)\le(1+h)\xi^2_{(2)}\le(1-h)\xi^2_{(1)}\le\xi^2_{(1)}w_{t^*}(h).
\]
 Hence $M(h)=w_{t^*}(h)$, so that $\E^*M(h)=1$ and $\Var^*M(h)=v(h)$. Let 
 \[
 \mu_X:=\E^*X(h)\ge a_0,\quad \sigma_X^2:=\Var^*X(h).
 \]
Then $|\mu_X-1|\le\Delta_n$ and, by Minkowski's inequality, $|\sigma_X-v(h)^{1/2}|\le2\Delta_n$. Therefore, using $v(h)\le1/3$,
\begin{equation*}
|\tau_{m+1}(h)-v(h)|\le\frac{|\sigma_X^2-v(h)|}{\mu_X^2}+v(h)\frac{|1-\mu_X^2|}{\mu_X^2}\le\frac{4\Delta_n(1+\Delta_n)+\Delta_n(2+\Delta_n)}{a_0^2},
\end{equation*}
which is of the form $C(\Delta_n+\Delta_n^2)$.
\end{proof}

\subsection{Proof of Theorem \ref{thm_firstround}}
\begin{proof}
Recall that $h_i$ is a function of $Y$.

\textbf{Genuine factors.} Let $1\le i\le m$ and $\eta_K:=4\exp\{-K\kappa^2/(4C_0^2)\}$, so that $C_0\sqrt{\log(4/\eta_K)/K}=\kappa/2$. By Lemma \ref{lem_mc_error}, applied conditionally on $Y$ with $h=h_i$,
\begin{equation*}
\mathbb P(\mathsf R_i>\kappa)\le\mathbb P\Big(\sup_{0<h\le\bar h}\frac{\tau_i(h)}{v(h)}>\frac\kappa2\Big)+\eta_K .
\end{equation*}
The first term tends to zero by Proposition \ref{prop_real_magnified}, and $\eta_K\to0$ because $K\to\infty$.

\textbf{The leading spurious index.} Fix $\varepsilon\in(0,1)$. By \eqref{eq_radius_nondegenerate}, there is $\delta\in(0,1)$ with 
\[
\limsup_n\mathbb P(r_n\le1+\delta)\le\varepsilon.
\]
 Let $\hat\rho:=\lambda_{m+1}/\lambda_{m+2}$, so that $\hat\rho/r_n\to1$ in probability by Lemma \ref{lem_radius_separation}. Let 
 \[
 A_n:=\{r_n>1+\delta,\ |\hat\rho/r_n-1|\le\delta/4\},
 \]
 then 
 \[
 \limsup_n\mathbb P(A_n^c)\le\varepsilon.
 \]
  On $A_n$, two facts hold. First, 
  \[
  \hat\rho\le r_n(1+\delta/4)\le r_n^2\Rightarrow s_{m+1}\le\hat\rho^{1/2}\le r_n.
  \]
  Then, $h_{m+1}$ satisfies \eqref{eq_order_preserving}. Second, 
  \[
  \hat\rho\ge(1+\delta)(1-\delta/4)\ge1+\delta/2
  \]
   so that 
   \[
   h_{m+1}\ge h_\delta:=(s_\delta-1)/(s_\delta+1)>0,\quad s_\delta:=\min\{(1+\bar h)/(1-\bar h),(1+\delta/2)^{1/2}\}. 
   \]
    By Proposition \ref{prop_fake_magnified}, on $A_n$,
\begin{equation*}
\frac{\tau_{m+1}(h_{m+1})}{v(h_{m+1})}\ge1-\frac{3C(\Delta_n+\Delta_n^2)}{h_\delta^2}.
\end{equation*}
Combining this with Lemma \ref{lem_mc_error} (with $(1-\kappa)/2$ in place of $\kappa/2$),
\begin{equation*}
\mathbb P(\mathsf R_{m+1}\le\kappa)\le\mathbb P(A_n^c)+\mathbb P\Big(\frac{3C(\Delta_n+\Delta_n^2)}{h_\delta^2}>\frac{1-\kappa}{2}\Big)+4\exp\Big\{-\frac{K(1-\kappa)^2}{4C_0^2}\Big\}.
\end{equation*}
The last two terms tend to zero by \eqref{eq_fake_approx} and $K\to\infty$. Hence 
\[
\limsup_n\mathbb P(\mathsf R_{m+1}\le\kappa)\le\varepsilon.
\]
 As $\varepsilon$ is arbitrary, this proves \eqref{eq_firstround_consistency}. Finally, we observe
 \[
 \mathbb P(r^*\ne m+1)\le\sum_{i\le m}\mathbb P(\mathsf R_i>\kappa)+\mathbb P(\mathsf R_{m+1}\le\kappa)\to0.
 \]
\end{proof}

\subsection{Proof of Theorem \ref{lem_secondround_consistency}}
\begin{proof}
By Algorithm \ref{alg_secondround_estrealsignal}, $\widehat r=\min\{k,r^*-1\}$. On the event $\{r^*=m+1\}\cap\{k\ge m\}$ we have $\widehat r=\min\{k,m\}=m$. Hence 
\[
\mathbb P(\widehat r\ne m)\le\mathbb P(r^*\ne m+1)+\mathbb P(k<m)\to0,
\]
by Theorem \ref{thm_firstround} and the assumption on $k$.
\end{proof}

%




\begin{thebibliography}{51}
\providecommand{\natexlab}[1]{#1}
\providecommand{\url}[1]{\texttt{#1}}
\expandafter\ifx\csname urlstyle\endcsname\relax
  \providecommand{\doi}[1]{doi: #1}\else
  \providecommand{\doi}{doi: \begingroup \urlstyle{rm}\Url}\fi

\bibitem[Ahn and Horenstein(2013)]{ahn2013eigenvalue}
Seung~C Ahn and Alex~R Horenstein.
\newblock Eigenvalue ratio test for the number of factors.
\newblock \emph{Econometrica}, 81\penalty0 (3):\penalty0 1203--1227, 2013.

\bibitem[Alessi et~al.(2010)Alessi, Barigozzi, and Capasso]{alessi2010improved}
Lucia Alessi, Matteo Barigozzi, and Marco Capasso.
\newblock Improved penalization for determining the number of factors in
  approximate factor models.
\newblock \emph{Statistics \& Probability Letters}, 80\penalty0
  (23-24):\penalty0 1806--1813, 2010.

\bibitem[Alex et~al.(2014)Alex, Erd{\H{o}}s, Knowles, Yau, and Yin]{Alex2014}
Bloemendal Alex, L{\'a}szl{\'o} Erd{\H{o}}s, Antti Knowles, Horng-Tzer Yau, and
  Jun Yin.
\newblock Isotropic local laws for sample covariance and generalized wigner
  matrices.
\newblock \emph{Electronic Journal of Probability}, 19:\penalty0 1--53, 2014.

\bibitem[Bai(2003)]{bai2003inferential}
Jushan Bai.
\newblock Inferential theory for factor models of large dimensions.
\newblock \emph{Econometrica}, 71\penalty0 (1):\penalty0 135--171, 2003.

\bibitem[Bai and Ng(2002)]{bai2002determining}
Jushan Bai and Serena Ng.
\newblock Determining the number of factors in approximate factor models.
\newblock \emph{Econometrica}, 70\penalty0 (1):\penalty0 191--221, 2002.

\bibitem[Bai et~al.(2018)Bai, Choi, and Fujikoshi]{bai2018consistency}
Zhidong Bai, Kwok~Pui Choi, and Yasunori Fujikoshi.
\newblock Consistency of aic and bic in estimating the number of significant
  components in high-dimensional principal component analysis.
\newblock \emph{The Annals of Statistics}, 46\penalty0 (3):\penalty0
  1050--1076, 2018.

\bibitem[Ball(1997)]{ball1997elementary}
Keith Ball.
\newblock An elementary introduction to modern convex geometry.
\newblock In \emph{Flavors of Geometry}, volume~31 of \emph{MSRI Publications},
  pages 1--58. Cambridge University Press, 1997.

\bibitem[Baltagi et~al.(2017)Baltagi, Kao, and Wang]{baltagi2017identification}
Badi~H Baltagi, Chihwa Kao, and Fa~Wang.
\newblock Identification and estimation of a large factor model with structural
  instability.
\newblock \emph{Journal of Econometrics}, 197\penalty0 (1):\penalty0 87--100,
  2017.

\bibitem[Bao et~al.(2025)Bao, Cheong, and Li]{bao2025signal}
Zhigang Bao, Kha~Man Cheong, and Yuji Li.
\newblock Signal detection from spiked noise via asymmetrization.
\newblock \emph{arXiv preprint arXiv:2504.19450}, 2025.

\bibitem[Barigozzi and Cho(2020)]{barigozzi2020consistent}
Matteo Barigozzi and Haeran Cho.
\newblock Consistent estimation of high-dimensional factor models when the
  factor number is over-estimated.
\newblock \emph{Electronic Journal of Statistics}, 14:\penalty0 2892--2921,
  2020.

\bibitem[Billingsley(1995)]{billingsley1995probability}
Patrick Billingsley.
\newblock \emph{Probability and Measure}.
\newblock Wiley, New York, 3rd edition, 1995.

\bibitem[Bingham et~al.(1987)Bingham, Goldie, and Teugels]{bingham1987regular}
N.~H. Bingham, C.~M. Goldie, and J.~L. Teugels.
\newblock \emph{Regular Variation}.
\newblock Cambridge University Press, 1987.

\bibitem[Bloemendal et~al.(2016)Bloemendal, Knowles, Yau, and
  Yin]{bloemendal2016principal}
Alex Bloemendal, Antti Knowles, Horng-Tzer Yau, and Jun Yin.
\newblock On the principal components of sample covariance matrices.
\newblock \emph{Probability theory and related fields}, 164\penalty0
  (1):\penalty0 459--552, 2016.

\bibitem[Breiman(1965)]{breiman1965some}
Leo Breiman.
\newblock On some limit theorems similar to the arc-sin law.
\newblock \emph{Theory of Probability and Its Applications}, 10\penalty0
  (2):\penalty0 323--331, 1965.

\bibitem[Burkholder(1973)]{burkholder1973distribution}
Donald~L. Burkholder.
\newblock Distribution function inequalities for martingales.
\newblock \emph{The Annals of Probability}, 1\penalty0 (1):\penalty0 19--42,
  1973.
\newblock \doi{10.1214/aop/1176997023}.

\bibitem[Cai et~al.(2020)Cai, Han, and Pan]{cai2020limiting}
T~Tony Cai, Xiao Han, and Guangming Pan.
\newblock Limiting laws for divergent spiked eigenvalues and largest nonspiked
  eigenvalue of sample covariance matrices.
\newblock \emph{Annals of Statistics}, 48:\penalty0 1255--1280, 2020.

\bibitem[Carlson(2007)]{Carlson2007Crash}
Mark Carlson.
\newblock A brief history of the 1987 stock market crash with a discussion of
  the {Federal Reserve} response.
\newblock Finance and Economics Discussion Series 2007-13, Board of Governors
  of the Federal Reserve System, 2007.

\bibitem[Chamberlain and Rothschild(1983)]{chamberlain1982arbitrage}
Gary Chamberlain and Michael Rothschild.
\newblock Arbitrage, factor structure, and mean-variance analysis on large
  asset markets.
\newblock \emph{Econometrica}, 51\penalty0 (5):\penalty0 1281--1304, 1983.

\bibitem[Ding and Yang(2018)]{Ding&Yang2018}
Xiucai Ding and Fan Yang.
\newblock A necessary and sufficient condition for edge universality at the
  largest singular values of covariance matrices.
\newblock \emph{The Annals of Applied Probability}, 28\penalty0 (3):\penalty0
  1679--1738, 2018.

\bibitem[Ding and Yang(2021)]{ding2021spiked}
Xiucai Ding and Fan Yang.
\newblock Spiked separable covariance matrices and principal components.
\newblock \emph{The Annals of Statistics}, 49\penalty0 (2):\penalty0
  1113--1138, 2021.

\bibitem[Dobriban and Owen(2018)]{Dobriban}
Edgar Dobriban and Art~B. Owen.
\newblock Deterministic parallel analysis: An improved method for selecting
  factors and principal components.
\newblock \emph{Journal of the Royal Statistical Society Series B: Statistical
  Methodology}, 81\penalty0 (1):\penalty0 163--183, 11 2018.
\newblock ISSN 1369-7412.

\bibitem[El~Karoui(2009)]{Karoui2009}
Noureddine El~Karoui.
\newblock Concentration of measure and spectra of random matrices: Applications
  to correlation matrices, elliptical distributions and beyond.
\newblock \emph{The Annals of Applied Probability}, 19\penalty0 (6):\penalty0
  2362--2405, 2009.

\bibitem[Fama and French(1993)]{fama1993common}
Eugene~F Fama and Kenneth~R French.
\newblock Common risk factors in the returns on stocks and bonds.
\newblock \emph{Journal of financial economics}, 33\penalty0 (1):\penalty0
  3--56, 1993.

\bibitem[Fan et~al.(2018)Fan, Liu, and Wang]{fan2018large}
Jianqing Fan, Han Liu, and Weichen Wang.
\newblock Large covariance estimation through elliptical factor models.
\newblock \emph{Annals of statistics}, 46\penalty0 (4):\penalty0 1383, 2018.

\bibitem[Fan et~al.(2021)Fan, Li, and Liao]{fan2021recent}
Jianqing Fan, Kunpeng Li, and Yuan Liao.
\newblock Recent developments in factor models and applications in econometric
  learning.
\newblock \emph{Annual Review of Financial Economics}, 13:\penalty0 401--430,
  2021.

\bibitem[French(2026)]{FrenchIndustry49}
Kenneth~R. French.
\newblock 49 industry portfolios.
\newblock Kenneth R. French Data Library, 2026.
\newblock Daily value-weighted returns, August 2026 data vintage. Accessed
  October 24, 2026.

\bibitem[Gut(2013)]{gut2013probability}
Allan Gut.
\newblock \emph{Probability: A Graduate Course}.
\newblock Springer, New York, 2nd edition, 2013.

\bibitem[Han et~al.(2018)Han, Xu, and Zhou]{han2018gaussian}
Fang Han, Sheng Xu, and Wen-Xin Zhou.
\newblock {On Gaussian comparison inequality and its application to spectral
  analysis of large random matrices}.
\newblock \emph{Bernoulli}, 24\penalty0 (3):\penalty0 1787 -- 1833, 2018.

\bibitem[Haynsworth(1968)]{haynsworth1968determination}
Emilie~V. Haynsworth.
\newblock Determination of the inertia of a partitioned {H}ermitian matrix.
\newblock \emph{Linear Algebra and its Applications}, 1\penalty0 (1):\penalty0
  73--81, 1968.

\bibitem[Hoeffding(1963)]{hoeffding1963probability}
Wassily Hoeffding.
\newblock Probability inequalities for sums of bounded random variables.
\newblock \emph{Journal of the American Statistical Association}, 58\penalty0
  (301):\penalty0 13--30, 1963.

\bibitem[Horn and Johnson(2012)]{horn2012matrix}
Roger~A. Horn and Charles~R. Johnson.
\newblock \emph{Matrix Analysis}.
\newblock Cambridge University Press, 2nd edition, 2012.

\bibitem[Hu et~al.(2019)Hu, Li, Liu, and Zhou]{hu2019high}
Jiang Hu, Weiming Li, Zhi Liu, and Wang Zhou.
\newblock High-dimensional covariance matrices in elliptical distributions with
  application to spherical test.
\newblock \emph{The Annals of Statistics}, 47\penalty0 (1):\penalty0 527--555,
  2019.

\bibitem[Hu et~al.(2026)Hu, Xie, Zhang, and Zhou]{SupplementaryMaterials}
Jiang Hu, Jiahui Xie, Yangchun Zhang, and Wang Zhou.
\newblock Supplement to “the spurious factor dilemma: robust inference in
  heavy-tailed elliptical factor models".
\newblock 2026.

\bibitem[Ke et~al.(2023)Ke, Ma, and Lin]{ke2023estimation}
Zheng~Tracy Ke, Yucong Ma, and Xihong Lin.
\newblock Estimation of the number of spiked eigenvalues in a covariance matrix
  by bulk eigenvalue matching analysis.
\newblock \emph{Journal of the American Statistical Association}, 118\penalty0
  (541):\penalty0 374--392, 2023.

\bibitem[Knowles and Yin(2013)]{knowles2013isotropic}
Antti Knowles and Jun Yin.
\newblock The isotropic semicircle law and deformation of {W}igner matrices.
\newblock \emph{Communications on Pure and Applied Mathematics}, 66\penalty0
  (11):\penalty0 1663--1749, 2013.

\bibitem[Kong(2020)]{kong2020random}
Xinbing Kong.
\newblock A random-perturbation-based rank estimator of the number of factors.
\newblock \emph{Biometrika}, 107\penalty0 (2):\penalty0 505--511, 2020.

\bibitem[Lam and Yao(2012)]{lam2012factor}
Clifford Lam and Qiwei Yao.
\newblock Factor modeling for high-dimensional time series: inference for the
  number of factors.
\newblock \emph{The Annals of Statistics}, 40\penalty0 (2):\penalty0 694--726,
  2012.

\bibitem[Li and Li(2005)]{li2005note}
Chi-Kwong Li and Ren-Cang Li.
\newblock A note on eigenvalues of perturbed {H}ermitian matrices.
\newblock \emph{Linear Algebra and its Applications}, 395:\penalty0 183--190,
  2005.

\bibitem[Lopes et~al.(2019)Lopes, Blandino, and Aue]{lopes2019bootstrapping}
Miles~E Lopes, Andrew Blandino, and Alexander Aue.
\newblock Bootstrapping spectral statistics in high dimensions.
\newblock \emph{Biometrika}, 106\penalty0 (4):\penalty0 781--801, 2019.

\bibitem[Onatski(2009)]{onatski2009testing}
Alexei Onatski.
\newblock Testing hypotheses about the number of factors in large factor
  models.
\newblock \emph{Econometrica}, 77\penalty0 (5):\penalty0 1447--1479, 2009.

\bibitem[Onatski(2010)]{onatski2010determining}
Alexei Onatski.
\newblock Determining the number of factors from empirical distribution of
  eigenvalues.
\newblock \emph{The Review of Economics and Statistics}, 92\penalty0
  (4):\penalty0 1004--1016, 2010.

\bibitem[Paul and Silverstein(2009)]{paul2009no}
Debashis Paul and Jack~W Silverstein.
\newblock No eigenvalues outside the support of the limiting empirical spectral
  distribution of a separable covariance matrix.
\newblock \emph{Journal of Multivariate Analysis}, 100\penalty0 (1):\penalty0
  37--57, 2009.

\bibitem[Resnick(1987)]{resnick1987extreme}
Sidney~I. Resnick.
\newblock \emph{Extreme Values, Regular Variation, and Point Processes}.
\newblock Springer, New York, 1987.

\bibitem[Roy et~al.(2021)Roy, Balasubramanian, and Erdogdu]{roy2021empirical}
Abhishek Roy, Krishnakumar Balasubramanian, and Murat~A Erdogdu.
\newblock On empirical risk minimization with dependent and heavy-tailed data.
\newblock \emph{Advances in Neural Information Processing Systems},
  34:\penalty0 8913--8926, 2021.

\bibitem[Stock and Watson(2002)]{stock2002forecasting}
James~H Stock and Mark~W Watson.
\newblock Forecasting using principal components from a large number of
  predictors.
\newblock \emph{Journal of the American statistical association}, 97\penalty0
  (460):\penalty0 1167--1179, 2002.

\bibitem[Trapani(2018)]{trapani2018randomized}
Lorenzo Trapani.
\newblock A randomized sequential procedure to determine the number of factors.
\newblock \emph{Journal of the American Statistical Association}, 113\penalty0
  (523):\penalty0 1341--1349, 2018.

\bibitem[Wen et~al.(2022)Wen, Xie, Yu, and Zhou]{Wen2021}
Jun Wen, Jiahui Xie, Long Yu, and Wang Zhou.
\newblock Tracy-{W}idom limit for the largest eigenvalue of high-dimensional
  covariance matrices in elliptical distributions.
\newblock \emph{Bernoulli}, 28\penalty0 (4):\penalty0 2941--2967, 2022.

\bibitem[Xie(2026)]{xie2026multiplier}
Jiahui Xie.
\newblock Multiplier bootstrap and edge phase transitions of high-dimensional
  covariance matrices.
\newblock \emph{arXiv preprint arXiv:2608.15053}, 2026.

\bibitem[Xie et~al.(2026)Xie, Yu, and Zhou]{xie2026limiting}
Jiahui Xie, Long Yu, and Wang Zhou.
\newblock Limiting laws for spiked eigenvalues and largest non-spiked
  eigenvalues of sample covariance matrices in elliptical distributions.
\newblock \emph{Bernoulli}, 32\penalty0 (3):\penalty0 2379--2402, 2026.

\bibitem[Yao and Lopes(2023)]{yao2021rates}
Junwen Yao and Miles~E Lopes.
\newblock Rates of bootstrap approximation for eigenvalues in high-dimensional
  pca.
\newblock \emph{Statistica Sinica}, 33:\penalty0 1461--1481, 2023.

\bibitem[Yu et~al.(2025)Yu, Zhao, and Zhou]{yu2024testing}
Long Yu, Peng Zhao, and Wang Zhou.
\newblock Testing the number of common factors by bootstrapped sample
  covariance matrix in high-dimensional factor models.
\newblock \emph{Journal of the American Statistical Association}, 120\penalty0
  (549):\penalty0 448--459, 2025.

\end{thebibliography}
\end{document}